\documentclass[final,12pt]{colt2022} 

\usepackage{bm}
\usepackage{dsfont}
\usepackage{enumitem}

 \setlist{nosep} 
 \setlist{noitemsep}

\newcommand{\qed}{\hfill$\blacksquare$}

\DeclareMathOperator*{\argmin}{\arg\!\min}
\DeclareMathOperator*{\argmax}{\arg\!\max}

\newcommand{\m}{\mathcal}

\newcommand{\RR}{\mathbb{R}}

\newcommand{\Var}{\mathrm{Var}}

\newcommand{\Surf}{\mathrm{Surf}}

\newcommand{\Expt}{\mathbb{E}}
\newcommand{\eps}{\varepsilon}
\newcommand{\Unif}{\mathrm{Unif}}
\newcommand{\poly}{\mathrm{poly}}

\newcommand{\polylog}{\mathrm{polylog}}

\newcommand{\Ind}{\mathds{1}}

\newcommand{\SphereD}{\m{S}^{d-1}}
\newcommand{\GMMd}{\mathrm{GMM}_{d,k}}
\newcommand{\Id}{\bm{I}}
\newcommand{\0}{\bm{0}}

\newcommand{\Prob}{\mathbb{P}}

\newcommand{\e}{\bm{e}}
\newcommand{\x}{\bm{x}}
\newcommand{\y}{\bm{y}}

\newcommand{\C}{\bm{C}}

\newcommand{\D}{\bm{D}}
\newcommand{\Dc}{\bm{\m{D}}}
\newcommand{\X}{\bm{X}}
\newcommand{\Xh}{\hat{\bm{X}}}
\newcommand{\G}{\bm{G}}
\newcommand{\Xc}{\bm{\m{X}}_k}
\newcommand{\Xch}{\hat{\bm{\m{X}}}_k}
\newcommand{\Gc}{\bm{\m{G}}_k}
\newcommand{\Z}{\bm{Z}}
\newcommand{\Y}{\bm{Y}}

\newcommand{\W}{\bm{W}}
\newcommand{\Ell}{\bm{\ell}}
\newcommand{\Loss}{\m{L}}
\newcommand{\LossAvg}{\m{L}_{\small \mbox{avg}}}
\newcommand{\LossMax}{\m{L}_{\small \mbox{max}}}
\newcommand{\SC}{n^*}
\newcommand{\AWGN}{\mathsf{AWGN}}
\newcommand{\RandLabel}{\ell}
\newcommand{\dist}{\mathrm{dist}}
\newcommand{\Rate}{\mathsf{R}_{d,k}}
\newcommand{\RateLim}{\mathsf{R}}
\newcommand{\Capacity}{\mathsf{C}}
\newcommand{\RegimeParam}{\beta}

\newcommand{\Decode}{\mathsf{Dec}}
\newcommand{\DecodeOpt}{\mathsf{DecOpt}}
\newcommand{\DecMMSE}{\mathsf{DecMMSE}}
\newcommand{\DecCORR}{\mathsf{DecCORR}}
\newcommand{\DecNN}{\mathsf{DecNN}}
\newcommand{\DecodeErrorSymb}{\#}

\newcommand{\CodeErrI}{\rho_i}
\newcommand{\CodeErrAvg}{\rho_{\small \mbox{avg}}}

\newcommand{\ExptCodeErrAvg}{\rho_{\small \mbox{avg}}}

\newcommand{\Net}{\m{T}}
\newcommand{\NetClose}{\m{T}_{\mathrm{Close}}}
\newcommand{\Test}{\mathsf{Test}}
\newcommand{\HClose}{\m{H}_{\mathrm{Close}}}
\newcommand{\HFar}{\m{H}_{\mathrm{Far}}}
\newcommand{\qClose}{\mathsf{q}_{\mathrm{Close}}}
\newcommand{\qFar}{\mathsf{q}_{\mathrm{Far}}}
\newcommand{\IndexSet}{\m{I}}
\newcommand{\epsI}{\eps_{\mathrm{I}}}
\newcommand{\XcTilde}{\tilde{\bm{\m{X}}}}
\newcommand{\XchI}{\XcTilde_{\mathrm{I}}}
\newcommand{\XTilde}{\tilde{\bm{X}}}
\newcommand{\XhI}{\XTilde}

\newcommand{\KLb}{ D_{\mathrm{KL}} }
\newcommand{\MI}{I}
\newcommand{\Ent}{H}

\newcommand{\Approx}{\mathrm{Approx}}
\newcommand{\isf}{\mathsf{i}}
\newcommand{\Avg}{\bm{A}}
\newcommand{\Proj}{\m{P}}
\newcommand{\ExptAvg}{\Expt_{\mbox{avg},\m{E}}}

\newcommand{\oneVec}{\bm{1}}

\makeatletter
\providecommand*{\cupdot}{%
  \mathbin{%
    \mathpalette\@cupdot{}%
  }%
}
\newcommand*{\@cupdot}[2]{%
  \ooalign{%
    $\m@th#1\cup$\cr
    \sbox0{$#1\cup$}%
    \dimen@=\ht0 %
    \sbox0{$\m@th#1\cdot$}%
    \advance\dimen@ by -\ht0 %
    \dimen@=.5\dimen@
    \hidewidth\raise\dimen@\box0\hidewidth
  }%
}

\providecommand*{\bigcupdot}{%
  \mathop{%
    \vphantom{\bigcup}%
    \mathpalette\@bigcupdot{}%
  }%
}
\newcommand*{\@bigcupdot}[2]{%
  \ooalign{%
    $\m@th#1\bigcup$\cr
    \sbox0{$#1\bigcup$}%
    \dimen@=\ht0 %
    \advance\dimen@ by -\dp0 %
    \sbox0{\scalebox{2}{$\m@th#1\cdot$}}%
    \advance\dimen@ by -\ht0 %
    \dimen@=.5\dimen@
    \hidewidth\raise\dimen@\box0\hidewidth
  }%
}
\makeatother

\title[]{On the Role of Channel Capacity in Learning Gaussian Mixture Models
}
\usepackage{times}

\coltauthor{%
 \Name{Elad Romanov} \Email{elad.romanov@gmail.com}\\
 \addr The Hebrew University of Jerusalem
 \AND
 \Name{Tamir Bendory} \Email{bendory@tauex.tau.ac.il}\\
 \addr Tel Aviv University
 \AND
 \Name{Or Ordentlich} \Email{or.ordentlich@mail.huji.ac.il}\\
 \addr The Hebrew University of Jerusalem
}

\begin{document}

\maketitle

\begin{abstract}%
	This paper studies the sample complexity of learning the $k$ unknown centers of a balanced Gaussian mixture model (GMM) in $\mathbb{R}^d$ with spherical covariance matrix $\sigma^2\bm{I}$. In particular, we are interested in the following question: what is the maximal noise level $\sigma^2$, for which the sample complexity is essentially the same as when estimating the centers from labeled measurements? To that end, we restrict attention to a Bayesian formulation of the problem, where the centers are uniformly distributed on the sphere $\sqrt{d}\mathcal{S}^{d-1}$. Our main results characterize the \emph{exact noise threshold} $\sigma^2$ below which the GMM learning problem, in the large system limit $d,k\to\infty$, is as easy as learning from labeled observations, and above which it is substantially harder. The threshold occurs at $\frac{\log k}{d} = \frac12\log\left( 1+\frac{1}{\sigma^2} \right)$, which is the capacity of the additive white Gaussian noise (AWGN) channel.
	Thinking of the set of $k$ centers as a code, this noise threshold can be interpreted as the largest noise level for which the error probability of the code over the AWGN channel is small. Previous works on the GMM learning problem have identified the \emph{minimum distance} between the centers as a key parameter in determining the statistical difficulty of learning the corresponding GMM. 
	While our results are only proved for GMMs whose centers are uniformly distributed over the sphere, they hint that perhaps it is the decoding error probability associated with the center constellation as a channel code that determines the statistical difficulty of learning the corresponding GMM, rather than just the minimum distance.

\end{abstract}

\section{Introduction}


Gaussian mixture models (GMMs) are 
widely used
in statistics and machine learning.
{Here, we consider the simplest case of a  \emph{spherical}, \emph{balanced} $d$-dimensional GMM with $k$-components.}
{Specifically, for centers $\Xc=(\X_1,\ldots,\X_k) \in \RR^{d\times k}$ and variance $\sigma^2$,} the corresponding GMM, denoted by $\GMMd(\Xc,\sigma^2)$, is described by the probability distribution $\Y\sim \GMMd(\Xc,\sigma^2)$:
\begin{equation}\label{eq:Intro:Y}
    \Y = \X_\RandLabel+\sigma\Z,\quad \RandLabel\sim\Unif([k]),\;\Z\sim \m{N}(\0,\Id)\,,
\end{equation}
where 
$[k]=\{1,\ldots,k\}$, and $\RandLabel\in [k]$ will sometimes be referred to as the \emph{label}
 of $\Y$ {and is statistically independent of $\Z$}. 
{Our focus is on the classical} GMM learning problem, where one observes $n$ independent samples $\Y_1,\ldots,\Y_n\sim \GMMd(\Xc,\sigma^2)$, and wishes to recover the unknown centers $\Xc$ {(throughout,
we always assume that the number of centers $k$ and the variance $\sigma^2$ are known)}. 

 This paper is devoted to studying the fundamental information-theoretic limits of the GMM learning problem, {namely, the \emph{sample complexity}: what is the smallest number of samples $n$ one needs to collect in order} to recover the centers (to within some prescribed precision)? {The main difficulty in learning the GMM centers is that the samples are unlabeled, and the sample complexity is clearly lower bounded by that of the ``genie-aided'' setup where each sample is labeled. For sufficiently small noise levels the measurements can be accurately clustered, and the problem is as easy as in the ``genie-aided'' case, while for large enough noise levels reliable clustering is impossible. The main question we seek to answer here is: \emph{what is the critical noise level below which the problem is as statistically easy as in the labeled case, and above which it is significantly harder?}}

 Past works have shown that the \emph{separation} between the centers $\X_1,\ldots,\X_k$ has a decisive effect on the statistical difficulty of the problem. Let ${\Delta(\Xc)=\min_{1\le i < j \le k}\|\X_i-\X_j\|}$ be the minimal separation between any two centers. The seminal paper \cite{regev2017learning} has accurately identified the \emph{scaling} of $\Delta(\Xc)$, in the large system limit ${k,d\to\infty}$, under which one can estimate the centers (say, to within a small constant precision) using only $n=\poly(k,d)$ many samples. They show:\footnote{We restrict our attention in this discussion, and throughout the paper, exclusively to an asymptotic regime where ${d,k\to\infty}$ together with ${\limsup_{d,k\to\infty} \frac{\log k}{d} <\infty}$.} 1) \emph{Upper bound:} If $\Delta=\Omega(\sigma\sqrt{\log k})$ then the centers may be estimated with $n=\poly(d,k)$ samples; 2) \emph{Lower bound:} For any $\gamma(k)=o\left(\sigma\sqrt{\log k}\right)$, the class of GMMs with minimum separation $\Delta\ge \sigma\gamma(k)$ is not learnable (in a minimax sense) from $n=\poly(k,d)$ samples. The upper bound was recently improved by \cite{kwon2020algorithm}, who showed that when $\Delta=\Omega(\sigma\sqrt{\log k})$, in fact $n=O(\sigma^2 k \cdot \polylog(k))$ samples suffice; this \emph{almost} matches (up to $\polylog(k)$ factors) the {sample complexity for the labeled case.}
 {Stated differently, the results above identify the critical noise level scaling for the minimax estimation problem as $\sigma^2\sim \frac{\Delta^2(\Xc)}{\log k}$.} 

 {
{The goal of this paper}
 is to develop a finer grained understanding of the \emph{exact} critical noise level $\sigma$, rather than only its scaling. 
{To}
 tackle this ambitious question, we make two modifications with respect to the setup studied in~\cite{regev2017learning} {and}~\cite{kwon2020algorithm}: 1) Rather than studying the minimax setting with respect to all sets of centers $\Xc$ with a given $\Delta(\Xc)$, we take a Bayesian approach and assume $\Xc\sim (\Unif(\sqrt{d}\SphereD))^{\otimes k}$;   2) We consider a ``more forgiving'' loss function, which measures the average error in the center reconstruction rather than the maximal error. The rationale behind these modifications will be clarified in the sequel.}

Under this setup, we show that the {the critical noise level is precisely characterized } by the equation $\frac{1}{2}\log\left(1+\frac{1}{\sigma^2}\right)=\frac{\log k}{d}$, which is, {by no accident}, the noise level below which a ``typical'' constellation $\Xc$ constitutes a good error correcting code for the $\AWGN(\sigma^2)$ channel (additive white Gaussian noise, with noise variance $\sigma^2$). 
Our analysis relies \emph{explicitly} on the decodability properties of $\Xc$, when thought of as a channel code. This is a ``global'' property of the constellation, compared to the minimum distance (note that it is well-known that at high coding rate, the minimum distance of a code is not entirely predictive of its error probability, {see e.g.~\cite{bf02}}).
{
Regarding the minimum separation, we remark that, as is to be expected, our results are consistent with \cite{regev2017learning} regarding the required scaling of $\Delta(\Xc)$ for statistically-efficient learning. Classical results on sphere packing, e.g., ~\cite{kabatiansky1978bounds}, imply that if $\log{k}/d$ is finite, ``typical'' constellations under $\Xc\sim (\Unif(\sqrt{d}\SphereD))^{\otimes k}$ have minimal separation $\Delta(\Xc)=\Theta(\sqrt{d})$. Thus, 1) When $\log{k}/d=\Theta(1)$ the critical noise level is at $\sigma^2=\Theta(1)$, so in terms of minimal separation, $\Delta(\Xc)/\sigma= \Theta(\sqrt{d})=\Theta(\sqrt{\log k})$; 2) On the other hand, when $\log{k}/d=o(1)$, the critical noise level is $\sigma^2=\Theta(d/\log{k})$ and so $\Delta(\Xc)/\sigma=\Theta(\sqrt{\log k})$.

Finally, our results \emph{hint} at the possibility of a deeper connection between channel coding and statistical inference: the decodability properties of the set of centers $\Xc$ (as a channel code) may determine, to an extent, the statistical difficulty of learning the corresponding GMM. The present paper takes a modest first step towards
showing such a connection, establishing it for the special case of spherical random codes, whose typical instances posses strong symmetry properties. 
}

\subsection{Formal Problem Formulation}
\label{sec:Model}

As mentioned before, we study the large system behavior of the sample complexity under a uniform spherical prior on the centers. Denote the (random) centers by
\begin{equation}\label{eq:Formulation:XcPrior}
    \Xc = (\X_1,\ldots,\X_k) \sim \left(\Unif(\sqrt{d}\SphereD)\right)^{\otimes k} \,.
\end{equation}
Note that we scale the problem so that $\|\X_i\|=\sqrt{d}$ for all $i\in[k]$. We observe $n$ measurements, $\Y_1,\ldots,\Y_n$, sampled from the GMM distribution whose centers are $\Xc$:
\begin{equation}\label{eq:Formulation:YGivenXc}
    \left[ \Y_1,\ldots,\Y_n \,\Big|\,\Xc\right] \;\overset{i.i.d.}{\sim}\; \GMMd(\Xc,\sigma^2) \,,
\end{equation}
see also (\ref{eq:Intro:Y}).
Per standard terminology in signal processing, $1/\sigma^2$ may be interpreted as the ``signal-to-noise ratio'' (SNR) per coordinate. Suppose that $\Xch=(\Xh_1,\ldots,\Xh_k)$ is an estimator of $\Xc$, computed from the measurements. The model admits the following Markov chain structure:
\begin{equation}
    \label{eq:Formulation:MarkovChain}
    \Xc=(\X_1,\ldots,\X_k)\longrightarrow (\Y_1,\ldots,\Y_n)\longrightarrow \Xch=(\Xh_1,\ldots,\Xh_k) \,.
\end{equation}
{At this point it is instructive to think about the much simpler estimation problem, where each 
{measurement $\Y_i$ is observed with its label $\RandLabel_i\in [k]$,}
and 
{every
center is observed \emph{exactly} $n/k$ times}. For this problem, the optimal mean squared error (MSE) in the reconstruction of each center is { $d^{-1}\Expt\|\X_i-\Xh_i\|^2=k\sigma^2/n$} (to leading order in $k/n$), and is attained for example, by the sample mean. In the GMM estimation problem the samples are not labeled, and furthermore, the number of times each center appears in the measurements is a  $\mathrm{Binomial}(n,1/k)$ random variable. While the mean of this random variable is indeed $n/k$, some centers will appear fewer times. In particular, when $n=o(k\log{k})$ some of the centers are likely to not appear even once (coupon collecting). 
{To} 
circumvent the issues arising due to this effect, and focus our study on the problem of dealing with the lack of labels, we} measure the discrepancy between $\Xc$ and $\Xch$, by the loss function
\begin{equation}\label{eq:LossAvg-def}
    \LossAvg(\Xc,\hat{\bm{\m{X}}}_k ) = \frac1k \sum_{i=1}^k d^{-1}\dist^2(\X_i,\hat{\bm{\m{X}}}_k) := \frac1k \sum_{i=1}^k \min_{1\le j \le k} d^{-1}\|\X_i-\Xh_j\|^2\,.
\end{equation}
In words: the average normalized squared distance between a center $\X_i$ and the list $\Xch$. {As we shall see, under this loss function it is possible to obtain a risk of $k\sigma^2/n$ for $\sigma$ below the critical noise level and $n/k$ large enough. 
In contrast, the more restrictive max-loss function $\LossMax(\Xc,\Xch)= \max_{1\le i \le k} d^{-1}\dist^2(\X_i,\Xch)$ considered in much of the prior work, does not decay with $n$ in the regime $n=o(k\log k)$, regardless of the noise level, due to the non-uniform empirical distribution of the center indices. Under $\LossAvg$, on the other hand, 
{to}
achieve $\eps$ error, it suffices to estimate only a fraction $1-O(\eps)$ of the centers within error $O(\eps)$, 
{having} the remaining centers incur an error {$O(1)$}. Thus, the effect of non-uniform label empirical distribution is bypassed by this loss function.}

Under our formulation of the GMM learning problem, the goal is to construct an estimation rule $\Xch\,:\,(\RR^d)^n\to \RR^{d\times k}$ (``algorithm'') so to minimize the risk: $\Expt \LossAvg(\Xc,\Xch(\Y_1,\ldots,\Y_n))$. 
Importantly, the expectation is taken over the randomness in \emph{both} the  sample generating process given the centers (\ref{eq:Formulation:YGivenXc}), \emph{as well as} the center prior distribution (\ref{eq:Formulation:XcPrior}), whose joint distribution adheres to the Markov chain structure in (\ref{eq:Formulation:MarkovChain}). We study the \emph{information-theoretic limits} of the aforementioned problem.
Consider the minimum attainable risk {over all estimation laws $\Xch$}:
\begin{equation}\label{eq:Formulation:MinRisk}
    R_n = \inf_{\Xch} \Expt\LossAvg(\Xc,\Xch(\Y_1,\ldots,\Y_n))\,.
\end{equation}
{
For a fixed precision level $\eps>0$, define the \emph{sample complexity},
}
\begin{equation}
    \label{eq:Formulation:SC}
    \SC_\eps=\SC_\eps(d,k,\sigma^2) = \min\left\{ n\,:\,R_n \le \eps \right\} \,.
\end{equation}  
{
Importantly, 
(\ref{eq:Formulation:MinRisk}) and (\ref{eq:Formulation:SC}) 
{make no assumptions about the computational difficulty}
of implementing $\Xch:(\RR^d)^n\to\RR^{d\times k}$,
{and in particular are not restricted to computational efficient algorithms ($\poly(d,k)$ runtime).}
Throughout, computational considerations shall be completely neglected.
}

\subsection{Main Results}\label{sec:MainResults}



{As our analysis relies on viewing the 
centers as a code
}
for the AWGN channel, the problem's  \emph{rate}
\begin{equation}
    \label{eq:Rate-def}
    \Rate := \frac{\log k}{d} \,,
\end{equation}
and the decreasing function $\Capacity:(0,\infty)\to(0,\infty)$
\begin{equation}
    \label{eq:CapacityFunc-def}
    \Capacity(\sigma^2) = \frac12 \log\left( 1+\frac{1}{\sigma^2} \right) \,,
\end{equation}
characterizing the $\AWGN(\sigma^2)$ channel capacity, will play a key role. Throughout the paper, we couple the noise level $\sigma^2$ to $k$ and $d$ by the parameter $\RegimeParam\in (0,\infty)$ via the equation
\begin{equation}\label{eq:Rate-Equals-Cap-Beta}
    \Rate = \Capacity(\RegimeParam\sigma^2) \,.
\end{equation}
When $\RegimeParam>1$, the rate is smaller than the capacity; when $\RegimeParam<1$, it is larger. This parametrization will turn out particularly useful in the statement of the results and their derivations.

We restrict attention to the large-system limit, where $d,k\to\infty$, and denote the \emph{limiting rate} by
\begin{equation}
    \label{RateLim-def}
    \RateLim = \lim_{d\to\infty} \Rate \in [0,\infty)\,.
\end{equation}
{We distinguish between two asymptotic regimes:}
\begin{itemize}[align=left]
    \item \textbf{(Positive Rate, $\RateLim>0$)}: $\sigma^2\in (0,\infty)$ is a fixed constant. In particular, $k=e^{\Theta(d)}$.
    \item \textbf{(Zero Rate, $\RateLim=0$)}: $\sigma^2\to \infty$. So that also $k\to\infty$, we also require impose $\sigma^2=o(d)$. 
\end{itemize}
{
We remark that, since we are interested in estimation to \emph{finite precision} $\eps$ in Theorems~\ref{thm:BelowCapacity} and~\ref{thm:AboveCapacity} below, the asymptotic regime $\log k =\omega(d)$, namely when the number of centers $k$ is super-exponential in $d$, becomes rather uninteresting. Indeed, for super-exponential $k$ one may simply take $\Xch$ to be some fixed $\sqrt{\eps d}$-net of the sphere $\sqrt{d}\SphereD$, which can be of size $\left(O(1/\eps)\right)^{d/2} \ll k$. Clearly, $\LossAvg(\Xc,\Xch)\le \eps$ for any $\Xc$, so under such asymptotics $\SC_\eps = 0$ exactly. 
}

Our first main result states that when the rate is below the channel capacity, $\Xc$ is learnable at essentially the same sample complexity as in the labeled case.

\begin{theorem}\label{thm:BelowCapacity}
    Suppose that $\RegimeParam>1$.
    Then
    {
    \begin{equation}\label{eq:Thm1}
        e^{-2\RateLim} \le \lim_{\eps\to0} \lim_{d\to\infty} \frac{\SC_\eps}{\left(\sigma^2 k/\eps\right)} \le 1 \,.
    \end{equation}
    }
\end{theorem}
{
Theorem~\ref{thm:BelowCapacity} implies that when the rate is below the channel capacity, for every \emph{fixed} and \emph{small} precision $\eps>0$, and for $d$ large, the sample complexity scales like $n=C\sigma^2 k /\eps$, where $C\in [e^{-2\RateLim},1]$. Remarkably, when $k$ is sub-exponential in $d$ ($\RateLim=0$) the pre-factor $C$ is precisely $1$. Thus, the sample complexity of the GMM learning problem is \emph{exactly} the same as that of the labeled case, up to lower order terms in $1/\eps$, and asymptotically ($d\to\infty$) vanishing correction terms. 
}



Our second main result states that above the capacity, the sample complexity is super-linear:

\begin{theorem}
    \label{thm:AboveCapacity}
    Suppose that $\RegimeParam<1$. Then for any fixed sufficiently small {$\eps<\eps_0(\RateLim)$},
    \begin{equation}\label{eq:Thm2:1}
        \lim_{d\to\infty} \frac{\SC_\eps}{\sigma^2k/\eps} = \infty \,.
    \end{equation}
    Moreover, the following quantitative bounds hold for all sufficiently small $\eps<\eps_0(\RateLim)$:
    \begin{enumerate}
        \item If $\RateLim>0$ then 
        \begin{equation}
            \frac{\SC}{\sigma^2 k} = \Omega_{\eps,\RegimeParam,\RateLim}\left(\sqrt{\frac{\log k}{\log \log k}} \right) \,.
        \end{equation}
        \item If $\RateLim=0$ then
        \begin{equation}
            \frac{\SC}{\sigma^2 k} = \Omega_{\eps,\RegimeParam}\left(\min\left\{ \sqrt{\frac{\log k}{\log\log k}}, \sqrt{\frac{d}{\log k}}   \right\}\right) \,.
        \end{equation}
    \end{enumerate}
\end{theorem}

Theorems~\ref{thm:BelowCapacity} and~\ref{thm:AboveCapacity} together reveal a dichotomy: \emph{precisely} at the channel capacity ($\RegimeParam=1$), the large-system behavior of the sample complexity undergoes a \emph{phase-transition}, from a linear growth in $\sigma^2k$, as in the labeled case, to super-linear growth.

\subsection{Prior Art}
\label{sec:PriorArt}

The problem of estimating the parameters of a Gaussian mixture model has a long and rich history, going back to the pioneering work of \cite{pearson1894contributions}. We briefly mention some pointers to the literature, though we emphasize that the list below is \emph{not exhaustive by any means}.

The first work to highlight the importance of minimum separation in learning GMMs is \cite{dasgupta1999learning}, who gave a poly-time algorithm assuming (in the spherical balanced case) $\Delta=\Omega(\sigma\sqrt{d})$. Subsequent works have gradually improved upon the required bound on $\Delta$. Early incarnations include \cite{sanjeev2001learning,vempala2004spectral,achlioptas2005spectral,dasgupta2007probabilistic,kannan2008spectral}, which culminated in a bound $\Delta=\Omega(\sigma k^{1/4})$ {as sufficient for estimation in polynomial time}. This barrier was broken only fairly recently 
\cite{diakonikolas2018list,hopkins2018mixture,kothari2018robust},who have shown that separation $\Delta=\Omega(\sigma k^\gamma)$ suffices for polynomial-time learnability, for \emph{any} constant $\gamma>0$. 

As for {statistical} lower bounds, it is known that in the absence of a separation condition, $n=\exp(k)$ samples are generally necessary to learn the parameters of a GMM \cite{moitra2010settling,hardt2015tight}. The work \cite{regev2017learning} has shown that separation $\Delta=\Omega(\sigma \sqrt{\log k})$ is a \emph{sufficient and necessary} condition for learning GMMs with $n=\poly(k,d)$ samples; the algorithm they proposed to prove their upper bound has exponential runtime. \cite{kwon2020algorithm} have recently improved their {upper} bound on the sample complexity, and have show that in fact $n=O(\sigma^2 k \cdot \polylog(k))$ samples suffice, which \emph{almost} matches the trivial lower bound of $n=\Omega(\sigma^2 k)$. 
{Their analysis consists of two components} : 1) An exponential-time initialization scheme, that finds points \emph{sufficiently close} to the true centers, based on the results of \cite{ashtiani2018nearly}; 2) New local convergence and finite-sample guarantees for (a slightly modified version of) the well-known Expectation Maximization (EM) algorithm. 
To our knowledge, the problem of learning $\Delta=\Omega(\sigma\sqrt{\log k})$-separated GMMs in polynomial time, or proving that this cannot be done 
{(the existence of a computational-statistical gap)}
is still open. 

Another line of work circumvents the minimal separation requirement, by instead 
restricting attention to ``typical'' problem instances, an approach much in line with the results of the present paper, and in the context of learning GMMs dates, to the best of our knowledge, to the study \cite{srebro2006investigation}. In the papers \cite{hsu2013learning,bhaskara2014uniqueness,goyal2014fourier,anderson2014more,anandkumar2014tensor,ge2015learning},
it is shown that when the center configuration satisfies certain algebraic \emph{non-degeneracy} conditions, methods based on tensor decomposition may be used to recover the centers; such  non-degenerate configurations are highly abundant when $d$ is large relative to $k$, specifically $k \le d^{O(1)}$.

Lastly, a different line of work considers learning GMMs by means of density estimation, that is, given samples $\Y_1,\ldots,\Y_n$ one has to construct a density $f$ which is close to $\GMMd(\Xc,\sigma^2)$ in, e.g., total variation distance. This problem may be considered in either in the setting of proper learning ($f$ has to be a $k$-component GMM) or improper learning (no such restriction), see for example \cite{feldman2006pac,kalai2010efficiently,chan2014efficient,acharya2014near,li2017robust,diakonikolas2019robust,ashtiani2018nearly}. 
{For}
well-seaprated spherical GMMs, $\Delta=\Omega(\sigma\sqrt{\log k})$, guarantees for proper distribution estimation may be translated to 
{error bounds} 
on the centers, see \cite{regev2017learning,kwon2020algorithm}.

{Our proof program closely follows that of~\cite{romanov2021multi}, which studied the sample complexity of the multi-reference alignment (MRA) problem. 
{MRA is a \emph{particular} instance of a GMM, with exactly $k=d$ components corresponding to different shifted versions of the same signal.}
While, similarly to~\cite{romanov2021multi}, the proof of our lower bound uses the mutual information method~\cite{polyanskiy2014lecture}, here the mutual information is upper bounded using the I-MMSE relation rather than the Fano-based argument of~\cite{romanov2021multi}. More importantly, the proof of the upper bound here requires overcoming several significant hurdles not present in the MRA model. In particular, while in MRA we always have $k=d$, in the GMM problem $k$ may be much greater, and even exponential in the dimension. Furthermore, in MRA there is a single signal to be estimated and all measurement are informative for its estimation. Here, on the other hand, many centers must be estimated, which significantly complicates the first step of our reconstruction algorithm with respect to that used in~\cite{romanov2021multi}.}

\paragraph{Paper outline.}
{
In Section~\ref{sec:Background} we provide brief background on channel coding and random spherical codes, which shall be used in the analysis to follow. In Section~\ref{sec:LowerBound} we outline the proof of the lower bound in Theorems~\ref{thm:BelowCapacity} and~\ref{thm:AboveCapacity}. In Section~\ref{sec:UpperBound} we outline the proof of the upper bound in Theorem~\ref{thm:BelowCapacity}. To keep within the space constraint, most of the technical details are deferred to the Appendix.
}


\section{Background on Channel Coding}
\label{sec:Background}

A key message of this paper is the following: the centers $\Xc$ are learnable at linear sample complexity exactly in the regime where the constellation $\Xc=(\X_1,\ldots,\X_k)$ defines (with high probability) a good codebook for the AWGN channel with noise variance $\sigma^2$. Throughout the analysis, the connection to the decoding capabilities of $\Xc$ will be instrumental. In this section, we briefly survey the required background from information and coding theory. We refer the reader to \cite{cover2012elements},~\cite{gallager1968information} and \cite{polyanskiy2014lecture} for a comprehensive treatment.

A \emph{coding scheme} for sending $\log{k}$ nats over the $d$-dimensional AWGN channel consists of a codebook and a decoder. The codebook is a set of $k$ codewords $\bm{\m{C}}=(\C_1,\ldots,\C_k)\in \RR^{d\times k}$, where codeword $\C_i$ encodes message $1\le i \le k$, and all codewords satisfy $\|\C_i\|^2\le d$. The code's rate is $\Rate=\frac{\log k}{d}$. The decoder $\Decode:\RR^d\to [k]$ is a mapping from channel outputs to messages. It is often convenient to allow the decoder to output symbols in $[k]\cup \{\DecodeErrorSymb\}$, where the special symbol $\DecodeErrorSymb$ corresponds to a \emph{declared decoding error}.

The \emph{decoding error} associated with message $i\in[k]$, for a given a codebook-decoder pair, is
\begin{equation}\label{eq:Error-i-def}
    P_{e,i}(\sigma^2|\bm{\m{C}},\Decode) = \Pr\left(i\ne \Decode(\X_i+\sigma\Z)\right) \,,
\end{equation}
and the \emph{average error} over all messages is
\begin{equation}\label{eq:Error-avg-def}
    P_{e,avg}(\sigma^2|\bm{\m{C}},\Decode) := \frac1k \sum_{i=1}^k P_{e,i}(\sigma^2|\bm{\m{C}},\Decode) = 
    \Pr_{\RandLabel\sim\Unif([k])}\left(\RandLabel\ne \Decode(\X_\RandLabel+\sigma\Z)\right) \,.
\end{equation}
For a given codebook $\bm{\m{C}}$, the optimal decoder, {in the sense of smallest average error,} is clearly given by the maximum a posteriori probability (MAP) rule
\begin{equation}\label{eq:Background:DecOpt}
    \DecodeOpt(\Y)=\argmax_{i\in[k]} \Pr(\ell=i\,|\,\Y,\bm{\m{C}})=\argmin_{1\le i\le k}\|\Y-\C_i\|^2\,,
\end{equation} 
where ties are broken arbitrarily. Accordingly, we define the {error} of the \emph{codebook} $\bm{\m{C}}$, and the corresponding individual errors as 
\begin{equation}
    \label{eq:Background:ErrorOPT}
    \begin{split}
        \CodeErrAvg(\sigma^2|\bm{\m{C}})&=P_{e,avg}(\sigma^2|\bm{\m{C}},\DecodeOpt)\,,\quad
        \CodeErrI(\sigma^2|\bm{\m{C}})=P_{e,i}(\sigma^2|\bm{\m{C}},\DecodeOpt)\,.
    \end{split}
\end{equation}

{In communication theory, one is interested in designing coding schemes with large rate and small error probability. We say a rate $\RateLim\in (0,\infty)$ is \emph{achievable} if there exists a sequence ($d\to\infty$) of codebooks $\bm{\m{C}}\in \RR^{d\times k}$ such that $\lim_{d\to\infty}\Rate= \RateLim$ and $\lim_{d\to\infty}\rho(\sigma^2|\bm{\m{C}})= 0$.
}
Shannon's celebrated channel coding theorem 
gives a precise characterization of all the achievable rates:
\begin{theorem}
    [Channel coding theorem, AWGN channel]
    \label{thm:ChannelCoding}
    Fix $\sigma^2$, and let $\Capacity(\cdot)$ be given in (\ref{eq:CapacityFunc-def}).
    \begin{enumerate}
        [topsep=3pt,itemsep=-1ex,partopsep=1ex,parsep=1ex]
        \item (Achievability). Any rate $\RateLim<\Capacity(\sigma^2)$ is achievable.
        \item (Converse). No rate $\RateLim>\Capacity(\sigma^2)$ is achievable.
    \end{enumerate}
\end{theorem}



%


The achievability part of the channel coding theorem is typically proved using a random coding argument with respect to the ensemble of i.i.d. Gaussian codebooks. However, it can also be proved using the ensemble of 
{
random spherical codebooks, $\bm{\m{C}} = \Xc = (\X_1,\ldots,\X_k)\sim \Unif(\sqrt{d}\SphereD)^{\otimes k}$.
}
In fact, the latter ensemble results in a favorable decay of the error probability with $d$,~\cite{shannon1959probability}. 
{We denote the decoding error, averaged over the codebook ensemble, by}
\begin{equation}\label{eq:Expt-Error-Xc-def}
    \begin{split}
        \ExptCodeErrAvg(\sigma^2) = \Expt[\CodeErrAvg(\sigma^2|\Xc)] \overset{(\star)}{=} \Expt[\CodeErrI(\sigma^2|\Xc)]\,,
    \end{split}
\end{equation}
where $(\star)$ holds since each $\X_i$ has the same distribution.
\begin{proposition}\label{prop:DecodingSphericalCodes}
    Let $\RegimeParam>1$ be fixed. Suppose that $d,k\to\infty$, with $\Rate=\Capacity(\beta\sigma^2)$, so that either: 1) $\sigma^2$ fixed;or 2) ${\omega(1)=\sigma^2=o(d)}$. Then $\lim_{d\to\infty} \ExptCodeErrAvg(\sigma^2) = 0$.
    


        

        

\end{proposition}
While Proposition~\ref{prop:DecodingSphericalCodes} is well-known 
{when}
$\sigma^2$ is fixed 
(positive rate)~
\cite{shannon1959probability}, the case of ${\omega(1)=\sigma^2=o(d)}$ has not been mainstreamed. We provide a self-contained proof of Proposition~\ref{prop:DecodingSphericalCodes} in Appendix, Section~\ref{sec:Appendix:Decoding}, since it will serve as the baseline for the derivations that follow.

\section{Proof of Lower Bounds}
\label{sec:LowerBound}

Our proof of the lower bounds in Theorems~\ref{thm:BelowCapacity} and \ref{thm:AboveCapacity} uses a standard framework for proving estimation lower bounds (e.g., \cite[Chapter 28]{polyanskiy2014lecture}).

Suppose $\Xch=\Xch(\Y_1,\ldots,\Y_n)$ attains $\Expt\LossAvg(\Xc,\Xch)\le \eps$. Consider the Markov chain (\ref{eq:Formulation:MarkovChain}). By the data processing inequality (DPI) \cite[Theorem 2.5]{polyanskiy2014lecture},
\begin{equation}\label{eq:LowerBound:DPI}
    \MI(\Xc;\Xch)\le \MI(\Xc;\Y_1,\ldots,\Y_n) \,.
\end{equation}
We lower bound the LHS of (\ref{eq:LowerBound:DPI}) in terms of $\eps$ and upper bound the RHS in terms of $n$ and $\sigma^2$.
Starting with $\MI(\Xc;\Xch)$, clearly,
\begin{equation}\label{eq:LowerBound:RDF}
    \MI(\Xc;\Xch)\ge \min_{P_{\Dc|\Xc}\,:\,\Expt\LossAvg(\Xc,\Dc)\le \eps} \MI(\Xc;\Dc) \,,
\end{equation}
where we minimize the mutual information (MI) over all conditional distributions of random variables $\Dc=(\D_1,\ldots,\D_k)\in\RR^{d\times k}$, 
{under the
expected loss constraint $\Expt\LossAvg(\Xc,\Dc)\le \eps$. 
}
The 
{minimization}
(\ref{eq:LowerBound:RDF}) is an instance of a \emph{rate-distortion} problem, 
{
typically encountered 
when studying the information-theoretic limits of lossy compression
\cite[Chapter 25]{polyanskiy2014lecture}. 
}

One complication that arises when attempting to solve the optimization problem in (\ref{eq:LowerBound:RDF}) is that the distortion measure, ${\LossAvg(\Xc,\Dc)=\frac{1}{dk}\sum_{i=1}^k \min_{1\le j\le k}\|\X_i-\D_j\|^2}$ is somewhat non-standard. If instead we had used the quadratic loss, ${\frac{1}{dk}\|\Xc-\Dc\|^2_F=\frac{1}{dk}\sum_{i=1}^k \|\X_i-\D_i\|^2}$, the resulting optimization problem
would essentially lend itself to the classical problem of computing the Gaussian quadratic rate-distortion function (RDF), which admits the solution 
${\frac{dk}{2}\log(1/\eps)}$. 
 
 The loss $\LossAvg$ differs from the standard quadratic loss in that it allows for $k$ additional degrees of freedom: every $i\in [k]$ is matched to the best index $j_i = \argmin_{j\in [k]}\|\X_i-\D_{j_i}\|$. Since the entropy of the $k$-tuple $(j_1,\ldots,j_k)$ is at most $k\log{k}$ nats, the RDF for $\LossAvg$ must be at most $k\log{k}$ nats away from the RDF for the standard quadratic loss. 
 {We prove 
 in  Appendix, Section~\ref{sec:proof-lem:LowerBound:RDF-lb}:}
 \begin{lemma}\label{lem:LowerBound:RDF-lb}
     Consider the Markov chain (\ref{eq:Formulation:MarkovChain}), with $\Expt\LossAvg(\Xc,\Xch)\le\eps$. For universal $c_0>0$,
     \begin{align*}
            \MI(\Xc;\Xch) 
            &\ge \frac{dk}{2}\log(1/\eps) - dk\log\left(1+c_0(\eps d)^{-1/2}\right) -k\log k \,. 
     \end{align*}
  \end{lemma}

    {
  Next, we upper bound $\MI(\Xc;\Y_1,\ldots,\Y_n)$, starting with a trivial bound. Let $\Ell=(\RandLabel_1,\ldots,\RandLabel_n)$ be the random labels, such that $\Y_j=\X_{\RandLabel_j}+\sigma\Z_j$. By the DPI,  $\MI(\Xc;\Y_1,\ldots,\Y_n)\le \MI(\Xc;\Y_1,\ldots,\Y_n,\Ell)$. 
  Now, given $\Ell$, 
  the mapping $\Xc \mapsto (\Y_1,\ldots,\Y_n)$ simply corresponds to $k$ parallel Gaussian channels, each used on average $n/k$ times. Thus, as we formally prove in 
  Appendix, Section~\ref{sec:proof-lem:LowerBound:Trivial},
  }
  \begin{lemma}
      \label{lem:LowerBound:Trivial}
      The following holds:
      \begin{align}\label{eq-from:lem:LowerBound:Trivial}
          \MI(\Xc;\Y_1,\ldots,\Y_n)\le \MI(\Xc;\Y_1,\ldots,\Y_n,\Ell) \le \frac{dk}{2}\log\left(1+\frac{n}{k\sigma^2}\right) \,.
      \end{align}
      Consequently, combining with (\ref{eq:LowerBound:DPI}) and Lemma~\ref{lem:LowerBound:RDF-lb},
      \begin{align*}
        \lim_{d\to\infty} \frac{\SC_\eps}{k\sigma^2} \ge e^{-2\RateLim}\eps^{-1} - 1 \,.
    \end{align*}
  \end{lemma}
{
The bound (\ref{eq-from:lem:LowerBound:Trivial}) misses a \emph{crucial} aspect of our problem: 
the observations are \emph{not} labeled.
We next derive a bound which does capture this effect, though at the loss of the ``correct'' dependence on $n$.}
%

  Observe that $\Y_1,\ldots,\Y_n$ are conditionally independent given $\Xc$. That is: the ``channel'' mapping the set of centers to samples is \emph{memoryless}. 
  It is an elementary fact  \cite[Theorem 5.1]{polyanskiy2014lecture} that in this case, the MI is subadditive
  \begin{equation}\label{eq:LowerBound:Memoryless}
      \MI(\Xc;\Y_1,\ldots,\Y_n) \le \sum_{i=1}^n \MI(\Xc;\Y_i) = n\cdot \MI(\Xc;\Y)\,.
  \end{equation}
{While this bound fails to correctly capture the dependence of $\MI(\Xc;\Y_1,\ldots,\Y_n)$ on $n$ when 
{$n/(k\sigma^2)$}
is large, it does suffice for establishing the phase transition of the sample complexity that we seek here.}
We proceed to bounding the single-sample MI, $\MI(\Xc;\Y)$, a much more manageable object.
Let ${\RandLabel\sim \Unif([k])}$ be the random label of $\Y$. Using the MI chain rule both ways,
\begin{align*}
    \MI(\Xc,\RandLabel;\Y) 
    &= \MI(\Xc;\Y) + \MI(\RandLabel;\Y|\Xc)
    = \MI(\RandLabel;\Y) + \MI(\Xc;\Y|\RandLabel)\,.
\end{align*}
Now, $\MI(\RandLabel,\Y)=0$ (since $\{\X_i\}_{i=1}^k$ are identically distributed, so $\Y$ does not depend on $\RandLabel$). Similarly, $\MI(\RandLabel;\Y|\Xc)=\Ent(\RandLabel|\Xc)-\Ent(\RandLabel|\Xc,\Y)$, and $\Ent(\RandLabel|\Xc)=\Ent(\RandLabel)=\log k$
Furthermore, ${\MI(\Xc;\Y|\RandLabel)=\MI(\X_\ell;\Y|\RandLabel)\le \Capacity(\sigma^2)d}$, as the AWGN channel capacity $\Capacity(\sigma^2)$ upper bounds $I(\X;\X+\sigma\Z)/d$ for any random variable on $\RR^d$ with second moment $\mathbb{E}\|\X\|^2\leq d$.
Combining these equalities and estimates and rearranging, we obtain
\begin{equation}\label{eq:LowerBound:SingleSample}
    \MI(\Xc;\Y) \le \Capacity(\sigma^2)d - \log k + \Ent(\RandLabel|\Xc,\Y)\,.
\end{equation}

In light of (\ref{eq:LowerBound:SingleSample}), it remains to estimate $\Ent(\RandLabel|\Xc,\Y)$, to be interpreted as the remaining uncertainty in a message $\ell$ that is sent across the channel, given the output $\Y$ as well as the known codebook $\Xc$. To that end, consider the non-increasing mapping $\RegimeParam\mapsto \varphi(\beta)= \Ent(\RandLabel|\Xc,\Y)$ (recall that larger $\RegimeParam$ corresponds to smaller $\sigma$). 
Since a typical realization of $\Xc$ results in a code  whose error vanishes when $\RegimeParam>1$, Fano's inequality implies that $\varphi(\beta)\big|_{\beta>1}=o(\log(k))$. Thus, for $\beta<1$, we have that $\varphi(\beta)=-\int_{\beta}^{1+\delta} \varphi'(s)ds+o(\log(k))$, for any $\delta>0$. Using the I-MMSE formula \cite{guo2005mutual}, a remarkable connection between information and estimation under Gaussian channels, the derivative $\varphi'(\beta)$ can be expressed as the minimum MSE (MMSE) in estimating $\X_{\ell}$ from $\Y$. Finally, we upper bound the MMSE by the optimal MSE for linear estimation, resulting in the following lemma, whose full proof appears in Appendix, Section~\ref{sec:proof-lem:LowerBound:SingleSample-IMMSE}. We denote by $\Capacity^{-1}$ the inverse of (\ref{eq:CapacityFunc-def}), and by $h_b(p)=p\log\frac1p + (1-p)\log\frac{1}{1-p}$ the binary entropy function.


\begin{lemma}\label{lem:LowerBound:SingleSample-IMMSE}
    Suppose that $\RegimeParam<1$, so that $\Rate=\Capacity(\RegimeParam\sigma^2)>\Capacity(\sigma^2)$. For $\delta> 0$, denote the corresponding noise level by ${\sigma_0^2(\delta)=\Capacity^{-1}\left((1+\delta)\Rate\right)}$ and let $e(\delta) = \ExptCodeErrAvg\left( \Capacity^{-1}\left((1+\delta)\Rate\right) \right)$
    be the ensemble average decoding error, (\ref{eq:Expt-Error-Xc-def}), over the $\AWGN(\sigma_0^2)$ channel.
    We have that
    \begin{align*}
        \Ent(\RandLabel|\Xc,\Y)\le \log k - \Capacity(\sigma^2)d +  h_b(e(\delta)) + \left( \delta + e(\delta) \right)\log k \,,
    \end{align*}
    and consequently, using (\ref{eq:LowerBound:SingleSample}),
    \begin{equation}\label{eq:LowerBound:SingleSample-Delta}
        \MI(\Xc;\Y) \le h_b(e(\delta)) + \left( \delta + e(\delta) \right)\log k  \,.
    \end{equation}
\end{lemma}

As mentioned above, when $\RegimeParam<1$, $e(\delta)=o(1)$ for \emph{all} fixed $\delta>0$, by Proposition~\ref{prop:DecodingSphericalCodes}; consequently, $\MI(\Xc;\Y)=o(\log k)$. 
{
Combining this with (\ref{eq:LowerBound:DPI}), Lemma~\ref{lem:LowerBound:RDF-lb} and (\ref{eq:LowerBound:Memoryless}),
}
assuming sufficiently small $\eps=O_{\RateLim}(1)$, we establish (\ref{eq:Thm2:1}):
\begin{align}\label{eq:LowerBound:SC-final}
    \frac{\SC_\eps}{\sigma^2 k} \ge C(\eps,\RateLim) \cdot \frac{d}{\sigma^2} \cdot ({\MI(\Xc;\Y)})^{-1}
    = \omega\left( \frac{1}{\sigma^2} \cdot \frac{d}{\log k} \right) = \omega\left( \frac{1}{\sigma^2\Capacity(\beta\sigma^2)} \right) = \omega(1) \,.
\end{align}
One can get \emph{quantitative} bounds by carefully setting $\delta=o(1)$,
{as we do in Appendix, Section~\ref{sec:proof-lem:LowerBound:Quantitative}:}

\begin{lemma}\label{lem:LowerBound:Quantitative}
    Suppose that $\RegimeParam<1$. For small enough fixed ${\eps\le\eps_0(\RateLim)}$:
    \begin{enumerate}
        \item (Positive rate). If $\RateLim>0$ then
        \begin{equation}
            \label{eq:lem8:pos}
            \frac{\SC_\eps}{\sigma^2 k } \ge C(\eps,\RegimeParam,\RateLim) \sqrt{\frac{\log k}{\log\log k}} \,.     
        \end{equation}
        \item (Zero rate). If $\RateLim=0$ then
        \begin{equation}
            \label{eq:lem8:0}
            \frac{\SC_\eps}{\sigma^2 k} \ge C(\eps,\RegimeParam)\min\left\{ \sqrt{\frac{\log k}{\log \log k}}, \sqrt{\frac{d}{\log k}} \right\} \,.
        \end{equation}
    \end{enumerate}
\end{lemma}

\begin{proof}
    (Of Theorem~\ref{thm:AboveCapacity}). Directly follows from Lemma~\ref{lem:LowerBound:Quantitative}.
\end{proof}

\section{Proof of Upper Bound}
\label{sec:UpperBound}

{

In this section we prove the upper bound of Theorem~\ref{thm:BelowCapacity}, assuming the rate is smaller than the capacity ($\RegimeParam>1$). The proof is constructive: we propose and analyze an algorithm (which runs in exponential time), whose output $\Xch$ satisfies 
${\Expt \LossAvg(\Xc,\Xch)\le \eps}$. It consists of two steps, each using \emph{different} measurements:
Step I is allocated $N$ samples, while Step II uses the remaining $\bar{N}=n-N$ samples.


Step I consists of a \emph{brute-force search} over an exponential-sized set of candidate centers. 
Let $\epsI>0$ be a given precision level, and fix $\Net$ a $\sqrt{\epsI d/2}$-net of the sphere $\sqrt{d}\SphereD$. For \emph{each} candidate $\Xh\in\Net$, we use the measurements $\Y_1,\ldots,\Y_N$ allocated for this step to essentially solve a composite hypothesis testing problem, distinguishing between two alternatives: 1) $\Xh$ is $\sqrt{\eps d/2}$-close to \emph{some} center $\X_i$; 2) $\Xh$ is $\sqrt{\eps d}$-far from \emph{all} the centers. We show that for ``typical'' center configurations $\Xc$, if $N\gtrsim \sigma^2k \frac{\log(1/\epsI)}{\epsI^2}$ then the test correctly throws away \emph{all} the far points, and retains \emph{most} of the close points. Since the true centers
$\Xc\sim \Unif(\sqrt{d}\SphereD)^{\otimes k}$
are (w.h.p.) $\Omega(\sqrt{d})$-separated,
the remaining points in $\Net$, that have not been discarded, may be clustered into at most $k$ parts.
Step I concludes by returning a list $\XchI$ containing one representation of every cluster.

The dependence of Step I on the precision is sub-optimal: $N$ has to scale like $\frac{\log(1/\epsI)}{\epsI^2}$ instead of $1/\epsI$. This sub-optimal rate is mended in Step II.
We show that there is a \emph{constant} precision level $\eps_0$, that depends on $\RateLim,\RegimeParam>1$ (namely, \emph{how much} the rate is smaller than the capacity) so that whenever $\epsI\le \eps_0$, one can construct a mismatched decoder, using $\XchI$, that \emph{consistently} decodes messages encoded by the true codebook $\Xc$. In other words: given a measurement $\Y=\X_{\RandLabel}+\sigma\Z$, one can consistently estimate the unknown label $\RandLabel$ (up to a global re-labeling). In Step II we observe $\bar{N}=n-N$ new measurements, and cluster them according to their decoded label. For every cluster $i\in [k]$, we compute the corresponding sample average $\Avg_i$, and project it onto the ball $\m{B}(\0,\sqrt{d})$ to get our final estimate $\Xh_i=\Proj(\Avg_i)$. Since each label $i$ witnesses, on average, $\bar{N}/k$ measurements, the MSE is, to leading order, $d^{-1}\Expt\|\X_i-\Xh_i\|^2=\sigma^2 k /\bar{N}$.  Thus, using $N=C(\RateLim,\RegimeParam)\sigma^2 k $ measurements for Step I, and $\bar{N}=\sigma^2 k /\eps$ Step II, yields a list $\Xch$ with $\Expt\LossAvg(\Xc,\Xch)\le \eps$.

In the remainder of this section, we provide the full details of the strategy outlined above.


}



\subsection{Step I: Brute-Force Search}\label{sec:UpperBound-StepI}

Let $\epsI\in(0,1/2)$ a precision parameter. Let $\Net$ be a fixed $\sqrt{\epsI d/2}$-net  of $\sqrt{d}\SphereD$, such that $\forall\X\in\sqrt{d}\SphereD \exists\Xh\in\Net$ with $\|\X-\Xh\|^2\le \epsI d/2$ . By standard estimates, e.g. \cite[Example 5.8]{wainwright2019high}, we can assume that $|\Net|\le e^{Cd\log(1/\epsI)}$ for some universal $C>0$.
Our goal is to devise a procedure that, given $N$ samples $\Y_1,\ldots,\Y_N\sim \GMMd(\Xc,\sigma^2)$, will allow us 
to discard all candidates $\Xh\in \Net$ that are $\sqrt{\epsI d}$-far from \emph{all} the centers $\X_1,\ldots,\X_k$, while keeping enough candidates $\Xh$ that are $\sqrt{\epsI d}$-close to some center; ideally, at least one candidate close to almost every $\X_i$. 
Denote the sets, $\HClose,\HFar\subset\RR^d$
\begin{equation}\label{eq:UpperBound:Hypotheses}
    \begin{split}
        \HClose = \left\{\Xh\,:\,\dist^2(\Xh,\Xc)\le \frac12\epsI d\right\}, \ 
        \HFar=\left\{\Xh\,:\,\dist^2(\Xh,\Xc)\ge \epsI d\right\} \,.
    \end{split}
\end{equation}
{We would like 
a test that, 
with high probability: 1) rejects all $\Xh\in\HFar\cap \Net$; 2) accepts most of $\Xh\in\HClose \cap \Net$. 
Note that since $\Net$ is a $\sqrt{\eps_I d/2}$-cover, then for \emph{every} $1\le i\le k$, there is some $\Xh\in\HClose\cap\Net$ such that in fact $\|\Xh-\X_i\|^2\le \epsI d/2$.
}

As a first step, we consider a ``local test'' $\Test:\RR^d\times \RR^d\to\{0,1\}$, that takes a candidate $\Xh\in\Net$ and a single sample $\Y\sim \GMMd$, and outputs a decision $\in \{0,1\}$. Consider the quantities:
\begin{equation}\label{eq:UpperBound:qs}
    \begin{split}
        \qClose(\Xc,\Net) &= \min_{\Xh\in \HClose\cap \Net} \Pr(\Test(\Xh,\Y)=1|\Xc)\,,\\
        \qFar(\Xc,\Net) &= \max_{\Xh\in \HFar\cap \Net} \Pr(\Test(\Xh,\Y)=1|\Xc) \,.
    \end{split}
\end{equation}
For a local test $\Test:\RR^d\times \RR^d\to\{0,1\}$, a cover $\Net$,  and $\nu>0$, define
\begin{align}
\mathcal{E}_{\Test,\Net,\nu}=\left\{\Xc\in \left(\sqrt{d}\SphereD\right)^{k} \ : \ \qClose(\Xc,\Net)\ge \frac{1}{2}k^{-1}, \ \qFar(\Xc,\Net)\le 2 k^{-1-\nu} \right\}.
\label{eq:codenogoodfortest}
\end{align}
We  construct a local test with the following properties.
\begin{lemma}\label{lem:UpperBound:Technical}
Assume that $\RegimeParam>1$, and fix a cover $\Net$ of size $|\Net|\le e^{Cd\log(1/\epsI)}$. There are positive constants $\eps_0,c$, that depend on $\RateLim,\RegimeParam$,
    and a local test, $\Test:\RR^d\times \RR^d\to\{0,1\}$ (that depends on $d,k,\sigma^2$) such that for every \emph{fixed} $\epsI \in (0,\eps_0)$
    \begin{align}
    \lim_{d\to\infty} \Pr\left(\Xc\notin \mathcal{E}_{\Test,\Net,c\epsI^2}\right)=0.
    \end{align}
    
    
    \end{lemma}
{
We propose a local test $\Test$ based on the capacity-achieving decoder used in the proof of Proposition~\ref{prop:DecodingSphericalCodes}. Due to space constraints, the details are deferred to Appendix, Section~\ref{sec:proof-lem:UpperBound:Technical}.
}

Note that if $\qFar(\Xc,\Net)\ll \qClose(\Xc,\Net)$, as is the case for $\Xc\in \mathcal{E}_{\Test,\Net,\nu}$, then by observing the statistics of the $N$ local test outputs $\{\Test(\Xh,\Y_j)\}_{j=1}^N$, which are i.i.d. Bernoulli random variables, one can distinguish between $\Xh\in\HClose$ and $\Xh\in\HFar$ with error probability vanishing in $N$. In particular, consider the candidates $\Xh\in\Net$ that pass the following threshold-based test
\begin{equation}\label{eq:NetClose}
    \NetClose = \left\{  \Xh\in\Net\quad:\quad \sum_{j=1}^N \Test(\Xh,\Y_j) \ge \frac{1}{4}k^{-1}N \right\}\,.
\end{equation}
{
\begin{lemma}\label{lem:UpperBound:TechnicalCorr}
    Fix any $\Xc\in \mathcal{E}_{\Test,\Net,c\epsI^2}$
    and suppose that
    \begin{equation}
        \label{eq:UpperBound:TechnicalCorr:N}
            N\ge C_1\sigma^2 k \frac{\log(1/\epsI)}{\epsI^2} + C_2k\log(1/\varphi)\,,
    \end{equation}
    where $C_1=C_1(\RateLim,\RegimeParam)$, $C_2>0$ is a universal constant and $\varphi\in (0,1)$. Then w.p. $1-\varphi-o_{\RegimeParam,\RateLim}(1)$ over $\Y_1,\ldots,\Y_n\sim\GMMd(\Xc,\sigma^2) $, the following event holds:
    \begin{enumerate}
        [topsep=3pt,itemsep=-1ex,partopsep=1ex,parsep=1ex]
        \item (No far candidates). $\NetClose\cap \HFar \ne \emptyset$.
        \item (Most centers have a cluster). There is $\IndexSet\subseteq [k]$ with $|\IndexSet|\ge(1-\varphi)k$ and ${\max_{i\in\IndexSet}\dist^2(\X_i,\NetClose)\le \epsI d}$.
    \end{enumerate}
\end{lemma}
The proof of Lemma~\ref{lem:UpperBound:TechnicalCorr} appears in the Appendix, Section~\ref{sec:proof-lem:UpperBound:TechnicalCorr}.

To conclude step I, note that if the minimal distance between centers is $>4\sqrt{\eps_Id}$, then two candidates that are $\sqrt{\epsI d}$-close to different centers are necessarily $2\sqrt{\epsI d}$-far from one another. 

Let $\XchI$ be any $2\sqrt{\epsI d}$-separated subset of $\NetClose$ of maximal size. We prove in Appendix, Section~\ref{sec:proof-lem:UpperBound-StepI-Final} that with high probability, $\Xc$ indeed has $\Omega(\sqrt{d})$ minimal distance, and so:
\begin{lemma}\label{lem:UpperBound-StepI-Final}
    %
    Assume that $\RegimeParam>1$, $\epsI\le \eps_0(\RateLim,\RegimeParam)$ is small enough, and $N$ satisfies \eqref{eq:UpperBound:TechnicalCorr:N}. W.p. ${1-\varphi-o_{\RegimeParam,\RateLim,\epsI}(1)}$ over \emph{both} $\Xc\sim\Unif(\sqrt{d}\SphereD)^{\otimes k}$ and $[\Y_1,\ldots,\Y_n\,|\,\Xc]\sim \GMMd(\Xc,\sigma^2)$, the following event holds:
    
    \begin{enumerate}
        [topsep=3pt,itemsep=-1ex,partopsep=1ex,parsep=1ex]
        \item $\XchI$ is a list of  size $(1-\varphi)k\le m \le k$.
        \item There is $\IndexSet\subseteq [k]$ with $|\IndexSet|=m$ so that for every $i\in \IndexSet$, there is a \emph{unique} $\XhI\in\XchI$ such that $\|\X_i-\XhI\|^2\le \epsI d$. 
    \end{enumerate}
\end{lemma}
}
\subsection{Step II: Clustering and Averaging}\label{sec:UpperBound-StepII}

Upon the successful completion of Step I, Lemma~\ref{lem:UpperBound-StepI-Final}, we have produced a partial codebook $\XchI$ of size $m\ge (1-\varphi)k$. Moreover, there is a large subset of messages $\IndexSet\subset[k]$, $|\IndexSet|=m$ such that for all $i\in\IndexSet$, the \emph{true, unknown} codeword $\X_i$ is $\sqrt{\epsI d}$-close to a unique codeword $\XTilde_l$ of $\XchI$. Provided that $\epsI\le\eps_0(\RateLim,\RegimeParam)$ is small enough (but constant), it turns out we can construct a ``mismatched decoder'', using $\XchI$, that can consistently decode measurements $Y=\X_\RandLabel+\sigma\Z$ in the following sense: 1) If $\RandLabel\in\IndexSet$ then, up to a global relabeling, the decoder returns the correct label $\RandLabel$; 2) If $\RandLabel\notin\IndexSet$, the decoder consistently returns an error symbol $\DecodeErrorSymb$.
Due to space constraints, we defer all the details to Appendix, Section~\ref{sec:Appendix:StepII:DecodingApprox}.


In Step II we are given $\bar{N}=n-N$ new measurements. We use $\XchI$, the codebook from Step I, to decode the corresponding labels; measurements for which the decoder returns $\DecodeErrorSymb$ are \emph{discarded}. We end up with $m$ clusters, and for each cluster $l\in[m]$ we compute the corresponding sample mean $\Avg_l$. Finally, we return the list $\Xch=(\Xh_1,\ldots,\Xh_k$) such that $\Xh_l=\Proj(\Avg_l)$ for $l\in [m]$, $\Proj(\cdot)$ being the projection onto the ball $\m{B}(\0,\sqrt{d})$, and $\Xh_l=\0$ for $m+1\le l \le k$. 

The following Lemma bounds the error of the \emph{entire} end-to-end procedure, including both Step I and II. The details are deferred to Appendix, Section~\ref{sec:proof-lem:StepII:Details}.
%
\begin{lemma}\label{lem:StepII}
    Assume that $\RegimeParam>1$ and $\eps\le \eps_0(\RateLim,\RegimeParam)$ is small enough. 
    Suppose that 
    \begin{enumerate}
        \item Step I is run with $N\ge C\sigma^2k + Ck\log(1/\varphi)$ measurements,
        \item Step II is run with
        $\bar{N}\ge \frac{k\sigma^2}{\eps} + C\frac{k}{\eps^{1/2}}\log(1/\varphi)$ measurements,
    \end{enumerate}
    where $C=C(\RateLim,\RegimeParam)$ is constant and $\varphi\in(0,1)$ is a parameter. 
    Then
    \[
    \lim_{d\to\infty}\Expt\LossAvg(\Xc,\Xch) \le \frac{\eps}{1-\eps^{1/4}} + 12\varphi \,.    
    \]
\end{lemma}

\begin{proof}
    (Of Theorem~\ref{thm:BelowCapacity}). The claimed lower bound follows from Lemma~\ref{lem:LowerBound:Trivial}. The upper bound follows by setting, e.g., $\varphi=\eps^2$ in Lemma~\ref{lem:StepII}, noting that as $\eps\to0$, $\frac{\eps}{1-\eps^{1/4}}=\eps+o(\eps)$.
\end{proof}

\acks{We are grateful to Uri Erez for helpful discussions. The work of ER
and OO is supported in part by the ISF under grant 1641/21. ER is supported in part by an
Einstein-Kaye fellowship from the Hebrew University of Jerusalem. TB is supported in part by
the ISF grant no. 1924/21, the BSF grant no. 2020159, and the NSF-BSF grant no. 2019752.}

\bibliography{refs.bib}

\begin{thebibliography}{48}
\providecommand{\natexlab}[1]{#1}
\providecommand{\url}[1]{\texttt{#1}}
\expandafter\ifx\csname urlstyle\endcsname\relax
  \providecommand{\doi}[1]{doi: #1}\else
  \providecommand{\doi}{doi: \begingroup \urlstyle{rm}\Url}\fi

\bibitem[Achlioptas and McSherry(2005)]{achlioptas2005spectral}
Dimitris Achlioptas and Frank McSherry.
\newblock On spectral learning of mixtures of distributions.
\newblock In \emph{International Conference on Computational Learning Theory},
  pages 458--469. Springer, 2005.

\bibitem[Adler and Taylor(2009)]{adler2009random}
Robert~J Adler and Jonathan~E Taylor.
\newblock \emph{Random fields and geometry}.
\newblock Springer Science \& Business Media, 2009.

\bibitem[Anandkumar et~al.(2014)Anandkumar, Ge, Hsu, Kakade, and
  Telgarsky]{anandkumar2014tensor}
Animashree Anandkumar, Rong Ge, Daniel Hsu, Sham~M Kakade, and Matus Telgarsky.
\newblock Tensor decompositions for learning latent variable models.
\newblock \emph{Journal of machine learning research}, 15:\penalty0 2773--2832,
  2014.

\bibitem[Anderson et~al.(2014)Anderson, Belkin, Goyal, Rademacher, and
  Voss]{anderson2014more}
Joseph Anderson, Mikhail Belkin, Navin Goyal, Luis Rademacher, and James Voss.
\newblock The more, the merrier: the blessing of dimensionality for learning
  large {G}aussian mixtures.
\newblock In \emph{Conference on Learning Theory}, pages 1135--1164. PMLR,
  2014.

\bibitem[Arora and Kannan(2001)]{sanjeev2001learning}
Sanjeev Arora and Ravi Kannan.
\newblock Learning mixtures of arbitrary gaussians.
\newblock In \emph{Proceedings of the thirty-third annual ACM symposium on
  Theory of computing}, pages 247--257, 2001.

\bibitem[Artstein-Avidan et~al.(2015)Artstein-Avidan, Giannopoulos, and
  Milman]{artstein2015asymptotic}
Shiri Artstein-Avidan, Apostolos Giannopoulos, and Vitali~D Milman.
\newblock \emph{Asymptotic geometric analysis, Part I}, volume 202.
\newblock American Mathematical Soc., 2015.

\bibitem[Ashtiani et~al.(2018)Ashtiani, Ben-David, Harvey, Liaw, Mehrabian, and
  Plan]{ashtiani2018nearly}
Hassan Ashtiani, Shai Ben-David, Nicholas~JA Harvey, Christopher Liaw, Abbas
  Mehrabian, and Yaniv Plan.
\newblock Nearly tight sample complexity bounds for learning mixtures of
  {G}aussians via sample compression schemes.
\newblock In \emph{Proceedings of the 32nd International Conference on Neural
  Information Processing Systems}, pages 3416--3425, 2018.

\bibitem[Barg and Forney(2002)]{bf02}
A.~Barg and G.D. Forney.
\newblock Random codes: minimum distances and error exponents.
\newblock \emph{IEEE Transactions on Information Theory}, 48\penalty0
  (9):\penalty0 2568--2573, 2002.

\bibitem[Bennatan et~al.(2008)Bennatan, Calderbank, and Shamai]{bcs08}
Amir Bennatan, A~Robert Calderbank, and Shlomo Shamai.
\newblock Bounds on the {MMSE} of “bad” {LDPC} codes at rates above
  capacity.
\newblock In \emph{2008 46th Annual Allerton Conference on Communication,
  Control, and Computing}, pages 1065--1072, 2008.

\bibitem[Bhaskara et~al.(2014)Bhaskara, Charikar, and
  Vijayaraghavan]{bhaskara2014uniqueness}
Aditya Bhaskara, Moses Charikar, and Aravindan Vijayaraghavan.
\newblock Uniqueness of tensor decompositions with applications to polynomial
  identifiability.
\newblock In \emph{Conference on Learning Theory}, pages 742--778. PMLR, 2014.

\bibitem[Boucheron et~al.(2013)Boucheron, Lugosi, and
  Massart]{boucheron2013concentration}
St{\'e}phane Boucheron, G{\'a}bor Lugosi, and Pascal Massart.
\newblock \emph{Concentration inequalities: A nonasymptotic theory of
  independence}.
\newblock Oxford university press, 2013.

\bibitem[Bustin and Shamai(2013)]{bs13}
Ronit Bustin and Shlomo Shamai.
\newblock {MMSE} of “bad” codes.
\newblock \emph{IEEE Transactions on Information Theory}, 59\penalty0
  (2):\penalty0 733--743, 2013.

\bibitem[Chan et~al.(2014)Chan, Diakonikolas, Servedio, and
  Sun]{chan2014efficient}
Siu-On Chan, Ilias Diakonikolas, Rocco~A Servedio, and Xiaorui Sun.
\newblock Efficient density estimation via piecewise polynomial approximation.
\newblock In \emph{Proceedings of the forty-sixth annual ACM symposium on
  Theory of computing}, pages 604--613, 2014.

\bibitem[Cover and Thomas(2012)]{cover2012elements}
Thomas~M Cover and Joy~A Thomas.
\newblock \emph{Elements of Information Theory}.
\newblock John Wiley \& Sons, 2012.

\bibitem[Dasgupta(1999)]{dasgupta1999learning}
Sanjoy Dasgupta.
\newblock Learning mixtures of {G}aussians.
\newblock In \emph{40th Annual Symposium on Foundations of Computer Science
  (Cat. No. 99CB37039)}, pages 634--644. IEEE, 1999.

\bibitem[Dasgupta and Schulman(2007)]{dasgupta2007probabilistic}
Sanjoy Dasgupta and Leonard~J Schulman.
\newblock A probabilistic analysis of {EM} for mixtures of separated, spherical
  {G}aussians.
\newblock \emph{Journal of Machine Learning Research}, 8:\penalty0 203--226,
  2007.

\bibitem[Diakonikolas et~al.(2018)Diakonikolas, Kane, and
  Stewart]{diakonikolas2018list}
Ilias Diakonikolas, Daniel~M Kane, and Alistair Stewart.
\newblock List-decodable robust mean estimation and learning mixtures of
  spherical {G}aussians.
\newblock In \emph{Proceedings of the 50th Annual ACM SIGACT Symposium on
  Theory of Computing}, pages 1047--1060, 2018.

\bibitem[Diakonikolas et~al.(2019)Diakonikolas, Kamath, Kane, Li, Moitra, and
  Stewart]{diakonikolas2019robust}
Ilias Diakonikolas, Gautam Kamath, Daniel Kane, Jerry Li, Ankur Moitra, and
  Alistair Stewart.
\newblock Robust estimators in high-dimensions without the computational
  intractability.
\newblock \emph{SIAM Journal on Computing}, 48\penalty0 (2):\penalty0 742--864,
  2019.

\bibitem[Erez and Zamir(2004)]{erezzamir04}
U.~Erez and R.~Zamir.
\newblock Achieving 1/2 log (1+{SNR}) on the {AWGN} channel with lattice
  encoding and decoding.
\newblock \emph{IEEE Transactions on Information Theory}, 50\penalty0
  (10):\penalty0 2293--2314, 2004.

\bibitem[Erez()]{UriLN}
Uri Erez.
\newblock Lecture notes for principles of coding and detection in
  communication: Capacity of the {AWGN} channel.
\newblock URL
  \url{https://www.eng.tau.ac.il/~anatolyk/courses/Uri/AWGN_sphere_decoder.pdf}.

\bibitem[Feldman et~al.(2006)Feldman, Servedio, and
  O’Donnell]{feldman2006pac}
Jon Feldman, Rocco~A Servedio, and Ryan O’Donnell.
\newblock {PAC} learning axis-aligned mixtures of {G}aussians with no
  separation assumption.
\newblock In \emph{International Conference on Computational Learning Theory},
  pages 20--34. Springer, 2006.

\bibitem[Gallager(1968)]{gallager1968information}
Robert~G Gallager.
\newblock \emph{Information theory and reliable communication}, volume~2.
\newblock Springer, 1968.

\bibitem[Ge et~al.(2015)Ge, Huang, and Kakade]{ge2015learning}
Rong Ge, Qingqing Huang, and Sham~M Kakade.
\newblock Learning mixtures of gaussians in high dimensions.
\newblock In \emph{Proceedings of the forty-seventh annual ACM symposium on
  Theory of computing}, pages 761--770, 2015.

\bibitem[Goyal et~al.(2014)Goyal, Vempala, and Xiao]{goyal2014fourier}
Navin Goyal, Santosh Vempala, and Ying Xiao.
\newblock Fourier {PCA} and robust tensor decomposition.
\newblock In \emph{Proceedings of the forty-sixth annual ACM symposium on
  Theory of computing}, pages 584--593, 2014.

\bibitem[Guo et~al.(2005)Guo, Shamai, and Verd{\'u}]{guo2005mutual}
Dongning Guo, Shlomo Shamai, and Sergio Verd{\'u}.
\newblock Mutual information and minimum mean-square error in gaussian
  channels.
\newblock \emph{IEEE Transactions on Information Theory}, 51\penalty0
  (4):\penalty0 1261--1282, 2005.

\bibitem[Guo et~al.(2013)Guo, Shamai, and Verd{\'u}]{guo2013interplay}
Dongning Guo, Shlomo Shamai, and Sergio Verd{\'u}.
\newblock The interplay between information and estimation measures.
\newblock \emph{Foundations and Trends{\textregistered} in Signal Processing},
  6\penalty0 (4):\penalty0 243--429, 2013.

\bibitem[Hardt and Price(2015)]{hardt2015tight}
Moritz Hardt and Eric Price.
\newblock Tight bounds for learning a mixture of two gaussians.
\newblock In \emph{Proceedings of the forty-seventh annual ACM symposium on
  Theory of computing}, pages 753--760, 2015.

\bibitem[Hopkins and Li(2018)]{hopkins2018mixture}
Samuel~B Hopkins and Jerry Li.
\newblock Mixture models, robustness, and sum of squares proofs.
\newblock In \emph{Proceedings of the 50th Annual ACM SIGACT Symposium on
  Theory of Computing}, pages 1021--1034, 2018.

\bibitem[Hsu and Kakade(2013)]{hsu2013learning}
Daniel Hsu and Sham~M Kakade.
\newblock Learning mixtures of spherical gaussians: moment methods and spectral
  decompositions.
\newblock In \emph{Proceedings of the 4th conference on Innovations in
  Theoretical Computer Science}, pages 11--20, 2013.

\bibitem[Johnson(1994)]{johnson1994continuous}
Norman Johnson.
\newblock \emph{Continuous univariate distributions}.
\newblock Wiley, New York, 1994.
\newblock ISBN 0471584959.

\bibitem[Kabatiansky and Levenshtein(1978)]{kabatiansky1978bounds}
Grigorii~Anatol'evich Kabatiansky and Vladimir~Iosifovich Levenshtein.
\newblock On bounds for packings on a sphere and in space.
\newblock \emph{Problemy peredachi informatsii}, 14\penalty0 (1):\penalty0
  3--25, 1978.

\bibitem[Kalai et~al.(2010)Kalai, Moitra, and Valiant]{kalai2010efficiently}
Adam~Tauman Kalai, Ankur Moitra, and Gregory Valiant.
\newblock Efficiently learning mixtures of two {G}aussians.
\newblock In \emph{Proceedings of the forty-second ACM symposium on Theory of
  computing}, pages 553--562, 2010.

\bibitem[Kannan et~al.(2008)Kannan, Salmasian, and Vempala]{kannan2008spectral}
Ravindran Kannan, Hadi Salmasian, and Santosh Vempala.
\newblock The spectral method for general mixture models.
\newblock \emph{SIAM Journal on Computing}, 38\penalty0 (3):\penalty0
  1141--1156, 2008.

\bibitem[Kothari et~al.(2018)Kothari, Steinhardt, and
  Steurer]{kothari2018robust}
Pravesh~K Kothari, Jacob Steinhardt, and David Steurer.
\newblock Robust moment estimation and improved clustering via sum of squares.
\newblock In \emph{Proceedings of the 50th Annual ACM SIGACT Symposium on
  Theory of Computing}, pages 1035--1046, 2018.

\bibitem[Kwon and Caramanis(2020)]{kwon2020algorithm}
Jeongyeol Kwon and Constantine Caramanis.
\newblock The {EM} algorithm gives sample-optimality for learning mixtures of
  well-separated gaussians.
\newblock In \emph{Conference on Learning Theory}, pages 2425--2487. PMLR,
  2020.

\bibitem[Li and Schmidt(2017)]{li2017robust}
Jerry Li and Ludwig Schmidt.
\newblock Robust and proper learning for mixtures of gaussians via systems of
  polynomial inequalities.
\newblock In \emph{Conference on Learning Theory}, pages 1302--1382. PMLR,
  2017.

\bibitem[Moitra and Valiant(2010)]{moitra2010settling}
Ankur Moitra and Gregory Valiant.
\newblock Settling the polynomial learnability of mixtures of gaussians.
\newblock In \emph{2010 IEEE 51st Annual Symposium on Foundations of Computer
  Science}, pages 93--102. IEEE, 2010.

\bibitem[Pearson(1894)]{pearson1894contributions}
Karl Pearson.
\newblock Contributions to the mathematical theory of evolution.
\newblock \emph{Philosophical Transactions of the Royal Society of London. A},
  185:\penalty0 71--110, 1894.

\bibitem[Polyanskiy and Wu(2014)]{polyanskiy2014lecture}
Yury Polyanskiy and Yihong Wu.
\newblock Lecture notes on information theory.
\newblock \emph{Lecture Notes for ECE563 (UIUC) and}, 6\penalty0
  (2012-2016):\penalty0 7, 2014.

\bibitem[Regev and Vijayaraghavan(2017)]{regev2017learning}
Oded Regev and Aravindan Vijayaraghavan.
\newblock On learning mixtures of well-separated gaussians.
\newblock In \emph{2017 IEEE 58th Annual Symposium on Foundations of Computer
  Science (FOCS)}, pages 85--96. IEEE, 2017.

\bibitem[Romanov et~al.(2021)Romanov, Bendory, and
  Ordentlich]{romanov2021multi}
Elad Romanov, Tamir Bendory, and Or~Ordentlich.
\newblock Multi-reference alignment in high dimensions: sample complexity and
  phase transition.
\newblock \emph{SIAM Journal on Mathematics of Data Science}, 3\penalty0
  (2):\penalty0 494--523, 2021.

\bibitem[Shannon(1957)]{shannon1957certain}
Claude~E Shannon.
\newblock Certain results in coding theory for noisy channels.
\newblock \emph{Information and control}, 1\penalty0 (1):\penalty0 6--25, 1957.

\bibitem[Shannon(1959)]{shannon1959probability}
Claude~E Shannon.
\newblock Probability of error for optimal codes in a {G}aussian channel.
\newblock \emph{Bell System Technical Journal}, 38\penalty0 (3):\penalty0
  611--656, 1959.

\bibitem[Srebro et~al.(2006)Srebro, Shakhnarovich, and
  Roweis]{srebro2006investigation}
Nathan Srebro, Gregory Shakhnarovich, and Sam Roweis.
\newblock An investigation of computational and informational limits in
  gaussian mixture clustering.
\newblock In \emph{Proceedings of the 23rd international conference on Machine
  learning}, pages 865--872, 2006.

\bibitem[Suresh et~al.(2014)Suresh, Orlitsky, Acharya, and
  Jafarpour]{acharya2014near}
Ananda~Theertha Suresh, Alon Orlitsky, Jayadev Acharya, and Ashkan Jafarpour.
\newblock Near-optimal-sample estimators for spherical gaussian mixtures.
\newblock \emph{Advances in Neural Information Processing Systems}, 27, 2014.

\bibitem[Vempala and Wang(2004)]{vempala2004spectral}
Santosh Vempala and Grant Wang.
\newblock A spectral algorithm for learning mixture models.
\newblock \emph{Journal of Computer and System Sciences}, 68\penalty0
  (4):\penalty0 841--860, 2004.

\bibitem[Vershynin(2018)]{vershynin2018high}
Roman Vershynin.
\newblock \emph{High-dimensional probability: An introduction with applications
  in data science}, volume~47.
\newblock Cambridge university press, 2018.

\bibitem[Wainwright(2019)]{wainwright2019high}
Martin~J Wainwright.
\newblock \emph{High-dimensional statistics: A non-asymptotic viewpoint},
  volume~48.
\newblock Cambridge University Press, 2019.

\end{thebibliography}

\newpage
\appendix 

\section{Proof of Proposition~\ref{prop:DecodingSphericalCodes}}
\label{sec:Appendix:Decoding}

As mentioned in the main text, the proof amounts to analyzing a certain sub-optimal decoder for the codebook $\Xc$.
While the decoders, and their analysis, are not new, we nonetheless provide all the details here as a ``warm-up'' for things to come.

We consider two different families of decoders, depending on whether one operates in the zero or positive rate regime.

\subsection{Rate Zero ($\RateLim=0$)}
\label{sec:Appendix:Decoding:0}

\paragraph{The decoder.} 
Observe that for a spherical code, the MAP decoder (\ref{eq:Background:DecOpt}) reduces to
\begin{equation}\label{eq:DecodeOpt-Sphere}
    \DecodeOpt(\Y) = \argmax_{1\le i \le k} d^{-1}\langle \Y,\X_i\rangle \,.
\end{equation}
For the analysis, we consider a sub-optimal decoder, based on thresholding the correlation in (\ref{eq:DecodeOpt-Sphere}). 

When $\Y=\X_i+\sigma\Z$, clearly, $\Expt[d^{-1}\langle\Y,\X_i\rangle]=1$, while for $j\ne i$,  $\Expt[d^{-1}\langle\Y,\X_j\rangle]=0$. Fix thresholds $0<\eta_1\le \eta_2$. Consider a decoding rule $\DecCORR_{\eta_1,\eta_2}:\RR^d\to[k]\cup\{\DecodeErrorSymb\}$ so that $\DecCORR_{\eta_1,\eta_2}(\Y)=i$ if and only if both of the following hold:
\begin{enumerate}
    \item $d^{-1}\langle\Y,\X_i\rangle\ge 1-\eta_1$.
    \item For all $j\ne i$, $d^{-1}\langle\Y,\X_j\rangle< 1-\eta_2$.
\end{enumerate}
Note that since $\eta_1\le \eta_2$, \emph{at most} one index $1\le i\le k$ can satisfy the above. 
If no such $i$ exists, we set $\DecCORR_{\eta_1,\eta_2}(\Y)=\DecodeErrorSymb$.

\paragraph{Analysis.} We proceed to bound the error of the decoder $\DecCORR_{\eta_1,\eta_2}$. 

By symmetry of the codebook generating process, the error probability (averaged over the ensemble) does not depend on the particular transmitted message (index).
For convenience, throughout this section, we always assume, without loss of generality, that the transmitted message is $\ell=i$ (and implicitly condition on this event). Thus, the value at the receiver end of the channel is $\Y=\X_i+\sigma\Z$.

\begin{lemma}\label{lem:Decode:0:A}
    Conditioned on any $\Xc\in (\sqrt{d}\SphereD)^k$,
    \begin{align*}
        \Pr\left(d^{-1}\langle\Y,\X_i\rangle < 1-\eta_1 \,|\,\Xc \right) \le e^{-\frac{\eta_1^2}{2\sigma^2}d} \,.
    \end{align*}
\end{lemma}
\begin{proof}
    $d^{-1}\langle\X_i+\sigma\Z,\X_i\rangle \le 1-\eta_1$ is equivalent to $d^{-1}\langle\sigma\Z,\X_i\rangle \le -\eta_1$. Since $\Z\sim\m{N}(\0,\Id)$, we have $d^{-1}\langle\sigma\Z,\X_i\rangle\sim \m{N}(0,\sigma^2/d)$, and the bound follows immediately.
\end{proof}

\begin{lemma}\label{lem:Decode:0:B}
    For fixed $\eta_1>0$, define the set
    \begin{equation*}
        \mathbb{X}_{i} = \left\{ \Xc\in(\sqrt{d}\SphereD)^k\;:\; \max_{j\in [n]\setminus\{i\}} d^{-1}\langle\X_i,\X_j\rangle \le \sqrt{\frac{2\log (k-1)}{d} + \frac{\eta_1^2}{\sigma^2}} \right\} \,.
    \end{equation*}
    For $\Xc\sim \Unif(\sqrt{d}\SphereD)$: 
$\Pr(\Xc\notin\mathbb{X}_i) \le e^{-\frac{\eta_1^2}{2\sigma^2}d}$.

\end{lemma}
\begin{proof}
    By the standard tail bound Lemma~\ref{lem:Unif-Sphere}, for $t\ge 0$, $\Pr( d^{-1}\langle \X_i,\X_j\rangle \ge t) \le e^{-dt^2/2}$. Taking a union bound over $(k-1)$ choices for $j\ne i$,
    $\Pr( \max_{j\ne i} d^{-1}\langle \X_i,\X_j\rangle \ge t) \le e^{-dt^2/2+\log (k-1)}$. Now set $t=\sqrt{\frac{2\log(k-1)}{d} + \frac{\eta_1^2}{\sigma^2}}$.
\end{proof}

\begin{lemma}\label{lem:Decode:0:C}
    Suppose that $1-\eta_2\ge \sqrt{\frac{2\log (k-1)}{d} + \frac{\eta_1^2}{\sigma^2}} + \sqrt{\frac{2\sigma^2\log(k-1)}{d}}$. Then
    \begin{align}\label{eq:lem:Decode:0:C}
        \Pr\left(\max_{j\ne i} d^{-1}\langle \Y,\X_j\rangle \ge1-\eta_2\right) 
        \le 
        e^{-\frac{d}{2}\left(1-\eta_2-\sqrt{\frac{2\log(k-1)}{d} + \frac{\eta_1^2}{\sigma^2}} - \sqrt{\frac{2\sigma^2\log(k-1)}{d}}\right)^2 } + e^{-\frac{\eta_1^2}{2\sigma^2}d} \,.
    \end{align}
\end{lemma}
\begin{proof}
    Fix $\Xc\in\mathbb{X}_i$, where the set $\mathbb{X}_i$ is from Lemma~\ref{lem:Decode:0:B}.
    Writing $\Y=\X_i+\sigma\Z$, we note that
    \begin{align*}
        \Pr\left(\max_{j\ne i} d^{-1}\langle \Y,\X_j\rangle \ge1-\eta_2 \,\big|\, \Xc\right) 
        \le 
        \Pr\left(\max_{j\ne i} d^{-1}\sigma \langle \Z ,\X_j\rangle \ge1-\eta_2-\sqrt{\frac{2\log (k-1)}{d} + \frac{\eta_1^2}{\sigma^2}} \,\big|\, \Xc\right) \,.
    \end{align*}
    Now, each $ d^{-1}\sigma \langle \Z ,\X_j\rangle  $ is Gaussian with mean $0$ and variance $\sigma^2/d$. By a standard bound on the maximum of Gaussian random variables, Lemma~\ref{lem:MaxGauss-Expt}, $\Expt[\max_{j\ne i} d^{-1}\sigma \langle \Z ,\X_j\rangle] \le \sqrt{\frac{2\sigma^2\log(k-1)}{d}}$. By the Borell-TIS inequality, Lemma~\ref{lem:Borell-TIS}, we obtain the first term of (\ref{eq:lem:Decode:0:C}). The second term is just 
    the bound on $\Pr(\Xc\notin \mathbb{X}_i)$ from  Lemma~\ref{lem:Decode:0:B}.
\end{proof}

\begin{proof}
    (Of Proposition~\ref{prop:DecodingSphericalCodes}, case $\RateLim=0$.)

    Combining Lemmas~\ref{lem:Decode:0:A} and \ref{lem:Decode:0:C}, for every $\eta_1,\eta_2$ satisfying
    \begin{equation}
        \label{eq:Decode:0:Condition}
        0 < \eta_1 \le \eta_2 < 1 - \sqrt{\frac{2\log (k-1)}{d} + \frac{\eta_1^2}{\sigma^2}} - \sqrt{\frac{2\sigma^2\log(k-1)}{d}}   
    \end{equation}
    the decoder $\DecCORR_{\eta_1,\eta_2}$ attains average error
    \begin{equation}\label{eq:Decode:0:ErrorProbability}
        \Pr\left(i \ne \DecCORR_{\eta_1,\eta_2}(\X_i+\sigma\Z)\right) \le 
        e^{-\frac{d}{2}\left(1-\eta_2-\sqrt{\frac{2\log(k-1)}{d} + \frac{\eta_1^2}{\sigma^2}} - \sqrt{\frac{2\sigma^2\log(k-1)}{d}}\right)^2 } + 2e^{-\frac{\eta_1^2}{2\sigma^2}d} \,.
    \end{equation}
    To prove the proposition, it clearly suffices to show that
    when $\Rate=\Capacity(\RegimeParam\sigma^2)$, $\RegimeParam>1$, then
    $\sqrt{\frac{2\log (k-1)}{d} + \frac{\eta_1^2}{\sigma^2}} + \sqrt{\frac{2\sigma^2\log(k-1)}{d}}    $ is at most a constant, which is strictly smaller than $1$. Indeed, since $\sigma^2=\omega(1)$, the first term is $o(1)$. As for the second term, 
    \begin{align*}
        \sqrt{\frac{2\sigma^2\log(k-1)}{d}} \le \sqrt{2\sigma^2\Rate}=\sqrt{2\sigma^2\Capacity(\RegimeParam\sigma^2)} \le \sqrt{1/\RegimeParam}<1\,,
    \end{align*}
    where we used $\Capacity(s)=\frac12\log(1+1/s)\le 1/(2s)$.

\end{proof}

\subsection{Positive Rate ($\RateLim>0$)}
\label{sec:Appendix:Decoding:Pos}

\paragraph{Remark.} The analysis of the previous section ($\RateLim=0$) unfortunately fails in the positive rate regime, where $\sigma^2$ is constant. To have any hope of finding $\eta_1\le \eta_2$ that satisfy condition (\ref{eq:Decode:0:Condition}), it is necessary that (taking $\eta_1,\eta_2\to 0$)
\[
    \sqrt{\frac{2\log (k-1)}{d}} + \sqrt{\frac{2\sigma^2\log(k-1)}{d}} \le 1 \,.
\]
Since $\log(k-1)/d=\Rate-O(d^{-1})$, this constrains the rate as ${\Rate\le \frac{1}{\sqrt{2}(1+\sigma)}+O(d^{-1})}$. For small $\sigma$, this bound is $\approx 1/\sqrt{2}$, while $\Capacity(\RegimeParam\sigma^2)\approx \log(1/\RegimeParam\sigma^2)$. Consequently, for $\sigma^2=O(1)$ this condition fails to hold, and the analysis from Section~\ref{sec:Appendix:Decoding:0} is not sufficient for proving the existence of capacity-approaching codes. We note that this is a well-known limitation of the \emph{analysis}; specifically, Lemma~\ref{lem:Decode:0:C} is too crude. It estimates the maximum over ``noise terms'', $\max_{j\ne i} d^{-1}\sigma\langle\Z,\X_j\rangle$ \emph{as if} they were all independent. In the zero rate regime, different codewords are essentially orthogonal: $d^{-1}\langle \X_i,\X_j\rangle \lesssim \sqrt{\frac{\log k}{d}}=o(1)$; consequently, by standard results (e.g. \cite[2.2.5]{adler2009random}), the maximum is indeed very close to the maximum of i.i.d. Gaussians. When $k$ is exponential, however, this is no longer the case, and the correlations between these noise terms can no longer be neglected once $\RateLim$ is sufficiently large. Thus, different techniques are necessary to carry out the analysis (cf. the classical book \cite{gallager1968information}).

\paragraph{The decoder.} 
To overcome the obstruction mentioned above, we consider a different, sub-optimal, decoder, which is similar to Shannon's information density threshold decoder \cite{shannon1957certain} for a Gaussian i.i.d. codebook, and to that used in~\cite{erezzamir04}, see also~\cite{UriLN} and~\cite{polyanskiy2014lecture}. Let $\alpha=\frac{1}{1+\sigma^2}$ and $\tau=\sigma^2\alpha=1-\alpha$. For parameters $\tau \le \tau_1 \le \tau_2$, consider a decoder $\DecMMSE_{\tau_1,\tau_2} : \RR^d\to[k]\cup\{\DecodeErrorSymb\}$ so that $\DecMMSE_{\tau_1,\tau_2}(\Y)=i$ if and only if both of the following hold:
\begin{enumerate}
    \item $d^{-1}\|\alpha\Y-\X_i\|^2 \le \tau_1$.
    \item For all $j\ne i$, $d^{-1}\|\alpha\Y-\X_i\|^2 > \tau_2$. 
\end{enumerate}
If no such $1\le i \le k$ exists, then $\DecMMSE_{\tau_1,\tau_2}(\Y)=\DecodeErrorSymb$.

As was before, in the zero rate case, we analyze the error probability conditioned on the transmitted message being some fixed $\RandLabel=i$; by symmetry, the (ensemble-averaged) error probability does not depend on $\RandLabel$. Thus, below, $\Y=\X_i+\sigma\Z$.

To justify the name $\DecMMSE$ recall that the best linear estimator of $\X_i$ from $\Y=\X_i+\sigma\Z$, in the sense of smallest MSE (LMMSE), is $\alpha\Y$.\footnote{When $\X_i\sim \m{N}(0,1)^{\otimes d}$ is i.i.d. Gaussian, the LMMSE is actually the MSE-optimal estimator (MMSE). Since we use a spherical prior for $\X_i\sim \Unif(\sqrt{d}\SphereD)$, this is no longer holds exactly, though the discrepancy is negligible when one operates in the regime $\sigma^2=\Omega(1)$.} Note also that $d^{-1}\Expt\|\alpha\Y-\X_i\|^2=\tau$, whereas for $j\ne i$, $d^{-1}\Expt\|\alpha\Y-\X_j\|^2=\alpha^2(1+\sigma^2)+1=\alpha+1>\tau$. 

\paragraph{Analysis.} 
We proceed to bound the error of the decoder $\DecMMSE_{\tau_1,\tau_2}$.

\begin{lemma}
    \label{lem:Decode:Pos:1}
    For any $\Xc\in(\sqrt{d}\SphereD)^k$,
    \begin{align*}
        \Pr\left(d^{-1}\|\alpha\Y-\X_i\|^2 > \tau_1 \,\big|\,\Xc\right) \le e^{-\frac12{(1+\sigma^2)(\sqrt{\tau_1/\tau}-1)^2}d} \,.
    \end{align*}
\end{lemma}
\begin{proof}
    The mapping
    $\Z\mapsto F(\Z)=d^{-1/2}\|\alpha(\X_i+\sigma\Z)-\X_i\|$ is $d^{-1/2}\alpha\sigma$-Lipschitz, with expectation
    \[
        \Expt F(\Z) \le \sqrt{ d^{-1}\Expt\|\alpha(\X_i+\sigma\Z)-\X_i\|^2 } \le \sqrt{\tau}\,.
    \]
    Applying the Gaussian Lipschitz concentration inequality, Lemma~\ref{lem:Gaussian-Lip},
    \begin{align*}
        \Pr\left( (F(\Z))^2 >\tau_1 \right)
        &=
        \Pr\left( F(\Z)>\sqrt{\tau_1} \right)\\
        &\le 
        \Pr\left(F(\Z)-\Expt F(\Z)> \sqrt{\tau_1}-\sqrt{\tau} \,\big|\,\Xc\right) \\
        &\le e^{-\frac12 \frac{(\sqrt{\tau_1}-\sqrt{\tau})^2}{(d^{-1/2}\alpha\sigma)^2}} \\
        &= e^{- \frac{\tau}{2\alpha^2\sigma^2}(\sqrt{\tau_1/\tau}-1)^2 d} \,.        
    \end{align*}
    Now plug $\tau=\alpha\sigma^2$, $\alpha=1/(1+\sigma^2)$ to get the claimed bound.
\end{proof}

\begin{lemma}\label{lem:Decode:Pos:2}
    For $j\ne i$, $\Xc\sim (\Unif(\sqrt{d}\SphereD))^k$,
    \begin{align*}
        \Pr\left(d^{-1}\|\alpha\Y_i-\X_j\|^2 \le \tau_2\right) 
        \le 
        \left( 1+\frac{1}{\sigma^2} \right)^{1/2} e^{-\left(\Capacity(\sigma^2) - \frac12 \log\frac{\tau_2}{\tau} \right)d} \,.
    \end{align*}
\end{lemma}
\begin{proof}
    For a compact convex body $K\subset \RR^d$, we denote its boundary by $\partial K$ and surface area by $\Surf(\partial K)$. In addition, we denote the Euclidean ball of radius $r$, centered around $\bm{a}\in\RR^d$, by $\m{B}(\bm{a},r)$.

    The event above, whose probability we wish to bound, is equivalent to the event ${\X_j\in\m{B}(\alpha\Y_i,\sqrt{\tau_2 d})}$. Since $\X_j\sim \Unif(\sqrt{d}\SphereD)$, this probability (conditioned on $\Y$), is given by the surface area ratios ${\Surf(\sqrt{d}\SphereD\cap \m{B}(\alpha\Y,\sqrt{\tau_2 d}) )}/\Surf(\sqrt{d}\SphereD)$. Since $\partial K\cap L\subseteq \partial (K\cap L)$, and the surface area of convex sets is monotonic with respect to containment (e.g., \cite[Theorem B.1.14]{artstein2015asymptotic}), $\Surf(\sqrt{d}\SphereD\cap \m{B}(\alpha\Y,\sqrt{\tau_2 d}) )\le \Surf(\partial \m{B}(\alpha\Y,\sqrt{\tau_2 d}))$. Consequently, the probability is bounded by 
    \[
        \frac{\Surf(\partial \m{B}(\alpha\Y_i,\sqrt{\tau_2 d}))}{\Surf(\sqrt{d}\SphereD)} = \tau_2^{\frac{d-1}{2}} = \tau_2^{-1/2} e^{-\frac12\left(\log\frac1\tau - \log\frac{\tau_2}{\tau} \right)d} \,.
    \]
    Lastly, use $\tau_2\ge \tau$ and $1/\tau=1+1/\sigma^2$.
\end{proof}

\begin{proof}
    (Of Proposition~\ref{prop:DecodingSphericalCodes}, case $\RateLim>0$.)

    Combining Lemmas~\ref{lem:Decode:Pos:1} and \ref{lem:Decode:Pos:2}, along with a union bound over all $j\ne i$, the decoding error, averaged over the ensemble $\Xc\sim\Unif(\sqrt{d}\SphereD)^{\otimes k}$, is bounded as 
    \begin{equation}\label{eq:Decode:Pos:ErrorProbability}
        \Pr\left(i\ne \DecMMSE_{\tau_1,\tau_2}(\Y)\right) \le 
        e^{- \frac12(1+\sigma^2)(\sqrt{\tau_1/\tau}-1)^2 d} 
        + (k-1)\left( 1+\frac{1}{\sigma^2} \right)^{1/2}e^{-\left(\Capacity(\sigma^2) - \frac12\log\frac{\tau_2}{\tau} \right)d} \,.
    \end{equation}

    Choose $\tau_1=c\tau$, $\tau_2=c^2\tau$ where $c>1$ is a sufficiently small constant. In that case, the first term of (\ref{eq:Decode:Pos:ErrorProbability}) clearly decays exponentially in $d$. 

    As for the second term, set $k=e^{d\Capacity(\RegimeParam\sigma^2)}$, so the term is bounded like $(1+\frac{1}{\sigma^2})^{1/2}e^{-d\left(\Capacity(\sigma^2)-\Capacity(\RegimeParam\sigma^2)-\log(c)\right)}$. Since $\Capacity(\sigma^2)-\Capacity(\RegimeParam\sigma^2)$ is a positive constant, if $c>1$ is small enough then the term decays exponentially fast in $d$.
\end{proof}

\newpage

\section{Proofs from Section~\ref{sec:LowerBound}}

\subsection{Proof of Lemma~\ref{lem:LowerBound:RDF-lb}}
\label{sec:proof-lem:LowerBound:RDF-lb}

We reduce the calculation into a ``standard'' rate-distortion function (RDF) under MSE distortion. 

For technical reasons, it will be more convenient to work with a Gaussian prior on the source signal, rather than the uniform distribution over the sphere. The reason is that the latter distribution is not absolutely continuous with respect to Lebesgue measure (it is supported on a manifold of positive co-dimension, namely, $\SphereD$), so that its differential entropy (in the usual sense) is not well-defined. 

Introduce Gaussian random variables, $\Gc=(\G_1,\ldots,\G_k)\sim \m{N}(0,1)^{\otimes dk}$, so that ${\X_i=\sqrt{d}\G_i/\|\G_i\|}$. We have the Markov chain,
\begin{align}\label{eq:proof-lem:LowerBound:RDF-lb:MarkovChainGaussian}
	\Gc\longrightarrow\Xc\longrightarrow\Xch\,,
\end{align}
so by the DPI,
\begin{equation}\label{eq:proof-lem:LowerBound:RDF-lb:0}
	\MI(\Xc;\Xch)\ge \MI(\Gc;\Xch)\,.
\end{equation}
The next Lemma shows that if $\Xch$ estimates $\Xc$ with small distortion, then it also estimates $\Gc$ with small comparable distortion:

\begin{lemma}\label{lem:LowerBound-Gaussian-Distortion}
	Suppose the the Markov chain (\ref{eq:proof-lem:LowerBound:RDF-lb:MarkovChainGaussian}) holds. Then for universal constant $c_0$,
	\[
	\Expt\Loss(\Gc,\Xch) \le \left(\sqrt{\Expt\Loss(\Xc,\Xch)} + c_0d^{-1/2}\right)^2  \,.	
	\]
\end{lemma}
The proof of Lemma~\ref{lem:LowerBound-Gaussian-Distortion} is straightforward, and deferred to Section~\ref{sec:proof-lem:LowerBound-Gaussian-Distortion}.

By assumption, $\Expt\Loss(\Xc,\Xch)\le \eps$ and therefore, by Lemma~\ref{lem:LowerBound-Gaussian-Distortion}, ${\Expt\Loss(\Gc,\Xch)\le (\eps^{1/2}+c_0d^{-1/2})^2}$. 

Let $J=(j_1,\ldots,j_k)\in[k]^k$ be indices such that $j_i \in \argmin_{1\le j \le k} \|\G_i - \Xh_{j}\|^2$. In other words, $\Loss(\Gc,\Xch)=(dk)^{-1}\sum_{i=1}^k \|\G_i-\Xh_{j_i}\|^2$. The random variable
$J$ is, clearly, deterministic given $\Gc,\Xch$. By the chain rule for mutual information,
\begin{equation}
	\label{eq:proof-lem:LowerBound:RDF-lb:1}
	\MI(\Gc;\Xch) = \MI(\Gc;\Xch,J)-\MI(\Gc;J|\Xch)\,.
\end{equation}
Since $J$ is a discrete random variable,
\begin{equation}
	\label{eq:proof-lem:LowerBound:RDF-lb:2}
	\MI(\Gc;J|\Xch) := \Ent(J|\Xch)-\Ent(J|\Gc,\Xch)\le \Ent(J)\le \log (k^k)= k \log k \,,
\end{equation}
where we used the standard facts that the entropy of a discrete variable is non-negative, and that conditioning decreases entropy.

Set $\D_i=\X_{j_i}$ and $\Dc=(\D_1,\ldots,\D_k)\in \RR^{k\times d}$, so that, by definition,
\begin{align*}
    (dk)^{-1}{\Expt\|\Gc-\Dc\|_F^2=\Expt\LossAvg(\Gc,\Dc)\le (\eps^{1/2}+c_0d^{-1/2})^2 }\,.    
\end{align*}
Since $\Dc$ is a function of $(\Xch,J)$, the DPI implies $\MI(\Gc;\Dc)\le \MI(\Gc;\Xch,J)$. Thus,
\begin{align}
    \MI(\Gc;\Xch,J)
    &\ge \MI(\Gc;\Dc) \nonumber\\
    &\ge \min_{P_{\tilde{\Dc}|\Gc}\,:\,(dk)^{-1}\|\Gc-\tilde{\Dc}\|^2\le (\eps^{1/2}+c_0d^{-1/2})^2} \MI(\Gc;\tilde{\Dc}) \nonumber\\
    &= \frac{dk}{2}\log\left( \frac{1}{(\eps^{1/2}+c_0d^{-1/2})^2} \right) \,,
    \label{eq:proof-lem:LowerBound:RDF-lb:3}
\end{align}
where (\ref{eq:proof-lem:LowerBound:RDF-lb:3}) is the solution to the classical Gaussian source rate-distortion problem \cite[Chapter 27]{polyanskiy2014lecture}. The proof of Lemma~\ref{lem:LowerBound:RDF-lb} concludes by combining (\ref{eq:proof-lem:LowerBound:RDF-lb:0})-(\ref{eq:proof-lem:LowerBound:RDF-lb:3}). 
\qed

\vspace{2mm}

We remark that for sufficiently small $\eps$ the lower bound 
\begin{align}
\min_{P_{\Xch|\Gc}:\Expt\LossAvg(\Gc,\Dc)\leq \eps} \MI(\Gc;\Xch)\geq \frac{dk}{2}\log\left(\frac{1}{\eps}\right)-k\log k,
\label{eq:RDFtightness}
\end{align}
which we derived within the proof above, is in fact tight (up to the difference between $k\log k$ and $\log|\mathcal{S}_k|$, where $\mathcal{S}_k$ is the symmetric group of permutations on $[k]$). To see this, we consider the Markov chain $\Gc\to\hat{\Gc}\to\Xch$, where the channel from $\Gc$ to $\hat{\Gc}$ is the test channel attaining the Gaussian RDF (see e.g.,~\cite[Theorem 10.3.2]{cover2012elements}), and the channel from $\hat{\Gc}\to\Xch$ is defined by applying a uniform random permutation $J$ on $\hat{\Gc}$, resulting in $\Xch$. Note that
\begin{align}
\MI(\Gc;\Xch,J)=\MI(\Gc;\hat{\Gc})=\frac{dk}{2}\log\left(\frac{1}{\eps}\right)
\end{align}
and that
\begin{align}
\MI(\Gc;J|\Xch)=\Ent(J|\Xch)-\Ent(J|\Gc,\Xch)=\Ent(J)-\Ent(J|\Gc,\Xch)\approx \Ent(J),
\end{align}
where the last approximation is due to the fact that for small $\eps$ we can recover $J$ from $\Gc$ and $\Xch$. Thus, the approximate tightness of~\eqref{eq:RDFtightness} follows from~\eqref{eq:proof-lem:LowerBound:RDF-lb:1}. 

The subtractive $k\log{k}$ term we lose here is the reason that the lower bound in Theorem \ref{thm:BelowCapacity} is $e^{-2\RateLim}$ instead of $1$. While we believe that $1$ is the correct lower bound, this loss seems to be inherent to the mutual information bounding program we follow here.

\subsection{Proof of Lemma~\ref{lem:LowerBound-Gaussian-Distortion}}
\label{sec:proof-lem:LowerBound-Gaussian-Distortion}

It is a well-known fact \cite[Eq. 18.15]{johnson1994continuous}
that 
$\Expt\|\G_i\|=\sqrt{d} + O(d{^{-1/2}})$.
Consequently,
\[
\Expt\|\G_i-\X_i\|^2 = 2d - 2\sqrt{d}\Expt\|\G_i\| = O(1) \,.
\]


Let $J=(j_1,\ldots,j_k)$ be $j_i=\argmin_{1\le j \le k}\|\X_i-\Xh_i\|$. 
By definition of $\LossAvg(\cdot,\cdot)$, (\ref{eq:LossAvg-def}), 
\[
\LossAvg(\Gc,\Xch) \le \frac{1}{dk}\sum_{i=1}^k \Expt\|\G_i-\Xh_{j_i}\|^2\,,
\]
while $\LossAvg(\Xc,\Xch)=\frac{1}{dk}\sum_{i=1}^k \Expt\|\X_i-\Xh_{j_i}\|^2$.
Moreover, observe that 
\[
\Dc=(\D_1,\ldots,\D_k)\mapsto \left(\frac{1}{dk}\sum_{i=1}^k \Expt\|\D_i\|^2\right)^{1/2}
\]
defines a semi-norm on $d\times k$ matrices (with square-integrable entries). 
Thus, by the triangle inequality,
\begin{align*}
    \left(\Expt\LossAvg(\Gc,\Xch) \right)^{1/2}
    &\le \left(\frac{1}{dk}\sum_{i=1}^k \Expt\|\G_i-\Xh_{j_i}\|^2\right)^{1/2} \\
    &=\left(\frac{1}{dk}\sum_{i=1}^k \Expt\|(\G_i-\X_i) + (\X_i-\Xh_{j_i})\|^2 \right)^{1/2} \\
    &\le 
    \underbrace{
    \left( 
    \frac{1}{dk}\sum_{i=1}^k \Expt\|\G_i-\X_i\|^2
    \right)^{1/2}
    }_{\left( \frac{1}{dk}\sum_{i=1}^k O(1)\right)^{1/2}=O(d^{-1/2})}
    + 
    \underbrace{
    \left(
    \frac{1}{dk}\sum_{i=1}^k \Expt\|\X_{i}-\Xh_{j_i}\|^2
    \right)^{1/2}
    }_{\left( \LossAvg(\Xc,\Xch) \right)^{1/2}}
    \,.
\end{align*}

\qed

\subsection{Proof of Lemma~\ref{lem:LowerBound:Trivial}}
\label{sec:proof-lem:LowerBound:Trivial}

Write $\MI(\Xc;\Y_1,\ldots,\Y_n,\Ell)=\MI(\Xc;\Ell) + \MI(\Xc;\Y_1,\ldots,\Y_n|\Ell)$, with $\MI(\Xc;\Ell)=0$. For ${1\le i\le k}$, let $n_i=n_i(\Ell)=|\Ell^{-1}(i)|$ be the number of measurements labeled $i$. The proof amounts to the following observation: the desired MI $\MI(\Xc;\Y_1,\ldots,\Y_n|\Ell)$ is simply the cumulative MI across $k$ parallel Gaussians channel, with independent inputs $\X_1,\ldots,\X_k$, such that one observes $n_i$ outputs (samples) of each channel $i$. We now quantify this statement.

Let $I(\sigma^2,m)=\MI(\X;\X+\sigma\Z_1,\ldots,\X+\sigma\Z_m)$ be the input-output MI between $\X$ and $m$ outputs through an $\AWGN(\sigma^2)$ channel. Since the sample mean is a sufficient statistic for the true mean under a Gaussian measurement model, we have 
\[
    I(\sigma^2,m)=\MI\left(\X,\frac1m\left((\X+\sigma\Z_1,\ldots,\X+\sigma\Z_m)\right)\right)=I\left(\frac{\sigma^2}{m},1\right) \le d\Capacity(\sigma^2/m)\,,
\]
where $\C(\cdot)$ denotes the AWGN channel capacity (\ref{eq:CapacityFunc-def}). Thus, 
\[
    \MI(\Xc;\Y_1,\ldots,\Y_n|\Ell) = \Expt\left[ \sum_{i=1}^k I(\sigma^2,n_i(\Ell)) \right] \le \Expt\left[ \sum_{i=1}^k d\Capacity(\sigma^2/n_i(\Ell))\right]    \,.
\]
One may readily verify that the function $m\mapsto \Capacity(\sigma^2/m)$ is concave. By Jensen's inequality,
\[
    \sum_{i=1}^k d\Capacity(\sigma^2/n(\Ell)) = k\cdot \frac{1}{k}\sum_{i=1}^k d\Capacity(\sigma^2/n(\Ell)) \le kd\Capacity\left( \frac{\sigma^2}{\frac1k\sum_{i=1}^k n_i(\Ell)} \right) = kd\Capacity\left(\sigma^2k/n\right)\,,
\]
and the claimed result follows. \qed


We remark that to prove the bound, we \emph{did not} actually need to use the fact that the labels all have the same probability;
the calculation above shows that a balanced label distribution in fact maximizes the MI between $\Xc$ and the observations $\Y_1,\ldots,\Y_n$
(though this will not be used later).

\subsection{Proof of Lemma~\ref{lem:LowerBound:SingleSample-IMMSE}}
\label{sec:proof-lem:LowerBound:SingleSample-IMMSE}

The proof relies on the celebrated I-MMSE relation of Guo, Shamai and Verdu (see~\cite{guo2005mutual},~\cite{guo2013interplay}, and also the works~\cite{bcs08} and~\cite{bs13} that apply the I-MMSE framework for studying the MSE of estimating the transmitted codeword from the output of the AWGN channel). 

Let $\RandLabel\sim \Unif([k])$ and $\Y(s)=\X_\RandLabel+\sqrt{s}\Z$. 
Denote $I(s)=\MI(\RandLabel;\Y(s)|\Xc)=\MI(\X_\ell;\Y(s)|\Xc)$, where equality holds since, with probability one, $\X_1,\ldots,\X_k$ are all distinct. Recall that our goal is to bound $\Ent(\RandLabel|\Y(\sigma^2),\Xc)=\Ent(\RandLabel|\Xc)-I(\sigma^2)$. 

Clearly, for any $\sigma_0^2<\sigma^2$,
\begin{align*}
    \Ent(\RandLabel|\Y(\sigma^2),\Xc) - \Ent(\RandLabel|\Y(\sigma_0^2),\Xc) = I(\sigma_0^2)-I(\sigma^2) = -\int_{\sigma_0^2}^{\sigma^2} \frac{d}{ds}I(s) ds \,.
\end{align*}
Using the I-MMSE relation, Lemma~\ref{lem:I-MMSE}, applied pointwise conditioned on $\Xc$,
\begin{align*}
    \frac{d}{ds}I(s) = -\frac{1}{2s^2} \Expt\left[ \left\| \X_\ell - \Expt(\X_\ell|\Y(s),\Xc) \right\|^2 \right] \,.
\end{align*}
Since $\Expt(\X_\ell|\Y(s),\Xc)$ is the minimum MSE estimator of $\X_\ell$ from $(\Y(s),\Xc)$, it holds that for any $(\Y(s),\Xc)$-measurable random variable $\Xh=\Xh(\Y(s),\Xc)$, we have $-{\frac{d}{ds}I(s) \le \frac1{2s^2} \Expt\|\X_\ell-\Xh\|^2}$. 

Choose the optimal linear estimator (LMMSE) of $\X_\ell$ from $\Y(s)$, namely ${\Xh=\alpha(s)\Y(s)}$, $\alpha(s)=\frac{1}{1+s}$, so that $\Expt\|\X_\ell-\Xh\|^2=\frac{s}{1+s}d$. One would think, at first sight, that this upper bound should be very loose: after all, the LMMSE is {optimal} for a Gaussian signal, whereas, conditioned on $\Xc$, the distribution of $\X_\ell$ is very much non-Gaussian; it is not even continuous! Recall, however, that we are interested in applying Lemma~\ref{lem:LowerBound:SingleSample-IMMSE} when the rate $\Rate$ is \emph{above} the capacity $\Capacity(\sigma^2)$; the key intuition is that when this is the case, the joint statistics of $(\X_\RandLabel,\Y=\X_\RandLabel+\sigma\Z)$ with $\RandLabel\sim\Unif([k])$, are in some sense ``indistinguishable'' from those of a joint Gaussian distribution ${(\bm{W},\Y=\bm{W}+\sigma\Z)}$, $\bm{W}\sim\m{N}(\0,\Id)$,
corresponding to the capacity-achieving distribution of the Gaussian channel. 

Continuing the calculation,
\begin{align*}
    -\int_{\sigma_0^2}^{\sigma^2} \frac{d}{ds}I(s) ds \le \int_{\sigma_0^2}^{\sigma^2} \frac{1}{2s^2} \cdot \frac{s}{1+s} ds \cdot d = \int_{\sigma_0^2}^{\sigma^2} \left(-\Capacity'(s)\right)ds \cdot d = \left(\Capacity(\sigma_0^2)-\Capacity(\sigma^2)\right)d \,,
\end{align*}
where $\Capacity(s)=\frac12\log(1+1/s)$ is from (\ref{eq:CapacityFunc-def}) and $\Capacity'(s)$ is its derivative. Combining,
\begin{equation}\label{eq:proof-lem:LowerBound:SingleSample-IMMSE:1}
    \Ent(\RandLabel|\Y(\sigma^2),\Xc) \le \Ent(\RandLabel|\Y(\sigma_0^2),\Xc) - \Capacity(\sigma^2)d + \Capacity(\sigma_0^2)d \,.
\end{equation}

Now, set $\sigma_0^2=\Capacity^{-1}\left( (1+\delta)\Rate \right)$. To apply (\ref{eq:proof-lem:LowerBound:SingleSample-IMMSE:1}), we need to verify that $\sigma_0^2<\sigma^2$. Applying the decreasing function $\Capacity(\cdot)$, the condition is equivalent to $\Capacity(\sigma^2)<\Capacity(\sigma_0^2)=(1+\delta)\Rate$, which certainly hold since we assume $\Rate>\Capacity(\sigma^2)$.

By definition, $\Capacity(\sigma_0^2)d=(1+\delta)\log k$. 

Define by $e(\delta)$ the error (averaged over the ensemble $\Xc$) for decoding $\ell$ under $\AWGN(\sigma_0^2)$, using codebook $\Xc$. In other words, it is the error of the MAP estimator for $\ell$ given $(\Y(\sigma_0^2),\Xc)$. By Fano's inequality, Lemma~\ref{lem:Fano}, $\Ent(\RandLabel|\Y(\sigma_0^2),\Xc)\le h_b\left(e(\delta)\right) + e(\delta)\log k$. Combined with (\ref{eq:proof-lem:LowerBound:SingleSample-IMMSE:1}), we obtain the bound claimed in Lemma~\ref{lem:LowerBound:SingleSample-IMMSE}.

\qed

\subsection{Proof of Lemma~\ref{lem:LowerBound:Quantitative}}
\label{sec:proof-lem:LowerBound:Quantitative}



Before getting to the computation, we emphasize that Lemma~\ref{lem:LowerBound:SingleSample-IMMSE} may be invoked with \emph{any} other  upper bound on the ensemble average error $\ExptCodeErrAvg(\cdot)$, that could possibly be obtained through other means, e.g., by analyzing a different decoder than the one from Section~\ref{sec:Appendix:Decoding}. There is much literature devoted to computing optimal error rates for both the Gaussian i.i.d. and the spherical code ensembles, primary in the regime of \emph{positive rate}. In particular, for rates between the so-called \emph{critical rate} and capacity the exact exponential decay rate is known: $\ExptCodeErrAvg(\sigma^2) = \exp(-{E}_{SP}^*(\RateLim,\sigma^2)d +o(d))$, where ${E}_{SP}^*(\RateLim,\sigma^2)$ is the \emph{sphere-packing error exponent}. See, for example, \cite{shannon1959probability,gallager1968information} for the exact expression. In the analysis that follows, we will need bounds on the error probability in the regime $\Capacity(\sigma^2)-\RateLim=o(1)$. In fact, for the zero rate regime, the capacity itself is $o(1)$, and sometimes it decays even as $o(d^{-1/2})$. In those cases, the sphere packing error exponent is of limited use. 

Instead, we use the upper bounds on $\ExptCodeErrAvg(\cdot)$ derived in Section~\ref{sec:Appendix:Decoding}. 

As before, the analysis is divided between the positive ($\RateLim>0$) and zero ($\RateLim=0$) rate regimes.



\subsubsection{Positive Rate}

Let us work under the slightly more general regime, where $\Rate$ is either positive or decays slow enough with $d$, specifically, $\Rate = \frac{\log k}{d} = \omega( d^{-1/2})$ as $d\to\infty$. 

We apply the bound (\ref{eq:Decode:Pos:ErrorProbability}) 
with noise variance
\[
\sigma^2_0=\Capacity^{-1}((1+\delta)\Rate) \,.    
\]
The second term of (\ref{eq:Decode:Pos:ErrorProbability}) is bounded By
\[
(1+1/\sigma_0^2)^{1/2}e^{-d\left( (1+\delta)\Rate-\Rate - \frac12\log(\tau_2/\tau) \right)} = O(1)\cdot e^{-d(\delta \Rate - \frac12\log(\tau_2/\tau))} \,.
\]
Set $\tau_1=\tau_2=(1+\frac12 \Rate\delta)\tau$, so that $\log(\tau_2/\tau)\le \frac12\Rate\delta$. Thus,  
\[
    O(1)\cdot e^{-d(\delta \Rate - \frac12\log(\tau_2/\tau))}= O(1)\cdot e^{-d(\delta \Rate - \frac14\Rate\delta)} \lesssim e^{-C\Rate\delta d}\,,
\]
for some $C>0$. On the other hand, the first term of (\ref{eq:Decode:Pos:ErrorProbability}) is
\[
    e^{- \frac12(1+\sigma_0^2)(\sqrt{\tau_1/\tau}-1)^2 d} \lesssim e^{-C\Rate^2\delta^2 d} = e^{-C\delta^2\frac{(\log k)^2}{d}}\,.    
\]
Note that since $\delta=o(1)$, this term is the most significant. 

Denote $A=\frac{(\log k)^2}{d}$; recall that for $\Rate=\omega( d^{-1/2})$, $A=\omega(1)$.

In light of the estimates above, we need to choose $\delta=o(1)$ so to minimize $\delta + e^{-C\delta^2 A}$.
Take 
\[
\delta = C_1 \sqrt{\frac{\log A}{ A}}\,,
\]
for large enough constant $C_1$, which yields 
\[
\delta + e(\delta) \lesssim \delta + e^{-C\delta^2A} \lesssim \sqrt{\frac{\log A}{A}}\,.
\]
Plugging this into  (\ref{eq:LowerBound:SingleSample-Delta}),
\[
\MI(\Xc;\Y) \lesssim \log{k} \sqrt{\frac{\log A}{A}}\,.
\] 
Using (\ref{eq:LowerBound:SC-final}),
\begin{align*}
\label{eq:LoweBound:I-LB:pos}
    \frac{\SC_\eps}{\sigma^2 k} 
    &\ge C_1(\eps,\RateLim) \cdot \frac{d}{\sigma^2} \cdot ({\MI(\Xc;\Y)})^{-1} \\
    &\gtrsim \frac{d}{\sigma^2 \log k} \sqrt{\frac{A}{\log A} }
\end{align*}
where $A=\frac{(\log k)^2}{d}$. Let us understand the asymptotic of this bound as $d\to\infty$ and $\Rate=\Capacity(\RegimeParam\sigma^2)=\frac{\log k}{d}\gg d^{-1/2}$. In that case, $\log A\approx \log\log k$, and so, the above reads
\begin{align*}
    \frac{\SC_\eps}{\sigma^2 k} \gtrsim \frac{d}{\sigma^2 \log k} \sqrt{\frac{A}{\log A} } \gtrsim \frac{d}{\sigma^2 \log k} \sqrt{\frac{\frac{(\log k)^2}{d}}{\log \log k} }\,.
\end{align*}
Using $1/(\RegimeParam\sigma^2) \ge \Capacity(\RegimeParam\sigma^2) = (\log{k})/d$ (since $\Capacity(s)\le 1/(2s))$ finally yields
\begin{equation}
    \frac{\SC_\eps}{\sigma^2 k} \gtrsim \sqrt{\frac{\log k}{d}} \sqrt{\frac{\log k}{\log\log k}} \,.
\end{equation}
Finally, note that in the positive rate regime, $\frac{\log k}{d}=\Omega(1)$.

\qed

\subsubsection{Rate Zero ($\RateLim=0$)}

Assume that $\lim_{d\to\infty}\Rate=0$ (including, possibly, $\Rate\gg d^{-1/2}$).

We would like to use the bound (\ref{eq:Decode:0:ErrorProbability}) with some $\eta=\eta_1=\eta_2=o(1)$ and $\sigma^2_0=\sigma_0^2(\delta)=\Capacity^{-1}((1+\delta)\Rate)$, for $\delta=o(1)$.

We start with the condition (\ref{eq:Decode:0:Condition}), namely,
\begin{align*}
     0  
     &\le F:= 1 - \eta - \sqrt{\frac{2\log k}{d} + \frac{\eta^2}{\sigma_0^2(\delta)}} - \sqrt{\frac{2\sigma^2_0(\delta)\log k}{d}} \\
     &= 1-\eta - \sqrt{2\Rate}\sqrt{1 + \frac{1}{2\sigma_0^2(\delta)\Rate}\eta^2} - \sqrt{2\sigma_0^2(\delta)\Rate} \,.
\end{align*}
(we replace $k-1$ with $k$, which yields a stronger condition.)

Use $\Capacity(s)=\frac12\log(1+1/s)\le 1/(2s)$, therefore $\Capacity^{-1}(y)\le1/(2y)$, and so $\sigma^2_0(\delta)\le \frac{1}{2(1+\delta)\Rate}$:
\[
    \sqrt{2\sigma_0^2(\delta)\Rate} \le \sqrt{\frac{1}{1+\delta}} = 1-\frac{\delta}{2} + O(\delta^2) \,  .  
\]
Moreover,
\[
    \sqrt{1 + \frac{1}{2\sigma_0^2(\delta)\Rate}\eta^2} \le   1+  \frac{1}{4\sigma_0^2(\delta)\Rate}\eta^2
\]
($\sqrt{1+x}\le 1+\frac12 x$ for all $x\ge 0$). Thus, 
\begin{align}\label{eq:Ugly}
    F \ge \delta/2 -O(\delta^2) -\eta - \sqrt{2\Rate} - 2\sqrt{\Rate} \cdot \frac{1}{4\sigma_0^2(\delta)\Rate}\eta^2 \,.
\end{align}

Note that $\sigma_0^2(\delta)\Rate=\Theta(1)$ for any $\delta=o(1)$; to see this, recall that $(1+\delta)\Rate = \Capacity(\sigma_0^2)$ (by definition), with $\Capacity(\sigma_0^2) = 1/(2\sigma_0^2) + O(1/\sigma_0^4)$ with $\sigma_0^2\to\infty$.
Consequently, the last term of (\ref{eq:Ugly}) above is necessarily of lower order than either $\sqrt{\Rate}$ or $\eta$. 

Introduce a constant parameter $\nu\in (0,1/2)$, and set $\eta=(1/2-\nu)\delta$. Observe that whenever
\begin{equation}\label{eq:LowerBound:delta-constraint}
    \delta \ge \frac{2\sqrt{2}}{\nu}\sqrt{\Rate} = \frac{2\sqrt{2}}{\nu}\sqrt{\frac{\log k}{d}},\quad \delta=o(1)\,,        
\end{equation}
plugging into (\ref{eq:Ugly}), we have $F\ge \frac{\nu}{2}\delta(1-o(1))>0$.  

Let us estimate the terms in (\ref{eq:Decode:0:ErrorProbability}). The first term is $e^{-\frac{d}{2}F^2} \le e^{-C_1d\nu^2\delta^2}$. The second term is  $2e^{-\frac{\eta^2}{2}\cdot \frac{d}{\sigma_0^2}}$. Using
\[
    \frac{\eta^2 d}{2\sigma_0^2} \ge \eta^2d(1+\delta)\Rate = \eta^2(1+\delta)\log k \ge \eta^2\log k\,,   
\]
(we used $\sigma_0^2\le \frac{1}{2(1+\delta)\Rate}$),
we deduce that the second term is $\le 2e^{-(1/2-\nu)^2\delta^2\log k}$. Since in the zero rate regime, $d\gg \log k$, we see that the first term is always negligible  compared to the second, regardless of how fast $\delta$ decays. Thus, we would like choose $\delta=o(1)$ so to minimize (the asymptotic decay rate of)
\begin{equation}
    \label{eq:LowerBound:Balance}
    \delta + e(\delta) \lesssim \underbrace{\delta}_{e_1(\delta)} + \underbrace{e^{-(1/2-\nu)^2\delta^2\log k}}_{e_2(\delta)} \,.
\end{equation}
Note that $e_1(\delta)$ is increasing in $\delta$, while $e_2(\delta)$ is decreasing. Denote
\begin{equation}
    \delta_1 = \frac{\sqrt{2}}{(1/2-\nu)}\sqrt{ \frac{\log\log k}{\log k} },\quad \delta_2 = \frac{2\sqrt{2}}{\nu}\sqrt{\frac{\log k}{d}} \,,
\end{equation}
so that $\delta_2$ is the smallest number $\delta$ that satisfies (\ref{eq:LowerBound:delta-constraint}).

One may readily verify that $\delta=\delta_1$ optimally balances between $e_1(\delta),e_2(\delta)$, in the sense of asymptotic growth:
\[
    e_1(\delta_1) \asymp e_2(\delta_1) \asymp \sqrt{\frac{\log\log k}{\log k}    } \,,\quad \implies e_1(\delta_1)+e_2(\delta_1) \lesssim \sqrt{\frac{\log\log k }{\log k}}
\] 
Recall, however, that not all assignments $\delta$ are applicable; we must satisfy the constraint (\ref{eq:LowerBound:delta-constraint}), $\delta\ge \delta_2$. If $\delta_2\le \delta_1$ then there is no problem; on the other hand, if $\delta_2>\delta_1$, assigning $\delta=\delta_2$,
\[
e_1(\delta_2) + e_2(\delta_2) \overset{(i)}{\le} e_1(\delta_2) + e_2(\delta_1) \overset{(ii)}{\lesssim} e_1(\delta_2) + e_1(\delta_1) \overset{(iii)}{\le} 2e_1(\delta_2)  \lesssim \sqrt{\frac{\log k}{d} } \,,
\]
where we used that: (i) $e_2(\cdot)$ is decreasing; (ii) $e_1(\delta_1)\asymp e_2(\delta_1)$; (iii) $e_1(\cdot)$ is increasing. 

Concluding the calculation, using (\ref{eq:LowerBound:SingleSample-Delta}), we have 
\begin{equation}
    \MI(\Xc;\Y) 
    \lesssim \log k \cdot\max\left\{ \sqrt{\frac{\log\log k}{\log k}}, \sqrt{\frac{\log k}{d}} \right\} \,.
\end{equation}
Finally, to deduce the lower bound on the sample complexity, use (\ref{eq:LowerBound:SC-final}):
\begin{align}
    \frac{\SC_\eps}{\sigma^2 k} 
    &\ge C_1(\eps) \cdot \frac{d}{\sigma^2} \cdot ({\MI(\Xc;\Y)})^{-1}        \nonumber\\
    &\ge C_1(\eps) \cdot 2\RegimeParam\log k  \cdot ({\MI(\Xc;\Y)})^{-1}        \nonumber\\
    &\ge C_2(\eps)\RegimeParam \min\left\{ \sqrt{\frac{\log k}{\log \log k}}, \sqrt{\frac{d}{\log k}} \right\} \,.
    \label{eq:LoweBound:I-LB:0}
\end{align}

\qed

\newpage

\section{Proofs for Section~\ref{sec:UpperBound-StepI}}

\subsection{Proof of Lemma~\ref{lem:UpperBound:Technical}}
\label{sec:proof-lem:UpperBound:Technical}

As in Section~\ref{sec:Appendix:Decoding}, we give different constructions between the zero rate ($\RateLim=0$) and positive rate ($\RateLim>0$) regimes.
The construction for the local test is guided by the form of the capacity-achieving decoder from Section~\ref{sec:Appendix:Decoding}.

\subsubsection{Rate Zero ($\RateLim=0$)}

Following the form of the decoder analyzed in Section~\ref{sec:Appendix:Decoding:0}, we consider a test of the form
\begin{equation}
    \label{eq:Appendix:Test:RateZero}
    \Test(\Xh,\Y) = \Ind\{d^{-1}\langle \Y,\Xh\rangle \ge 1-\eta\}\,,
\end{equation}
where the choice of $\eta$ will be specified below. 

Suppose that $\Xh\in\sqrt{d}\SphereD$ is such that, for some particular $i\in [k]$, $d^{-1}\|\Xh-\X_i\|^2\le 0.5\epsI$. Note that this may be written equivalently as $d^{-1}\langle \X_i,\Xh\rangle\ge 1-0.25\epsI $. Thus, 
\[
    d^{-1}\langle \X_i+\sigma\Z , \Xh\rangle \ge 1-0.25\epsI + \underbrace{(d^{-1/2}\sigma)\langle \Z, d^{-1/2}\Xh\rangle}_{\sim \m{N}(0, \sigma^2/d)}\,.    
\]
Setting
\begin{equation}\label{eq:Appendix:Test:RateZero:eta}
    \eta=0.25\epsI\,,
\end{equation}
we get
\begin{align*}
    \Pr(\Test(\Xh,\Y)=1)
    &\ge \frac1k \Pr(\Test(\Xh,\Y)=1\,|\,\RandLabel=i) \\
    &\ge \frac1k\Pr\left(\m{N}(1-0.25\epsI, \sigma^2/d)\ge 1-0.25\epsI \right) = 0.5/k \,.
\end{align*}
Consequently, with probability $1$, $\qClose(\Xc)\ge 0.5k^{-1}$. 

The challenging part of the analysis is to control $\qFar(\Xc)$. 

Observe that if $d^{-1}\|\Xh-\X_i\|^2\ge \epsI$ then $d^{-1}\langle \X_i,\Xh\rangle\le 1-0.5\epsI$. 
For $\Xh\in\sqrt{d}\SphereD$, denote
\begin{equation}
    \label{eq:UpperBound:0:Q-def}
    \begin{split}
        &Q_i(\Xh|\Xc) = \Pr_{\Z\sim \m{N}(\0,\Id)} \left( \Test(\Xh,\X_i+\sigma\Z) =1 \,\big|\,\Xc\right) \,,\\
        &\bar{Q}(\Xh|\Xc) = \frac{1}{k}\sum_{i=1}^k Q_i(\Xh|\Xc) \,.
    \end{split}
\end{equation}
Note that $Q_i(\Xh|\Xc)$ depends on $\Xc$ only through $\X_i$. 

By definition, $\qFar(\Xc)=\max_{\Xh\in\Net\cap\HFar} \bar{Q}(\Xh|\Xc)$. We start with a trivial bound.
\begin{lemma}\label{lem:UpperBound:0:Q-infty}
    Suppose that $d^{-1}\langle \X_i,\Xh\rangle \le 1-\nu-0.25\epsI$ for $0\le \nu \le 1-0.25\epsI$. Then
    \begin{align*}
        Q_i(\Xh|\Xc) \le  k^{-\RegimeParam\nu^2} \,.
    \end{align*}
    Consequently, if $\Xh\in\HFar$ then for all $i$, $Q_i(\Xh|\Xc)\le k^{-\frac{\RegimeParam}{16}\epsI^2}$.
\end{lemma}
\begin{proof}
    $d^{-1}\langle \X_i + \sigma\Z , \Xh\rangle \le 1-0.25\epsI - \nu + \underbrace{d^{-1}\sigma\langle \Z,\Xh\rangle}_{\sim  \m{N}(0,\sigma^2/d)}$. Thus,
    \begin{align*}
        \Pr(\Test(\Xh,\X_i+\sigma\Z)=1) \le \Pr\left(\m{N}(0,\sigma^2/d)\ge \nu\right) \le e^{-\frac{d}{2\sigma^2}\nu^2} \,.
    \end{align*}
    Now, $k=e^{d\Capacity(\RegimeParam\sigma^2)}\le e^{\frac{d}{2\RegimeParam\sigma^2}}$ (since $\Capacity(s)=\frac12\log(1+1/s)\le1/(2s)$), therefore $e^{-\frac{d}{2\sigma^2}\nu^2} \le k^{-\RegimeParam\nu^2}$. 
    Finally, if $\Xh\in\HFar$, then $d^{-1}\langle \X_i,\Xh\rangle \le 1-\nu-0.25\epsI$ with $\nu=0.25\epsI$. 
\end{proof}

As mentioned, Lemma~\ref{lem:UpperBound:0:Q-infty} gives us the trivial bound $\bar{Q}(\Xh|\Xc)\le k^{-\frac{\RegimeParam}{16}\epsI^2}$ for all $\Xh\in \HFar$. This is a highly wasteful bound: it treats $\Xh$ as if it is \emph{simultaneously} $\sqrt{\epsI d}$-close to all of $\X_1,\ldots,\X_k$. 
In practice, however, ``typical'' instances of $\Xc$ create constellations that do not cluster around any particular point; consequently, for most $i\in [k]$, it has to be that, in fact, $d^{-1}\langle \X_i,\Xh\rangle \approx 0$. 

For $t\in (0,1)$, set
\begin{equation}\label{lem:UpperBound:0:Eq1}
    N_t(\Xh|\Xc) = \sum_{i=1}^k \Ind\{d^{-1}\langle \X_i,\Xh\rangle \ge t\} \,,
\end{equation}
the number of centers $\X_i$ that have correlation $\ge t$ with $\Xh$. 

Choose some  constants $\eps_0,\nu_0\in (0,1)$ such that $\nu_0 < 1-0.25\eps_0$ and $\RegimeParam\nu_0^2>1$. This can certainly be done, since $\RegimeParam>1$.  
By Lemma~\ref{lem:UpperBound:0:Q-infty} above, for every $\Xh\in \HFar\cap \Net$, assuming $\epsI\le \eps_0$,
\begin{equation}
    \bar{Q}(\Xh|\Xc) \le k^{\RegimeParam\nu_0^2} + \frac1k \cdot N_{(1-\nu_0-0.25\eps_0)}(\Xh|\Xc) \cdot k^{-\frac{\RegimeParam}{16}\epsI^2} \,.
\end{equation}
That is, $\X_i$-s whose correlation with $\Xh$ is $<1-\nu_0-0.25\eps_0<1-\nu_0-0.25\epsI$ contribute each at most $Q_i(\Xh|\Xc)\le k^{-\beta\nu_0^2}$; on the other hand, centers whose correlations is higher give, at most, the worst-case contribution $Q_i(\Xh|\Xc)=k^{-\frac{\RegimeParam}{16}\epsI^2}$. In light of (\ref{lem:UpperBound:0:Eq1}), clearly,
\begin{equation}
    \label{lem:UpperBound:0:Eq2}
    \qFar(\Xc) \le  k^{\RegimeParam\nu_0^2} + \max_{\Xh\in\Net\cap\HFar} N_{(1-\nu_0-0.25\eps_0)}(\Xh|\Xc) \cdot k^{-1-\frac{\RegimeParam}{16}\epsI^2}
\end{equation} 
Thus, it remains to show that, with high probability, $\max_{\Xh\in\Net\cap\HFar} N_{(1-\nu_0-0.25\eps_0)}(\Xh|\Xc)$ is small.

\begin{lemma}\label{lem:UpperBound:0:Nt}
    Fix any $\Xh\in \Net$. 
    There are universal $C_1,C_2$ such that whenever $t\ge C_1\sqrt{\frac{\log k}{d}}$, for all $M\ge 1$,
    \[
    \Pr\left( N_t(\Xh|\Xc) \ge M  \,\big|\,\Xh\in\HFar\right)  \le (C_2ke^{-d\frac{t^2}{2}})^M \,,
    \]
    where the probability is with respect to $\Xc\sim \Unif(\sqrt{d}\SphereD)^{\otimes k}$, and conditioned on the event that $\Xh\in\HFar$. 
\end{lemma}
\begin{proof}
    Observe that conditioned on the event $\Xh\in \HFar$, the centers $\X_1,\ldots,\X_k$ are i.i.d. and $\sim \Unif(\sqrt{d}\SphereD\setminus \m{B}(\Xh,\sqrt{\epsI d}))$. For any \emph{non-negative} $f(\cdot)$,
    \begin{equation}\label{eq:ChangeOfMeasureSphere}
        \begin{split}
            \Expt_{\X_i\sim \Unif(\sqrt{d}\SphereD\setminus \m{B}(\Xh,\sqrt{\epsI d}))} \left[ f(\X_i) \right]
            &\le \frac{\Surf(\sqrt{d}\SphereD)}{\Surf(\sqrt{d}\SphereD\setminus \m{B}(\Xh,\sqrt{\epsI d}))} \cdot \Expt_{\X_i\sim \Unif(\sqrt{d}\SphereD)} \left[ f(\X_i) \right] \\
            &= \frac{1}{1-\epsI^{d-1}} \cdot \Expt_{\X_i\sim \Unif(\sqrt{d}\SphereD)} \left[ f(\X_i) \right] \\
            &= (1+o(1))\Expt_{\X_i\sim \Unif(\sqrt{d}\SphereD)} \left[ f(\X_i) \right] \,.
        \end{split}
    \end{equation}
    Consequently, $N_t(\Xh,\Xc)\sim\mathrm{Binomial}(n,\mathsf{p})$, with 
    \begin{align*}
        \mathsf{p} 
        &= \Expt_{\X_i\sim \Unif(\sqrt{d}\SphereD\setminus \m{B}(\Xh,\sqrt{\epsI d}))} \left[ \Ind\{d^{-1}\langle \X_i,\Xh\rangle \ge t\} \right] \\
        &= (1+o(1)) \Pr_{\X_i\sim \Unif(\sqrt{d}\SphereD)} \left( d^{-1}\langle \X_i,\Xh\rangle \ge t \right)\\
        &\le 2e^{-d\frac{t^2}{2}} \,,
    \end{align*}
    where we used the standard tail bound Lemma~\ref{lem:Unif-Sphere}. We have
    \begin{align*}
        \Pr\left( X_t(\Xh|\Xc) \ge M \,\big|\,\Xh\in\HFar \right) = \sum_{l=M}^k \binom{k}{l}\mathsf{p}^l (1-\mathsf{p})^{k-l} \le \sum_{l=M}^k \left( \frac{ke}{l}\mathsf{p} \right)^l \,.
    \end{align*}
    Assuming $t\ge C\sqrt{\frac{\log k}{d}}$ for large enough (universal) $C>0$, $\frac{ke}{l}\mathsf{p}\le \frac{1}{2}$,and so  ${\sum_{l=M}^k \left( \frac{ke}{l}\mathsf{p} \right)^l \le 2\left( {ke}\mathsf{p} \right)^M}$. 
\end{proof}

\begin{lemma}\label{lem:UpperBound:0:B}
    There are $C_1,C_2,C_3$ universal such that the following holds. 

    Suppose that $t\ge C_1\sqrt{\frac{\log k}{d}}$, $\epsI<1/2$. Then with probability at least $1-e^{-d}$ over ${\Xc\sim\Unif(\sqrt{d}\SphereD)^{\otimes k}}$,
    \[
    \max_{\Xh\in\Net\cap\HFar} N_t(\Xh|\Xc) \le C_2\frac{\log(1/\epsI)}{\frac12 t^2 - \frac{\log k + C_3}{d}} =: M_0    \,.
    \]
\end{lemma}
\begin{proof}
    Let $B$ be the number of candidates $\Xh\in\Net$ such that $N_t(\Xh|\Xc)>M_0$. Our goal is to show that w.p. $\ge 1-e^{-d}$,   $B=0$. By Markov's inequality and Lemma~\ref{lem:UpperBound:0:Nt},  assuming $t\gtrsim \sqrt{\frac{\log k}{d}}$,
    \begin{align*}
        \Pr(B\ge 1) \le \Expt[B] = (Cke^{-d\frac{t^2}{2}})^{M_0}|\Net| \le (Cke^{-d\frac{t^2}{2}})^{M_0} e^{Cd\log(1/\epsI)}\,.
    \end{align*}
    Taking $M_0=C_2\frac{\log(1/\epsI)}{\frac12 t^2 - \frac{\log k + C_3}{d}}$ for large enough $C_2,C_3$, 
    the above probability is $\le e^{-d}$.
\end{proof}

\begin{proof}
    (Of Lemma~\ref{lem:UpperBound:Technical}, $\RateLim=0$.) Recall that, by construction, $\qClose(\Xc)\ge 0.5k^{-1}$ holds with probability $1$.
    
    Recall (\ref{lem:UpperBound:0:Eq2}), and apply Lemma~\ref{lem:UpperBound:0:B} with $t=1-\nu_0-0.25\eps_0$, which is a positive constant. Consequently, with probability $\ge 1-e^{-d}$, 
    \[
    \qFar(\Xc) \le k^{-\RegimeParam\nu_0^2} + C\log(1/\epsI)k^{-1-\frac{\RegimeParam}{16}\epsI^2} \,.
    \]
    Since $\RegimeParam\nu_0^2>1$, one may indeed choose some $c=c(\RegimeParam)$ small enough such that 
    $\qFar(\Xc)\le 2k^{-1-c\epsI^2}$ holds with probability $1-o_{\RegimeParam,\epsI}(1)$. 

\end{proof}

\subsubsection{Positive Rate ($\RateLim>0$)}

Moving on to the positive rate regime, our construction is guided by the decoder of Section~\ref{sec:Appendix:Decoding:Pos}.

Set $\alpha=\frac{1}{1+\sigma^2}$, $\tau=\frac{\sigma^2}{1+\sigma^2}$.

For fixed $\Xh,\X_i\in\sqrt{d}\SphereD$, denoting $\Y=\X_i+\sigma\Z$, $\Z\sim\m{N}(\0,\Id)$,
\begin{align}
    \Expt_{\Z}\|\alpha\Y - \Xh\|^2  
    &= \Expt_{\Z}\|\alpha\Y - \X_i\|^2 + \|\X_i-\Xh\|^2 + 2\Expt_{Z}\langle \alpha\Y - \X_i, \X_i-\Xh\rangle \nonumber\\
    &= \tau d + \|\X_i-\Xh\|^2 + 2(\alpha-1)\langle \X_i,\X_i-\Xh\rangle \nonumber \\
    &= \tau d + \alpha\|\X_i-\Xh\|^2\,,\label{eq:UpperBound:pos:Expt}
\end{align}
where we used $\|\X_i\|^2=\|\Xh\|^2=d$ and $-\langle \X_i,\Xh\rangle = \frac12(\|\X_i-\Xh\|^2-\|\X_i\|^2-\|\Xh\|^2)$. 

Assume that $\|\X_i-\Xh\|^2\le 0.5\epsI d$. 
By the Gaussian Lipschitz concentration inequality (Lemma~\ref{lem:Gaussian-Lip}), applied for $F(\Z)=d^{-1/2}\|\alpha(\X_i+\sigma\Z) - \Xh\|$, which is $(d^{-1/2}\alpha \sigma)$-Lipschitz with expectation $\Expt_{\Z}F(\Z)\le \sqrt{\Expt(F(\Z))^2}
\le \sqrt{\tau + 0.5\alpha\eps_0}$, 
\begin{align*}
    \Pr_{\Z}\left(d^{-1/2}\|\alpha\Y_i - \Xh\| \ge \sqrt{\tau+0.5\alpha \epsI} +\eta\right) \le e^{-\frac12\frac{\eta^2}{\alpha^2\sigma^2}d} \,.
\end{align*}
Consider the test
\begin{equation}\label{eq:UpperBound:pos:Test}
    \Test(\Xh,\Y) = \Ind\{d^{-1/2}\|\alpha\Y-\Xh\|\le \sqrt{{\tau}+0.5\alpha \eps_I} + \eta  \},\quad \eta=\sqrt{\frac{2\alpha^2\sigma^2\log 2}{d}} = O_{\RateLim,\RegimeParam}(d^{-1/2})\,,
\end{equation}
so that by construction, $\qClose(\Xc)\ge 0.5k^{-1}$ holds with probability $1$.

It remains to bound $\qFar(\Xc)$ with high probability. We follow the notation (\ref{eq:UpperBound:0:Q-def}), where $\Test(\Xh,\Y)$ that appears in (\ref{eq:UpperBound:0:Q-def})  is now defined by (\ref{eq:UpperBound:pos:Test}). 
Our goal is to bound, with high probability over $\Xc$,
\[
\qFar(\Xc) = \max_{\Xh\in\Net\cap\HFar} \bar{Q}_i(\Xh|\Xc) =   \max_{\Xh\in\Net\cap\HFar} \frac1k\sum_{i=1}^k  Q_i(\Xh|\Xc) \,. 
\]
As was in the zero rate case, for any fixed $\Xh$, conditioned on the event $\{\Xh\in\HFar\}$, the centers $\X_1,\ldots,\X_k$ are i.i.d. and $\X_i\sim \Unif(\sqrt{d}\SphereD\setminus\m{B}(\Xh,\sqrt{\epsI d}))$. Consequently, $\bar{Q}(\Xh|\Xc)$ is the average of $k$ i.i.d. random variables. We shall show that its expectation is very small, specifically $\Expt[\bar{Q}(\Xh|\Xc)]\le k^{-1-c}$; moreover, we shall show that it concentrates tightly about this expectation, to the extent that the maximum over the net, $\max_{\Xh\in\Net\cap\HFar} \bar{Q}_i(\Xh|\Xc)$, can be controlled as well. To do this, we use Bernstein's inequality for sums of i.i.d. bounded random variables, Lemma~\ref{lem:BernsteinBounded}. 

For brevity, we introduce some notation. For $\Xh\in \Net$ fixed, denote by $\m{E}=\m{E}(\Xh)$ the event $\m{E}=\{\Xh\in\HFar\}$ (with respect to the probability on $\Xc\sim \Unif(\sqrt{d}\SphereD)^{\otimes k}$). Denote by $\Expt^{\m{E}}[\cdot],\|\cdot\|_{\infty}^{\m{E}}$ respectively the expectation and $L_\infty$ norm with respect to the conditional measure on $\Xc$; and $\Prob^{\m{E}}(S):=\Expt^{\m{E}}[\Ind_{S}]$. 

To use Bernstein's inequality, we need two components: an $L_\infty$ bound and a bound on the expectation. 

We start with the $L_{\infty}$ bound:
\begin{lemma}\label{lem:UpperBound:pos:Linf}
    There are $C,\eps_0$, that depend on $\RateLim,\RegimeParam$, such that the following holds.
    
    For any $\Xh\in\Net$,
    whenever $\epsI<\eps_0$ and $d$ is sufficiently large, $\epsI=\Omega_{\RateLim,\RegimeParam}(d^{-1/2})$, then 
    \begin{equation}\label{eq:UpperBound:pos:Linf}
        \|Q_i(\Xh|\Xc)\|_\infty^{\m{E}} \le 2k^{-C\eps_I^2} \,.
    \end{equation}
\end{lemma}
\begin{proof}
    Fix any $\Xh,\Xc$ such that $\Xh\in\HFar$. Let $\mu=d^{-1/2}\|\alpha\X_i-\Xh\|\le 2$. By the rotational invariance of $\Z \sim \m{N}(\0,\Id)$,\footnote{An alternative method to the one below (which is itself very standard) is to use deviation inequalities for non-central $\chi^2$ random variables, that are readily available in the literature, though somewhat more ``messy''.}
    \[
        \|\alpha(\X_i+\sigma\Z)-\Xh\|^2 \overset{d}{=} \|\alpha\sigma\Z + \mu\oneVec\|^2 = \sum_{j=1}^d (\alpha\sigma Z_j + \mu)^2 \,,
    \]
    where $\oneVec=(1,\ldots,1)$ and $\overset{d}{=}$ denotes equality in distribution. 
    The expression above is a sum of i.i.d. sub-Exponential random variables. The sub-Exponential norm, denoted $\|\cdot\|_{\psi_1}$, is upper bounded by
    \[
        \|(\alpha\sigma Z_j+\mu)^2\|_{\psi_1} \lesssim \alpha^2\sigma^2+\mu^2 = O_{\RegimeParam,\RateLim}(1) \,.
    \]
For background on sub-Exponential random variables, including the definition of the sub-Exponential norm (and Orlicz norms in general), we refer to the book \cite[Chapter 2]{vershynin2018high}. Recall by (\ref{eq:UpperBound:pos:Expt}) that 
\[
    d^{-1}\Expt \|\alpha\sigma\Z + \mu\oneVec\|^2 = \tau  + \alpha d^{-1}\|\X_i-\Xh\|^2 \ge \tau  + \alpha\epsI     \,,
\]
and set
\begin{align*}
    t 
    &= \tau + \alpha\epsI - \left(\sqrt{{\tau}+0.5\alpha \eps_I} + \eta\right)^2 \\
    &=     \tau + \alpha\epsI - \left(\sqrt{{\tau}+0.5\alpha \eps_I} + O_{\RateLim,\RegimeParam}(d^{-1/2})\right)^2 \\
    &= 0.5\alpha\epsI + O_{\RateLim,\RegimeParam}(d^{-1/2}) \,.
\end{align*}
Using Bernstein's inequality for sub-Exponential random variables, Lemma~\ref{lem:Bernstein}, 
    \begin{align*}
        Q_i(\Xh|\Xc)
        &=
        \Pr_{\Z}(\Test(\Xh,\X_i+\sigma\Z)=1)\\
        &=
        \Pr\left(d^{-1}\|\alpha(\X_i+\sigma\Z)-\Xh\|^2 \le \left(\sqrt{{\tau}+0.5\alpha \eps_I} + \eta\right)^2 \right)\\
        &\le
        \Pr\left(d^{-1}\|\alpha(\X_i+\sigma\Z)-\Xh\|^2 - \Expt_{\Z} d^{-1}\Expt_{Z}\|\alpha(\X_i+\sigma\Z)-\Xh\|^2 \le -t \right) \\
        &\le 2\exp\left( -c\min\left\{ \frac{(d \cdot t)^2}{d\cdot \|\alpha\sigma Z_j+\mu\|_{\psi_1}^2}, \frac{d\cdot t}{\|\alpha\sigma Z_j+\mu\|_{\psi_1}} \right\} \right) \\
        &\overset{(\star)}{\le} 2 \exp\left( -C(\RateLim,\RegimeParam)\min\{\epsI,\epsI^2\}\cdot d \right)\\
        &= 2 \exp\left( -C(\RateLim,\RegimeParam)\epsI^2\cdot d \right)
    \end{align*}
    where to get ($\star$), we used $\|\alpha\sigma Z_j+\mu\|_{\psi_1}=O_{\RegimeParam,\RateLim}(1)$ and $\epsI\gtrsim d^{-1/2}$. Finally, to deduce (\ref{eq:UpperBound:pos:Linf}), recall that $k=e^{\RateLim \cdot d}$.
\end{proof}

Moving on to the expectation:
\begin{lemma}\label{lem:UpperBound:pos:Expt}
    For any $\epsI<\eps_0$ sufficiently small and $d$ sufficiently large such that $\epsI=\Omega_{\RateLim,\RegimeParam}(d^{-1/2})$,
    \[
    \Expt^{\m{E}}[Q_i(\Xh|\Xc)] \le 2k^{-1-c}\,,
    \]
    where $c,\eps_0$ depend on $\RateLim,\RegimeParam$. 
\end{lemma}
\begin{proof}
    Observe that the test, defined in (\ref{eq:UpperBound:pos:Test}), is an orthogonally invariant function of its argument; that is, $\Test(\x,\y)=\Test(R\x,R\y)$ for any $R\in O(d)$, where $O(d)$ is the group of $d\times d$ orthogonal matrices. Introduce an independent $R\sim \mathrm{Haar}(O(d))$, and note that, by the orthogonal invariance of $\Z\sim\m{N}(\0,\Id)$, $R(\X_i+\sigma\Z)\overset{d}{=} R\X_i + \sigma\Z$. 

    Now, conditioned on $R\X_i$, the conditional distribution of $R\Xh$ is $\sim \Unif(\sqrt{d}\SphereD\setminus\m{B}(R\X_i,\sqrt{\epsI d}))$. Thus, again owing to orthogonal invariance, we have
    \begin{align*}
        \Expt^{\m{E}}[Q_i(\Xh|\Xc)] 
        &= \Pr_{\substack{\W\sim \Unif(\sqrt{d}\SphereD\setminus\m{B}(\x,\sqrt{\epsI d})),\\ \Z\sim\m{N}(\0,\Id)}}\left(\|\alpha(\x+\sigma\Z)-\W \| \le \sqrt{{\tau}+0.5\eps_I} + \eta\right) \\
        &\overset{(\star)}{\le} \frac{1}{1-\epsI^{d-1}} \Pr_{\substack{\W\sim \Unif(\sqrt{d}\SphereD),\\ \Z\sim\m{N}(\0,\Id)}}\Pr\left(\|\alpha(\x+\sigma\Z)-\W \| \le \sqrt{{\tau}+0.5\eps_I} + \eta\right)\,,
    \end{align*}
    where $\x\in\sqrt{d}\SphereD$ is any fixed vector, and $(\star)$ follows from (\ref{eq:ChangeOfMeasureSphere}).  The probability above has been bounded in Lemma~\ref{lem:Decode:Pos:2}, which yields
    \begin{align*}
        \le (1+o(1))\left( 1+\frac{1}{\sigma^2} \right)^{1/2}e^{-\left(\Capacity(\sigma^2) - \frac12\log\frac{(\sqrt{{\tau}+0.5\eps_I} + \eta)^2}{\tau} \right)d}\,.
    \end{align*}
    Since $k=e^{\Capacity(\RegimeParam\sigma^2)d}$ for $\RegimeParam>1$ constant (hence $\Capacity(\sigma^2)-\Capacity(\RegimeParam\sigma^2)$ is a positive constant), and assuming $\epsI\le\eps_0$ is small enough, we get that $\Expt^{\m{E}}[Q_i(\Xh|\Xc)] \le 2k^{-1-c}$ for some $c=c(\RegimeParam,\RateLim)$, for $d$ large enough.
\end{proof}

We are ready to bound $\qFar(\Xc)$:
\begin{lemma}\label{lem:UpperBound:pos:final}
    There are $C,\eps_0$ that depend on $\RateLim,\RegimeParam$, such that whenever $\epsI<\eps_0$ then
    \[
    \qFar(\Xc)\le 3k^{-1-C\epsI^2}    
    \] 
    holds with probability $1-o_{\epsI,\RateLim,\RegimeParam}(1)$ over $\Xc\sim\Unif(\sqrt{d}\SphereD)^{\otimes k}$. 
\end{lemma}
\begin{proof}
    Fix $\Xh\in\Net$. We start by showing that conditioned on $\Xh\in\HFar$, $\bar{Q}(\Xh|\Xc)$ is very small with high probability; to that end, we shall use Bernstein's inequality, Lemma~\ref{lem:BernsteinBounded}. 

    By Lemma~\ref{lem:UpperBound:pos:Linf}, $\|Q_i(\Xc|\Xh)\|_{\infty}^{\m{E}}\le 2k^{-C_1\eps_I^2}$. By Lemma~\ref{lem:UpperBound:pos:Expt}, $\Expt^{\m{E}}[Q_i(\Xh|\Xc)]\le 2k^{-1-C_2}$. 
    Consequently, for small enough $\eps_I$, $\|Q_i(\Xc|\Xh)-\Expt^{\m{E}}[Q_i(\Xh|\Xc)]\|_{\infty}^{\m{E}} \le 4k^{-C_1\epsI^2}$. 
    Note moroever that since $Q_i(\Xh|\Xc)\ge0$,
    \[
        \Var^{\m{E}}(Q_i(\Xc|\Xh)) \le \Expt^{\m{E}}((Q_i(\Xc|\Xh)^2) \le  \|Q_i(\Xc|\Xh)\|_{\infty}^{\m{E}}\Expt^{\m{E}}[Q_i(\Xh|\Xc)]\le 4k^{-1-C_2-C_1\epsI^2} \,.    
    \] 
    Since $Q_1(\Xh|\Xc),\ldots,Q_k(\Xh|\Xc)$ are i.i.d. conditioned on $\m{E}$, by Bernstein's inequality for sums of independent bounded random variables, for some universal $c$,
    \begin{align*}
        \Prob^{\m{E}}(\bar{Q}(\Xh|\Xc)\ge t + 2k^{-1-C_2}) \le 2\exp\left(-c k\min\left\{ \frac{t^2}{\Var^{\m{E}}(Q_i(\Xc|\Xh))}, \frac{t}{\|Q_i(\Xc|\Xh)-\Expt^{\m{E}}[Q_i(\Xh|\Xc)]\|_{\infty}^{\m{E}}} \right\}\right)\,.
    \end{align*}
    Setting $t=k^{1-C_1\epsI^2/2}$ and assuming $\epsI$ is small enough, for $c$ (perhaps other) universal,
    \begin{align*}
        \Prob^{\m{E}}(\bar{Q}(\Xh|\Xc)\ge 3k^{-1-C_1\epsI^2/2} ) 
        \le \exp\left(-ck^{C_1\epsI^2/2}\right) \,.
    \end{align*} 

    So far we have shown that with overwhelming probability over the configuration $\Xc$, conditioned on $\Xh\in \HFar$, $\bar{Q}(\Xh|\Xc)$ is very small.
We now wish to control $\qFar(\Xc)=\max_{\Xh\in\Net\cap\HFar}\bar{Q}(\Xh|\Xc)$. 

    Let $N=\sum_{\Xh\in\Net} \Ind\{\Xh\in\HFar\textrm{ and }\bar{Q}(\Xh|\Xc)\ge 3k^{-1-C_1\epsI^2/2}\}$. Of course, $N=0$ implies that $\qFar(\Xc)\le 3k^{-1-C_1\epsI^2/2}$. By Markov's inequality,
    \[
    \Pr(N\ge 1) \le \Expt[N] \le |\Net|  \Prob^{\m{E}}(\bar{Q}(\Xh|\Xc)\ge 3k^{-1-C_1\epsI^2/2} ) \le e^{Cd\log(1/\epsI)} \exp\left(-ck^{C_1\epsI^2/2}\right)\,.
    \]
    Recall that $k$ is exponential in $d$; consequently, for $\epsI\ge C_2 \sqrt{\frac{\log d}{d}}$ with large enough $C_2$ (in particular, whenever $\epsI>0$ is constant),
    \[
        e^{Cd\log(1/\epsI)} \exp\left(-ck^{C_1\epsI^2/2}\right) 
        \le e^{-Cd\log d} \exp(-cd^2)
         = o_{\epsI,\RegimeParam,\RateLim}(1) \,.  
    \]
\end{proof}

\begin{proof}
    (Of Lemma~\ref{lem:UpperBound:Technical}). 
    Follows immediately from Lemma~\ref{lem:UpperBound:pos:final}, where note that since $\qFar(\Xc)\le 1$, we may change the prefactor $3$ in Lemma~\ref{lem:UpperBound:pos:final} to whatever number $>1$ we like (at the expense of changing the exponent). We do so for convenience.
\end{proof}

\subsection{Proof of Lemma~\ref{lem:UpperBound:TechnicalCorr}}
\label{sec:proof-lem:UpperBound:TechnicalCorr}

We start by showing that w.h.p., $\NetClose\cap \HFar=\emptyset$, namely, we do not retain candidates which are $\sqrt{\epsI d}$-far from all centers $\X_i$.

Let $\Xh\in\Net\cap \HFar$.
Conditioned on the event of Lemma~\ref{lem:UpperBound:Technical}, the random variables $\{\Test(\Xh,\Y)\}$ are an i.i.d. sequence of $N$ Bernoulli trials, with success probability $\le 2k^{-1-c\epsI^2}$. By Chernoff's inequality (Lemma~\ref{lem:Chernoff}) and the estimate of Lemma~\ref{lem:KL-binary}, 
\begin{align*}
    \Pr (\Xh\in \NetClose) \le e^{-C_1N\epsI^2 k^{-1}\log k} \,,
\end{align*}
for some $C_1=C_1(\RateLim,\RegimeParam)$, and assuming $k$ is large enough. Taking a union bound,
\begin{align*}
    \Pr(\NetClose\cap\HFar \ne \emptyset) \le |\HFar\cap\Net|e^{-C_1N\epsI^2 k^{-1}\log k} \le e^{Cd\log(1/\epsI)}e^{-C_1N\epsI^2 k^{-1}\log k} \,.
\end{align*}
Observe that this is $o_{\RateLim,\RegimeParam}(1)$ whenever $N\ge C_2 \frac{\log(1/\epsI)}{\epsI^2}k\frac{d}{\log k}$, for some $C_2=C_2(\RateLim,\RegimeParam)$ large enough. Now, when $\RateLim>0$ then $\frac{d}{\log k}=1/\Rate=1/\RateLim$ is just a constant, and so is $\sigma^2$; so for a suitably modified $C_3=C_3(\RateLim,\RegimeParam)$, $N\ge C_3 \frac{\log(1/\epsI)}{\epsI^2}k\sigma^2$ suffices.  As for the case $\RateLim=0$, $\Rate=\Capacity(\RegimeParam\sigma^2)=\frac{1}{\RegimeParam\sigma^2} + O(\sigma^{-4})$, and so $\frac{d}{\log k} \lesssim \sigma^2$; consequently, $N\ge C_3 \frac{\log(1/\epsI)}{\epsI^2}k\sigma^2$ suffices. This show the first claim of the Lemma.

Moving on, we need to show that for many $i\in [k]$, $\NetClose$ indeed contains a vector within $\sqrt{\epsI d}$-distance to $\X_i$. 

Fix any $\Xh_1,\ldots,\Xh_k\in\Net$ such that $\|\X_i-\Xh_i\|^2\le0.5\epsI d$. Since $\Net$ is an $\sqrt{0.5\epsI d}$-net of $\sqrt{d}\SphereD$, there certainly are such vectors in $\Net$. We shall show that ${|\NetClose\cap\{\Xh_1,\ldots,\Xh_k\}|\ge(1-\varphi)k}$ holds with the claimed probability; this clearly suffices. By the properties of the test, Lemma~\ref{lem:UpperBound:Technical}, $\Pr(\Test(\Xh_i,\Y)=1)\ge 0.5k^{-1}$. By Chernoff's inequality (Lemma~\ref{lem:Chernoff}),
\[
\Pr(\Xh_i\notin \NetClose) \le e^{-cNk^{-1}},    
\]
for some universal $c$. Consequently, by Markov's inequality,
\[
\Pr\left( \sum_{i=1}^k \Ind\{\Xh_i\notin \NetClose\} \ge \varphi k \right) \le \varphi^{-1}e^{-cNk^{-1}} \,.
\]
Thus, when $N\ge Ck\log(1/\varphi)$, for some universal $C$, we get that with probability $\ge 1-\varphi$, we have ${|\NetClose\cap\{\Xh_1,\ldots,\Xh_k\}|\ge(1-\varphi)k}$. 

\qed

\subsection{Proof of Lemma~\ref{lem:UpperBound-StepI-Final}}
\label{sec:proof-lem:UpperBound-StepI-Final}

As discussed in the main text, it suffices to show that with probability $1-o_{\RegimeParam,\RateLim}(1)$, the centers in $\Xc$ have minimal distance $\sqrt{Ld}$ for some $L=L(\RegimeParam,\RateLim)$. In that case, choosing $\eps_0\le L/16$, the required results follows immediately from Lemmas~\ref{lem:UpperBound:Technical} and~\ref{lem:UpperBound:TechnicalCorr}.

The following argument is standard. Sample centers $\X_1,\ldots,\X_k$ sequentially. Let $E_i$ be the event that $\X_i\notin \bigcup_{l=1}^{i-1} \m{B}(\X_l,\sqrt{Ld})$ for $i=2,\ldots,k$. Clearly, $\Xc$ has minimal distance $\ge \sqrt{L d}$ if and only if $\bigcap_{l=2}^k E_i$ holds. Notice that $\Pr(E_i|\bigcap_{l=1}^{i-1}E_i)\ge 1-L^{d-1}(i-1)$, since $\X_i$ has to evade $i-1$ disjoint neighborhoods $\sqrt{d}\SphereD\cap \m{B}(\X_i,\sqrt{RLd})$, that amount to total surface area at most ${\le (i-1)\Surf(\partial \m{B}(\X_i,\sqrt{Ld})) = (i-1)L^{d-1}\Surf(\sqrt{d}\SphereD)}$. Somewhat crudely, we lower bound:
\begin{equation}\label{eq:Lemma:Distance}
    \Pr\left(\bigcap_{i=2}^k E_i\right) = \prod_{i=2}^k \Pr(E_l|\bigcap_{l=2}^{i-1}E_l) \ge (1-L^{d-1}k)^k\,.
\end{equation}
Whenever $L^{d-1}k=o(1/k)$, the bound in (\ref{eq:Lemma:Distance}) tends to $1$ as $k\to\infty$. When $k=e^{o(d)}$, any constant $L<1$ will work. When $k=e^{d\RateLim}$, any constant $L<e^{-2\RateLim}$ will work.

\qed

\section{Proofs for Section~\ref{sec:UpperBound-StepII}}
\label{sec:Appendix:StepII}


\subsection{Decoding Using a Corrupted Codebook}
\label{sec:Appendix:StepII:DecodingApprox}

Upon successful completion of Step I, by Lemma~\ref{lem:UpperBound-StepI-Final}, we will have constructed a list $\XchI=(\XhI_1,\ldots,\XhI_m)$ of size $m\ge (1-\varphi)k$, such that every member of $\XchI$ is $\sqrt{\epsI d}$-close to some unique codeword of $\Xc$. In this section, we show that whenever $\epsI$ is smaller than \emph{some particular threshold} $\eps_0=\eps_0(\RegimeParam,\RateLim)$, then $\XchI$ may be used to successfully decode messages encoded with $\Xc$, in the following sense:
\begin{itemize}
    \item Whenever $\XchI$ contains a point $\XhI_l$ close to $\X_{\isf(l)}\in\Xc$, applying the decoder on observation ${\Y=\X_{\isf(l)}+\sigma \Z}$ will indeed return, with high probability, the correct index $l$.
    \item As importantly, whenever $\XchI$ does not contain a keyword which is close to $\X_i$, then applying the decoder on observation $\Y=\X_i+\sigma\Z$ will consistently return an error symbol ``$\DecodeErrorSymb$''; that is, the decoder will not erroneously assign a sample $\Y$ to a label even if it does not have a close approximation for its corresponding center $\X_i$.
\end{itemize} 
We now proceed to formalize the discussion above.

Denote by 
\begin{equation}
    \label{eq:Approx}
    \Approx(\Xc) \subseteq \bigcupdot_{m=0}^k (\sqrt{d}\SphereD)^{\otimes m}
\end{equation}
the set of all lists $\XcTilde=(\XTilde_1,\ldots,\XTilde_m)$ (for \emph{any} $0\le m \le k$), for which there exists a permutation $\isf:[k]\to[k]$ satisfying 
\begin{equation}
    \label{eq:Approx:Perm}
    \|\XTilde_l - \X_{\isf(l)}\|^2 \le \eps d \quad  \textrm{for all }1\le l \le m\,.
\end{equation}
A \emph{family of decoders} is a mapping $\XcTilde \mapsto \Decode(\cdot|\XcTilde)$, mapping codebooks of any length $0\le m \le k$ to decision rules $\RR^d\to[k]\cup\{\DecodeErrorSymb\}$. $\Decode(\cdot|\XcTilde)$ may depend on $d,k,\sigma^2$ as well, and this shall be implied from now on. 

For example, consider the nearest-neighbor family of decoders:
\begin{equation}
    \label{eq:Approx:NN}
    \DecNN(\Y|\XcTilde) = \argmin_{l\in [m]} \|\Y-\XTilde_l\| \,.
\end{equation}
Recall that when $\XcTilde=\Xc$, the resulting decoder $\DecNN(\cdot|\Xc)$ is optimal (in the sense of average error) for decoding a message $\ell\sim \Unif([k])$ encoded using $\Xc$. However, when $\XcTilde$ is only a partial sub-codebook of $\Xc$, using the decoder (\ref{eq:Approx:NN}) might not be a good idea from a practical standing: an observation $Y=\X_i+\sigma\Z$ corresponding to a codeword $\X_i$ which is absent from $\XcTilde$ will \emph{necessarily} be decoded into an erroneous message. It is desirable that having identifed such a case, the decoder would instead \emph{declare} an error.

For a codebook $\Xc\in(\sqrt{d}\SphereD)^{\otimes k}$ and $\XcTilde\in\Approx_\eps(\Xc)$, let $\isf:[k]\to[k]$ be the permutation that satisfies (\ref{eq:Approx:Perm}). When there is more than one such permutation (as will surely be the case when $m\le k-2$), suppose that $\isf$ is chosen in some systematic way, such that the assignment $\Approx_\eps(\Xc)\to \mathrm{Sym}(k)$, $\XcTilde\mapsto \isf$ is well-defined. Note that $\isf([m])$ are the indices $\subseteq [k]$ of codewords $\X_i$ for which $\XcTilde$ contains an $\sqrt{\eps d}$-distance approximation. 

Similar to (\ref{eq:Error-avg-def}), we consider the error probability of decoding a message $i\in[k]$, encoded using codebook $\Xc$, with a decoder $\Decode(\cdot|\XcTilde)$, $\XcTilde\in \Approx_\eps(\Xc)$:
\begin{equation}\label{eq:Approx:P-err}
    P_{approx,i}(\sigma^2|\Xc,\XcTilde,\Decode(\cdot|\cdot)) = \begin{cases}
        \Pr(\isf^{-1}(i) \ne \Decode(\X_i+\sigma\Z|\XcTilde))\quad&\textrm{if }i\in \isf([m]) \\
        \Pr(\DecodeErrorSymb \ne \Decode(\X_i+\sigma\Z|\XcTilde))\quad&\textrm{if }i\notin \isf([m])
    \end{cases} \,.
\end{equation}
The ``twist'' over (\ref{eq:Error-avg-def}) is that if $i\notin \isf([m])$, we consider decoding to be successful if the decoder declares error.

We are ready to state the technical result of this section.

\begin{proposition}\label{prop:Approx}
    Suppose that $\RegimeParam>1$, $\Rate=\Capacity(\RegimeParam\sigma^2)$ (in either positive or zero rate, as before), $\Xc\sim \Unif(\sqrt{d}\SphereD)^{\otimes k}$. 

    There is $\eps_0=\eps_0(\RegimeParam,\RateLim)$ and a family of decoders $\Decode(\cdot|\cdot)$ such that for all $\eps\le \eps_0$:
        \[
        \lim_{d\to\infty}\Expt\left[ \sup_{\XTilde\in\Approx_\eps(\Xc)} \frac{1}{k}\sum_{i=1}^k P_{approx,i}(\sigma^2|\Xc,\XcTilde,\Decode(\cdot|\cdot)) \right] = 0 \,.
    \]
\end{proposition}
In words: Proposition~\ref{prop:Approx} states that provided that $\eps$ deceeds some particular threshold, one can construct a decoder family $\Decode(\cdot|\cdot)$, so that for any \emph{adversarially chosen} $\XcTilde\in\Approx_\eps(\Xc)$, 
the average decoding error (in the sense of \eqref{eq:Approx:P-err}) is uniformly small;
and that this holds for ``most'' random codebooks $\Xc$. 





\subsection{Proof of Proposition~\ref{prop:Approx}}

We simply adapt the decoders used in the proof of Proposition~\ref{prop:DecodingSphericalCodes}, and appearing in Section~\ref{sec:Appendix:Decoding}. We give a different construction at zero ($\RateLim=0$) and positive ($\RateLim>0$) rate.

\subsubsection{Zero Rate ($\RateLim=0$)}

We adapt the decoder from Section~\ref{sec:Appendix:Decoding:0}. 
Recall the decision rule implemented by this decoder (assuming access to the true codebook $\Xc$): it returns $i\in [k]$ if and only if: 1) $d^{-1}\langle\Y,\X_i\rangle\ge 1-\eta_1$; 2) For all $j\in[k]\setminus\{i\}$, $d^{-1}\langle\Y,\X_j\rangle\le 1-\eta_2$; if no such $i$ exists it returns an error. Here, $0<\eta_1<\eta_2$ are sufficiently small constants.

From now on, suppose without loss of generality that a message $i$ is sent. Denote by $\Y=\X_i+\sigma\Z$, the channel output. Let $\m{S}_i$ the event that the above decoder (which utilizes $\Xc$) succeeds; in other words,
\begin{equation}
    \m{S}_i(\eta_1,\eta_2) = \{d^{-1}\langle\Y,\X_i\rangle\ge 1-\eta_1\} \cap \bigcap_{j\in[k]\setminus \{i\}} \{d^{-1}\langle\Y,\X_i\rangle\le 1-\eta_2\} \,.    
\end{equation}
Recall: in Section~\ref{sec:Appendix:Decoding:0} we proved that $\Expt_{\Y_i,\Xc}[\Ind_{\m{S}_i(\eta_1,\eta_2)}]=1-o(1)$ for all sufficiently small constants $0<\eta_1<\eta_2$.

We now adapt the  construction of Section~\ref{sec:Appendix:Decoding:0} to use the ``corrupted'' codebook $\XcTilde$.
For thresholds $\tilde{\eta}_1,\tilde{\eta}_2$, the decoder returns $l$ if and only if $\XTilde_l$ is such that $d^{-1}\langle \Y,\XTilde_l\rangle \ge 1-\tilde{\eta}_1$, while for all other $j\ne l$, $d^{-1}\langle \Y,\XTilde_l\rangle < 1-\tilde{\eta}_2$. If no such codeword exists, it returns ``$\DecodeErrorSymb$''.
Given that message $i$ is sent, the decoder succeeds upon the following event:
\begin{itemize}
    \item If $i\in \isf([m])$, then: 1) $d^{-1}\langle \Y,\XTilde_{\isf^{-1}(i)}\rangle \ge 1-\tilde{\eta}_1$; 2) For all $l\ne \isf^{-1}(i)$, ${d^{-1}\langle \Y,\XTilde_{l}\rangle < 1-\tilde{\eta}_2}$.
    \item If $i\notin \isf([m])$, then for all $l\in [m]$, $d^{-1}\langle \Y,\XTilde_{l}\rangle < 1-\tilde{\eta}_2$.
\end{itemize}
From now on, we assume that $i\in \isf([m])$; when $i\notin \isf([m])$, the analysis follows in a similar manner.

Set $\e_l=\XTilde_l-\X_{\isf(l)}$, which may be chosen adversarially (but is independent of the noise $\Z$), such that $\|\e_l\|\le \sqrt{\eps d}$ (by definition of $\XcTilde\in\Approx_\eps(\Xc)$). We have 
\begin{align*}
    d^{-1}\langle \Y,\XTilde_{l}\rangle 
    &= d^{-1}\langle \Y,\X_{\isf(l)}\rangle + d^{-1}\langle \Y,\e_l\rangle \\
    &= d^{-1}\langle \Y,\X_{\isf(l)}\rangle + d^{-1}\langle \X_{i},\e_l\rangle + d^{-1}\sigma \langle \Z,\e_l \rangle \,.
\end{align*}
Clearly, $|d^{-1}\langle \X_{i},\e_l\rangle|\le \sqrt{\eps}$ (Cauchy-Schwartz).
Set $M = \max_{l\in [m]} |d^{-1}\sigma\langle \Z,\e_l\rangle|$. Observe:
\begin{equation}
    \begin{split}
        &d^{-1}\langle \Y,\X_{i}\rangle\ge 1-\eta_1 \implies d^{-1}\langle \Y,\XTilde_{\isf^{-1}(i)}\rangle \ge 1- \eta_1 -\sqrt{\eps}- M\,, \\
        &d^{-1}\langle \Y,\X_{j}\rangle\le 1-\eta_2 \implies d^{-1}\langle \Y,\XTilde_{\isf^{-1}(j)}\rangle \le 1-\eta_2 +\sqrt{\eps} + M \,.
    \end{split}
\end{equation}
We claim that with probability $1-2k^{-5}$, it holds that $M \le C \sqrt{\eps}$ for some $C=C(\RegimeParam)$. Consequently, if $\eps$  is small enough, and the thresholds $\tilde{\eta}_1<\tilde{\eta}_2$ are chosen such that 
\begin{align*}
    1-\eta_1-\sqrt{\eps} - C\sqrt{\eps} > 1-\tilde{\eta}_1,\quad 1-\eta_2+\sqrt{\eps} + C\sqrt{\eps} < 1-\tilde{\eta}_2\,,
\end{align*}
then under the high-probability event $\m{S}_i(\eta_1,\eta_2)\cap \{M\le C\sqrt{\eps}\}$, the adapted decoder necessarily returns the correct message $\isf^{-1}(i)$.  Since $1-\eta_1>1-\eta_2$, then this can clearly be made to hold whenever $\eps<\eps_0$ for small enough constant $\eps_0$. The proof of Proposition~\ref{prop:Approx} therefore concludes by the Lemma below.

\qed

\begin{lemma}
    Assume the conditions of Proposition~\ref{prop:Approx} (with $\RateLim=0$) and the setup described above.

    There is $C=C(\RegimeParam)$ such that $M\le C\sqrt{\eps}$ holds with probability $\ge 1-2k^{-5}$.
\end{lemma}
\begin{proof}
    Observe that $d^{-1}\sigma \langle \Z,\e_l \rangle$ is Gaussian with mean $0$ and variance $\le \eps\frac{\sigma^2}{d}$. and recall that, by definition $M = \max_{l\in [m]} |d^{-1}\sigma \langle \Z,\e_l \rangle|$. 

    By standard results on the maxima of Gaussian random variables, Lemmas~\ref{lem:MaxGauss-Expt} and~\ref{lem:Borell-TIS}, there is some universal $c$ such that $M_i \le \sqrt{\eps\frac{\sigma^2}{d}}\cdot \sqrt{c\log k}$ holds with probability $\ge 1-2k^{-5}$. 

    It remains to observe that 
        \[
            \frac{\sigma^2}{d} \log k = \sigma^2\Rate = \sigma^2\Capacity(\RegimeParam\sigma^2) \le \frac{1}{2\RegimeParam}\,,
    \]
    where we used $\Capacity(s)=\frac12\log(1+1/s)\le 1/(2s)$.
\end{proof}

\subsubsection{Positive Rate ($\RateLim>0$)}

We adapt the decoder from Section~\ref{sec:Appendix:Decoding:Pos}. Recall the decision rule implemented by this decoder (assuming access to the true codebook $\Xc$): the decoder returns an index $i$ if: 1) $d^{-1/2}\|\alpha\Y-\X_i\|\le \sqrt{\tau_1}$; 2) For all $j\ne i$, $d^{-1/2}\|\alpha\Y-\X_i\|> \sqrt{\tau_2}$; if no such $i$ exists, it returns $\DecodeErrorSymb$. Here $0<\sqrt{\tau_1}<\sqrt{\tau_2}$ are appropriately chosen thresholds.

We now adapt the aforementioned decoder to use $\XcTilde$ instead of $\Xc$. 

Observe that by the triangle inequality, 
\[
    d^{-1/2}\|\alpha\Y-\X_{\isf(l)}\| - \sqrt{\eps} \le d^{-1/2}\|\alpha\Y-\XTilde_l\| \le    d^{-1/2}\|\alpha\Y-\X_{\isf(l)}\| + \sqrt{\eps}\,.
\]
Consequently,
\begin{equation}\label{eq:Adapted:TriangleInequality}
    \begin{split}
        &d^{-1/2}\|\alpha\Y-\X_{\isf(l)}\| \le \sqrt{\tau_1} \implies d^{-1/2}\|\alpha\Y-\XTilde_l\| \le \sqrt{\tau_1} + \sqrt{\eps}  \,,\\
        &d^{-1/2}\|\alpha\Y-\X_{\isf(l)}\| > \sqrt{\tau_2} \implies d^{-1/2}\|\alpha\Y-\XTilde_l\| > \sqrt{\tau_2}-\sqrt{\eps} \,.
    \end{split}
\end{equation}

The adapted decoder will operate as follows. Assume that $\eps<\eps_0$ for some $2\sqrt{\eps_0}<\sqrt{\tau_2}-\sqrt{\tau_1}$, and set
\[
    \sqrt{\tilde{\tau}_1} = \sqrt{\tau_1} + \sqrt{\eps_0},\quad \sqrt{\tilde{\tau}_2}=\sqrt{\tau_2}-\sqrt{\eps_0} \,.
\]
The adapted decoder implements the following rule. It returns an index $l\in[m]$ whenever: $\XTilde_l$ is such that $\|\alpha\Y-\XTilde_l\|\le \sqrt{\tilde{\tau}_1}$; 2) For all $j\in[m]\setminus\{l\}$,  $\|\alpha\Y-\XTilde_j\|> \sqrt{\tilde{\tau}_2}$; if no such $l$ exists, it returns $\DecodeErrorSymb$.

Let us bound the error probability of the decoder. 
Suppose that a message $i$ was sent, so that $\Y=\X_i+\sigma\Z$, and assume without loss of generality that $i\in \isf([m])$; the case $i\notin \isf([m])$ follows similarly.  Consider the event
\[
S(\tau_1,\tau_2) = \{d^{-1/2}\|\alpha\Y-\X_i\|\le \sqrt{\tau_1}\}    \cap \bigcap_{j\in[k]\setminus\{i\}}\{d^{-1/2}\|\alpha\Y-\X_j\|> \sqrt{\tau_2}\} \,.
\]
In Section~\ref{sec:Appendix:Decoding:Pos} it was shown that, for appropriately chosen $\tau_1,\tau_2$, $S(\tau_1,\tau_2)$ is a high-probability event. Notice that by (\ref{eq:Adapted:TriangleInequality}), under the event $S(\tau_1,\tau_2)$, the adapted decoder indeed returns $\isf^{-1}(i)$; thus, we have proven that the average error probability is vanishing.



\qed





\subsection{Proof of Lemma~\ref{lem:StepII}}
\label{sec:proof-lem:StepII:Details}

Towards the proof of Lemma~\ref{lem:StepII},
we analyze the performance of Step II of the algorithm (Section~\ref{sec:UpperBound-StepII}) under a slightly more general setting, that does not use the strong symmetry properties that are available (w.h.p.) for a random spherical codebook $\Xc$, and that were crucial for constructing the test of Step I (Section~\ref{sec:UpperBound-StepI}). Instead, we only assume that one has access to some decoder under which the codebook attains small \emph{average error probability}.

Let $\Xc\subseteq (\sqrt{d}\SphereD)^k$ be a fixed codebook. Let $\Psi:\RR^d\to [k]\cup\{\DecodeErrorSymb\}$ be a decoder. Denote by $P_{i,j} = \Pr_{\Z}\left( \Psi(\X_i+\sigma\Z)=j \right)$ the probability (over the noise $\Z$)  that $\Psi(\cdot)$ outputs symbol $j\in [k]\cup\{\DecodeErrorSymb\}$ given that the true label was $i$. To keep the presentation light, we start by introducing some notation.

We say that the decoder $\Psi(\cdot)$ satisfies the \emph{$(\rho,\varphi)$- average error probability guarantee} if there exists an index set $\IndexSet\subseteq [k]$ of size $|\IndexSet|\ge (1-\varphi)k$ such that $\mathrm{Range}(\Psi)\subseteq \IndexSet\cup\{\DecodeErrorSymb\}$ and the following holds. Denote by
\begin{equation}
    \label{eq:Appendix:StepII:ErrorProbability}
    \bar{P}_i = \begin{cases}
        1-P_{i,i} \quad&\textrm{ if }i\in \IndexSet,\\
        1-P_{i,\DecodeErrorSymb}\quad&\textrm{ if }i\notin\IndexSet
    \end{cases}   
\end{equation}
the \emph{error probability} of the $i$-th message. Note that this is the same notion of error probability as in (\ref{eq:Approx:P-err}) from Section~\ref{sec:Appendix:StepII:DecodingApprox} above. Then we have
\begin{equation}
    \label{eq:Appendix:StepII:AvgErrorProbability}
    \frac1k \sum_{i=1}^k \bar{P}_i \le \rho \,. 
\end{equation}

Note: under the conditions of Lemma~\ref{lem:StepII}, upon successful completion of Step I, and with high probability over $\Xc\sim \Unif(\sqrt{d}\SphereD)^{\otimes k}$, Proposition~\ref{prop:Approx} implies that one may construct a decoder $\Psi(\cdot)=\Decode(\cdot|\XchI)$ which satsifes the $(\rho,\varphi)$ average error probability guarantee (up to a global relabeling, which we shall ignore henceforth), for some $\rho=o(1)$. 

Let us describe once again in detail the procedure of Step II, stated in terms of the notation above. One has access to a decoder (constructed from $\XchI$), and uses it to label a batch of $\bar{N}$ new samples, $\Y_1,\ldots,\Y_{\bar{N}}$. Let $S_i\subseteq [\bar{N}]$ be the subset of all measurements that have been assigned label $i$: 
\[
S_i = \left\{ j\;:\;\Psi(\Y_j)=i \right\} \,.
\]
Note  that measurements assigned $\DecodeErrorSymb$ are simply discarded. Next, we compute the cluster means:
\[
\Avg_i = \frac{1}{|S_i|}\sum_{j\in S_i}\Y_j\,,    
\]
so that $\Avg_i=\0$ if $S_i=\emptyset$. The final centers returned by the procedure are simply the projections of $\Avg_i$ onto the ball $\m{B}(\0,\sqrt{d})$:\footnote{Note that we project onto the ball rather than the sphere $\sqrt{d}\SphereD$ since projection onto convex sets is contracting in Euclidean norm. This is not so much the case for projection onto the sphere.}
\[
\Xh_i = \Proj(\Avg_i) \,.    
\]

The next Lemma summarizes our guarantees for Step II as described above.

\begin{lemma}\label{lem:StepII-Final}
    Suppose that $\RegimeParam>1$, and let $\eps,\varphi\in (0,1)$ be constants.
    Let $\Xc\subset (\sqrt{d}\SphereD)^k$ be some fixed center configuration, and suppose that one is given a decoder $\Psi(\cdot)$ satisfying the $(\rho,\varphi)$ average error probability guarantee, for some arbitrary $\rho=o_{\RegimeParam,\RateLim}(1)$. 

    Suppose that Step II is run with $\bar{N}\ge \frac{k\sigma^2}{\eps} + C\frac{k}{\eps^{1/2}}\log(1/\varphi)$ for some sufficiently large univeral $C>0$. Then
    \[
        \Expt\LossAvg(\Xc,\Xch) \le \frac{\eps}{1-\eps^{1/4}} + 8\varphi + o_{\RegimeParam,\RateLim}(1) \,,
    \]
    where the expectation is only taken over the randomness in the measurements $\Y_j=\X_{\RandLabel_j}+\Z_j$, namely $\RandLabel_j\sim\Unif([k])$ and $\Z_j\sim\m{N}(\0,\Id)$ (and the rate of decay in the $o(1)$ term depends on $\rho$).
\end{lemma}

The proof of Lemma~\ref{lem:StepII-Final} shall be given momentarily, in Section~\ref{sec:proof-lem:StepII-Final} below. Before getting to it, let us show how it immediately implies Lemma~\ref{lem:StepII}:

\begin{proof}
    (Of Lemma~\ref{lem:StepII}). Note that the diameter of the ball $\m{B}(\0,\sqrt{d})$ is $2\sqrt{d}$, so necessarily ${\dist^2(\X_i,\Xch)\le 4d}$ and therefore $\LossAvg(\Xc,\Xch)\le 4$. Thus, for any event $\m{E}$, 
    \begin{equation}
        \label{eq:IntermediateStepII}
    \Expt\LossAvg(\Xc,\Xch) \le \Expt\left[\LossAvg(\Xc,\Xch)\Ind\{\m{E}\}\right] + 4\Pr(\m{E}^c) \,.
    \end{equation}

    Let $\XchI$ be the list returned by Step I of the algorithm. 
    Let $\eps_0=\eps_0(\RegimeParam,\RateLim)$ be the threshold from Proposition~\ref{prop:Approx}, and consider the event
    \[
    \m{E}_{I} = \left\{ \XchI\in \Approx_{\eps_0}(\Xc), \;|\IndexSet|\ge (1-\varphi)k \right\} \,,
    \]
    where $\IndexSet$ is defined in Lemma~\ref{lem:UpperBound-StepI-Final}.
    By Lemma~\ref{lem:UpperBound-StepI-Final}, provided that $N$ is large enough,  $\Pr(\m{E}_{1}^c)\le \varphi+o_{\RegimeParam,\RateLim}(1)$. Let $\Psi(\cdot)=\Decode(\cdot|\XchI)$ be the decoder promised by Proposition~\ref{prop:Approx}. For a suitably chosen (large enough) $\rho=o_{\RegimeParam,\RateLim}(1)$, denote the event
    \[
    \m{E} = \left\{ \Psi(\cdot)\textrm{ satisfies the }(\rho,\varphi)\textrm{ average probability guarantee} \right\} \,.
    \]
    By Proposition~\ref{prop:Approx}, along with Markov's inequality, we have $\Pr(\m{E}^c|\m{E}_{I})=o_{\RegimeParam,\RateLim}(1)$. Thus,
    \[
    \Pr(\m{E}^c) \le \Pr(\m{E}_I^c) + \Pr(\m{E}^c|\m{E}_I) \le \varphi + o_{\RegimeParam,\RateLim}(1) \,.
    \]
    Use the event $\m{E}$ in \eqref{eq:IntermediateStepII}. By Lemma~\ref{lem:StepII-Final}, for large enough $\bar{N}$, $\Expt\left[\LossAvg(\Xc,\Xch)\Ind\{\m{E}\}\right]\le \frac{\eps}{1-\eps^{1/4}} + 8\varphi + o_{\RegimeParam,\RateLim}(1)$, and so Lemma~\ref{lem:StepII} follows.
    
\end{proof}

\subsection{Proof of Lemma~\ref{lem:StepII-Final}}
\label{sec:proof-lem:StepII-Final}

Conceptually, the proof of Lemma~\ref{lem:StepII-Final} is quite straightforward. Denote by $L_i\subseteq [\bar{N}]$ the measurements $\Y_j=\X_{\ell_j}+\Z_j$ whose \emph{true label} is $\ell_j=i$. Imagine, for a moment, that we had access to a \emph{genie-aided} decoder, that always assigns measurements to their true labels. In that case, $S_i=L_i$, and so $\Avg_i$ is just the sample mean of $|L_i|$ i.i.d. Gaussian measurements $\m{N}(\X_i,\sigma^2\Id)$. Since, \emph{on average}, $|L_i|= \bar{N}/k$, the MSE is $\Expt\|\X_i-\Avg_i\|^2 \approx \sigma^2 k /\bar{N}$, which is $\approx \eps$ when $\bar{N}\approx \sigma^2 k/\eps$. In practice, however, one does not have access to a clairvoyant decoder: we only assume an average error probability guarantee. Another difficulty is that we {cannot} guarantee that $|L_i|\approx \bar{N}/k$ \emph{simultaneously} for all $i$, unless $\sigma^2\gtrsim \log k$ (recall the coupon-collecting issue highlighted in the main paper: we need $\bar{N}\gtrsim k\log k$ to even \emph{observe} a measurement of every label). Consequently, the analysis has to be carried out somewhat delicately. 

We start with the trivial observation, that the average error probability guarantee implies that, in fact, \emph{most} individual labels $i\in [k]$ must have a small error probability:
\begin{lemma}\label{lem:StepII:Set1}
    Let $(\Xc,\Psi(\cdot))$ satisfy the $(\rho,\varphi)$ average error probability guarantee (\ref{eq:Appendix:StepII:AvgErrorProbability}). There is a set of indices $\m{J}_1\subseteq [k]$ of size $|\m{J}_1|\ge (1-\rho^{1/2})k$ such that 
    \begin{equation}
        \bar{P}_i \le \rho^{1/2}\quad \textrm{ for all }i \in \m{J}_1 \,.
    \end{equation}
\end{lemma}
\begin{proof}
    An immediate consequence of Markov's inequality.
\end{proof}

Aside from lower bounding $|L_i|\gtrsim \bar{N}/k$ in expectation, we shall also need to control the number of measurements that were \emph{erroneously} assigned label $i$ by $\Psi(\cdot)$. Following the notation of (\ref{eq:Appendix:StepII:ErrorProbability}), let
\begin{equation}
    \bar{Q}_i = \frac1k \sum_{l\in [k]\setminus\{i\}} P_{l,i} 
\end{equation}
be the probability that a random measurement $\Y=\X_\RandLabel + \sigma\Z$, $\RandLabel\sim \Unif([k])$ has true label $\ell\ne i$, but is erroneously assigned label $i$ by $\Psi(\cdot)$. 

\begin{lemma}\label{lem:StepII:Set2}
    Let $(\Xc,\Psi(\cdot))$ satisfy the $(\rho,\varphi)$ average error probability guarantee (\ref{eq:Appendix:StepII:AvgErrorProbability}). There is a set of indices $\m{J}_2\subseteq [k]$ of size $|\m{J}_2|\ge (1-\rho^{1/2})k$ such that 
    \begin{equation}
        \bar{Q}_i \le \rho^{1/2}/k\quad \textrm{ for all }i \in \m{J}_2 \,.
    \end{equation}
\end{lemma}
\begin{proof}
    Observe that 
    \[
    \frac1k \sum_{i=1}^k \left( \sum_{l\in [k]\setminus\{i\}} P_{l,i}  \right)  = \frac1k \sum_{l=1}^k \left(\sum_{i\in [k]\setminus \{l\}} P_{l,i} \right)\le \frac1k \sum_{l=1}^k \bar{P}_l \le \rho\,.
    \]
    Consequently, by Markov's inequality, there is $\m{J}_2\subseteq [k]$ of size $|\m{J}_2|\ge (1-\rho^{1/2})k$ such that $\sum_{l\in [k]\setminus\{i\}} P_{l,i}\le \rho^{1/2}$ for all $i\in \m{J}_2$.
\end{proof}

Recall: $L_i\subseteq [\bar{N}]$ are the measurements whose true label is $i$; $S_i\subseteq [\bar{N}]$ are the measurements assigned label $i$ by $\Psi(\cdot)$ (whether truthfully or erroneously). Define the event
\begin{equation}\label{eq:StepII:Event-i-1}
    \m{E}_{i,1} = \left\{ |S_i\setminus L_i| \le \rho^{1/4}\frac{\bar{N}}{k}\quad \cap \quad |L_i|\ge (1-\eps^{1/4})\frac{\bar{N}}{k} \right\} \,.
\end{equation}
\begin{lemma}\label{lem:StepII:Event-i-1}
    Under the conditions of Lemma~\ref{lem:StepII-Final}, for any $i\in \m{J}_2$,
    \[
    \Pr\left(\m{E}_{i,1}^c\right) \le \varphi + o(1) \,.
    \]
\end{lemma}
\begin{proof}
    Let us start by showing that $|S_i\setminus L_i|$ is small with high probability. By definition, ${\Pr(\Y\in S_i\setminus L_i)=\bar{Q}_i}$. Since $i\in \m{J}_2$ this probability is $\le \rho^{1/2}/k$. Thus, by Markov's inequality,
    \[
    \Pr\left( |S_i\setminus L_i|>\rho^{1/4}\frac{\bar{N}}{k} \right)    \le \frac{\Expt |S_i\setminus L_i|}{\rho^{1/4}\frac{\bar{N}}{k}} \le \frac{\rho^{1/2}\frac{\bar{N}}{k}}{\rho^{1/4}\frac{\bar{N}}{k}}=\rho^{1/4}=o(1) \,.
    \]

    Moving on, observe that $|L_i|\sim \mathrm{Binomial}(1/k,\bar{N})$. By Chernoff's inequality, Lemma~\ref{lem:Chernoff-SmallDeviations},
    \[
    \Pr\left( |L_i|\le (1-\eps^{1/4})\frac{\bar{N}}{k} \right) \le e^{-c\eps^{1/2}\frac{\bar{N}}{k}}\,,
    \]
    which is $\le \varphi$ for $\bar{N}\ge C\frac{k}{\eps^{1/2}}\log(1/\varphi)$.
\end{proof}

For $i\in \IndexSet$, define the event 
\begin{equation}
    \label{eq:StepII:Event-i-2}
    \m{E}_{i,2} = \left\{ |S_i\cap L_i| \ge (1-\rho^{1/4})|L_i| \right\} \,.
\end{equation}
\begin{lemma}\label{lem:StepII:Event-i-2}
    Under the conditions of Lemma~\ref{lem:StepII-Final}, for any $i\in \IndexSet\cap \m{J}_1$,
    \[
        \Pr\left(\m{E}_{i,2}^c\right) = o(1)\,.
    \]
\end{lemma}
\begin{proof}
    By the definition of $\m{J}_1$, conditioned on $\RandLabel_j=i$, $\Pr(j\notin S_i|\ell_j=i)=\bar{P}_i\le \rho^{1/2}$. Thus, by Markov's inequality,
    \[
\Pr\left( |L_i\setminus S_i|\ge \rho^{1/4}|L_i| \right) \le \rho^{1/4} = o(1)\,.
    \]
\end{proof}

We are ready to bound $\Expt\LossAvg(\Xc,\Xch)$. Observe that for any events $\m{E}_1,\ldots,\m{E}_k$,
\begin{equation}\label{eq:StepII-Final-Eq}
    \begin{split}
        \Expt\LossAvg(\Xc,\Xch) 
        &= \frac1k\sum_{i=1}^k d^{-1}\Expt \dist^2(\X_i,\Xch) \\
        &\overset{(i)}{\le} \frac1k\sum_{i=1}^k d^{-1}\Expt \left[\dist^2(\X_i,\Xch)\Ind\left\{ \m{E}_i, i\in\IndexSet\cap\m{J}_1\cap \m{J}_2 \right\}\right] \\
        &+ \frac{4}{k}\sum_{i=1}^k\Pr(\m{E}_i^c) + 4\frac{|\IndexSet^c|}{k} + 4\frac{|(\m{J}_1\cap\m{J}_2)^c|}{k} \\
        &\overset{(ii)}{\le} \frac1k\sum_{i=1}^k d^{-1}\Expt \left[\dist^2(\X_i,\Xch)\Ind\left\{ \m{E}_i, i\in\IndexSet\cap \m{J}_1\cap \m{J}_2 \right\}\right] + \frac{4}{k}\sum_{i=1}^k\Pr(\m{E}_i^c) + 4\varphi + o(1) \,,
    \end{split}
\end{equation}
where: (i) follows from $\dist(\Xc,\Xch)\le 4d$ (the diameter of the ball is $2\sqrt{d}$); (ii) Follows from $|\IndexSet|\le \varphi k$ (by definition of the $(\rho,\varphi)$ error probability guarantee) and $|(\m{J}_1\cap\m{J}_2)^c|=o(k)$ follows from Lemmas~\ref{lem:StepII:Set1} and~\ref{lem:StepII:Set2}, with $\rho=o(1)$.

We take the event
\begin{equation}
    \label{eq:StepII:Event-i}
    \m{E}_i = \m{E}_{i,1}\cap \m{E}_{i,2} \cap \m{E}_{i,3}\,,
\end{equation}
where $\m{E}_{i,1}$ is defined in (\ref{eq:StepII:Event-i-1}), $\m{E}_{i,2}$ is defined in (\ref{eq:StepII:Event-i-2}), and the definition of $\m{E}_{i,3}$ shall be deferred for later; its details would be somewhat obtuse at this point in the analysis.

To lighten the notation, introduce
\begin{align}\label{eq:StepII:ExptAvg}
    \ExptAvg[F_i] = \frac{1}{kd}\sum_{i=1}^k \Expt[F_i \Ind\left\{ \m{E}_i, i\in\IndexSet\cap\m{J}_1\cap \m{J}_2 \right\}]\,,
\end{align}
where $F_i$ is any sequence indexed by $i\in [k]$.

For $i\in \IndexSet$, 
\[
\dist^2(\X_i,\Xch) \le \|\X_i-\Xh_i\|^2 \le \|\X_i-\Avg_i\|^2    \,,
\]
where we used $\Xh_i=\Proj(\Avg_i)$ and that projection onto convex sets is contracting with respect to Euclidean norm. Moreover, note that 
\[
\Dc=(\D_1,\ldots,\D_k) \in \RR^{d\times k}\quad \mapsto \quad     (\ExptAvg[\|\D_i\|^2])^{1/2}
\]
is a semi-norm, hence satisfies the triangle inequality. 

We decompose
\begin{align*}
    \Avg_i 
    &= \frac{1}{|S_i|}\sum_{j\in S_i} \Y_j \\
    &=    \frac{1}{|S_i|}\sum_{j\in L_i} \Y_j  +  \frac{1}{|S_i|}\sum_{j\in S_i\setminus L_i} \Y_j -  \frac{1}{|S_i|}\sum_{j\in L_i\setminus S_i} \Y_j \\
    &=
    \frac{1}{|S_i|}\sum_{j\in L_i} \Y_i   +  \frac{1}{|S_i|}\sum_{j\in S_i\setminus L_i} \X_i -  \frac{1}{|S_i|}\sum_{j\in L_i\setminus S_i} \X_{\ell_j} + \frac{\sigma}{|S_i|}\sum_{j\in S_i\setminus L_i} \Z_j -  \frac{\sigma}{|S_i|}\sum_{j\in L_i\setminus S_i} \Z_j
    \,,    
\end{align*}
and accordingly bound the first term of (\ref{eq:StepII-Final-Eq}):
\begin{equation}
    \label{eq:StepII:LotsOfIs}
    \begin{split}
        \left( \ExptAvg\dist^2(\X_i,\Xch) \right)^{1/2} 
        &\le \left( \ExptAvg\|\X_i-\Avg_i\|^2 \right)^{1/2} \\
        &\le \underbrace{\left( \ExptAvg\left\|\X_i-\frac{|L_i|}{|S_i| }\X_i\right\|^2 \right)^{1/2}}_{I_1}
        +  \underbrace{\left( \ExptAvg\left\|\frac{|L_i|}{|S_i| }\X_i-\frac{1}{|S_i|}\sum_{j\in L_i} \Y_j\right\|^2 \right)^{1/2}}_{I_2}\\
        &+ \underbrace{\left( \ExptAvg\left\|\frac{1}{|S_i|}\sum_{j\in S_i\setminus L_i} \X_i -  \frac{1}{|S_i|}\sum_{j\in L_i\setminus S_i} \X_{\ell_j}\right\|^2 \right)^{1/2}}_{I_3}\\
        &+ \underbrace{\left( \ExptAvg\left\|\frac{\sigma}{|S_i|}\sum_{j\in S_i\setminus L_i} \Z_j\right\|^2 \right)^{1/2}}_{I_4}
        + \underbrace{\left( \ExptAvg\left\|\frac{\sigma}{|S_i|}\sum_{j\in L_i\setminus S_i} \Z_j\right\|^2 \right)^{1/2}}_{I_5} \,.
    \end{split}
\end{equation}
We proceed to bound the terms above. Starting with $I_1$, observe that 
\[
    I_1^2 = \ExptAvg\left\|\X_i-\frac{|L_i|}{|S_i| }\X_i\right\|^2 = d\cdot \ExptAvg\left| 1-\frac{|L_i|}{|S_i|} \right|,
\]
and therefore $I_1=o(1)$, since under $\m{E}_i$ we must have $(1-o(1))|L_i| \le |S_i|\le (1+o(1))|L_i|$. 

Moving on to $I_2$, 
\begin{align*}
    I_2^2 =  \ExptAvg\left\|\frac{|L_i|}{|S_i| }\X_i-\frac{1}{|S_i|}\sum_{j\in L_i} \Y_j\right\|^2   =  \ExptAvg\left[\frac{\sigma^2|L_i|}{|S_i|^2}\left\|\frac{1}{\sqrt{|L_i|}}\sum_{j\in L_i} \Z_j\right\|^2 \right]\,.
\end{align*}
Noting that $\frac{1}{\sqrt{|L_i|}}\sum_{j\in L_i} \Z_j\sim \m{N}(\0,\Id)$, and that under $\m{E}_i$, $\frac{|L_i|}{|S_i|^2}\le (1+o(1))\frac{1}{1-\eps^{1/4}}\frac{k}{\bar{N}}$ we get that $I_2^2\le \frac{\eps}{1-\eps^{1/4}} + o(1)$ since $\bar{N}\ge \sigma^2 k /\eps$.

As for $I_3$, 
\begin{align*}
    I_3^2 = \ExptAvg\left\|\frac{1}{|S_i|}\sum_{j\in S_i\setminus L_i} \X_i -  \frac{1}{|S_i|}\sum_{j\in L_i\setminus S_i} \X_{\ell_j}\right\|^2 \le d \cdot \ExptAvg \left(\frac{|S_i\setminus L_i| + |L_i\setminus S_i|}{|S_i|}\right)^2 = o(1) \,.
\end{align*}

Bounding the terms $I_4,I_5$ is somewhat more involved, and requires the introduction of a new event, $\m{E}_{i,3}$, whose definition has been deferred up to this point.

\begin{lemma}\label{lem:StepII:NoiseTermsNegligible}\label{lem:StepII:Event-i-3}
    Define the event $\m{E}_{i,3}$ in (\ref{eq:StepII:Event-i-3}). Then $\Pr(\m{E}_{i,3})=o(1)$ and moreover 
    \[
    I_4,I_5 = o(1) \,.    
    \]
\end{lemma}
To keep the narrative flow, we defer the proof of Lemma~\ref{lem:StepII:NoiseTermsNegligible} to Section~\ref{sec:proof-lem:StepII:NoiseTermsNegligible} below.

We are ready to tie all loose ends, and finish the proof of Lemma~\ref{lem:StepII-Final}. By (\ref{eq:StepII:LotsOfIs}) and the upper bounds we have shown for $I_1,\ldots,I_5$, we get $\ExptAvg\dist^2(\X_i,\Xch) \le \frac{\eps}{1-\eps^{1/4}} + o(1)$. Combining Lemmas~\ref{lem:StepII:Event-i-1}, \ref{lem:StepII:Event-i-2} and \ref{lem:StepII:Event-i-3}, we get $\Pr(\m{E}_i)\le \varphi + o(1)$. Plugging these estimates into (\ref{eq:StepII-Final-Eq}), we finally get the claimed bound of Lemma~\ref{lem:StepII-Final}.

\subsubsection{Proof of Lemma~\ref{lem:StepII:NoiseTermsNegligible}}
\label{sec:proof-lem:StepII:NoiseTermsNegligible}

The terms $I_4,I_5$ correspond to sums of independent ``noise'' vectors, whose mean is zero. Therefore,
one expects different $\Z_j$-s to cancel out one another on average,
so that, for example (considering $I_4$),
\[
    \Expt\left\| \frac{\sigma}{|S_i|}\sum_{j\in S_i\setminus L_i}\Z_j \right\|^2 \approx \frac{\sigma^2}{|S_i|^2} \cdot |S_i\setminus L_i|d    
\]
rather than $\frac{\sigma^2}{|S_i|^2} \cdot |S_i\setminus L_i|^2 d$ which is what we would have gotten by naive application of the triangle inequality. A subtle point is that the set $S_i$ actually depends on the noise vectors $\Z_j$, so one needs to apply some care when taking the expectation above.\footnote{\emph{A priori}, we cannot discount the possibility that conditioned on $j\in S_i$, the noise $\Z_j$ biases towards some particular direction.}
We propose to overcome this difficulty through a rather crude bound.

For a subset $B\subseteq [n]$ (which itself may be random, but independent of $\{\Z_j\}_{j\in[\bar{N}]}$, let
\begin{equation*}
    \D(B,m)=\max_{S\subseteq B,|S|\le m} \left\| \sum_{j\in S} \Z_j \right\|^2 \,.
\end{equation*}
Since under $\m{E}_{i,1}\cap \m{E}_{i,2} \subset \m{E}_{i}$ we have
\[
|S_i\setminus L_i| \le \rho^{1/4} \frac{\bar{N}}{k},\quad |L_i\setminus S_i|\le \rho^{1/4}|L_i|,\quad |L_i|\ge (1-\eps^{1/4})\frac{\bar{N}}{k} \,,
\]
recalling the definition of $I_4,I_5$, (\ref{eq:StepII:LotsOfIs}), clearly,
\begin{equation}\label{eq:StepII:I1-I2}
    \begin{split}
        &I_4^2 = \sigma^2 \cdot \ExptAvg \left[ \frac{1}{|S_i|^2} \left\|\sum_{j\in S_i\setminus L_i} \Z_j\right\|^2 \right] \\
        &\phantom{I_4^2}  \lesssim \frac{\sigma^2}{(\bar{N}/k)^2}\cdot \ExptAvg \left[\D([\bar{N}],\rho^{1/4}(\bar{N}/k))\right], \\
        &I_5^2 = \sigma^2\cdot \ExptAvg \left[ \frac{1}{|S_i|^2} \left\|\sum_{j\in L_i\setminus S_i} \Z_j\right\|^2 \right] \\
        &\phantom{I_5^2}
        \lesssim \sigma^2\cdot \ExptAvg \left[ \frac{1}{|L_i|^2} \cdot \D(L_i,\rho^{1/4}|L_i|) \right]\,.
    \end{split}
\end{equation}
We are ready to define the event $\m{E}_{i,3}$, which has been deferred up to this point.

\paragraph{The event $\m{E}_{i,3}$.} For $C$ a sufficiently large universal constant, define
\begin{equation}\label{eq:StepII:Event-i-3}
    \begin{split}
        \m{E}_{i,3} 
        &= \left\{ \D([n],\rho^{1/4}(\bar{N}/k)) \le C \left(\rho^{1/4}(\bar{N}/k)\right)^2\log\frac{\bar{N}e}{\rho^{1/4}(\bar{N}/k)} + C\rho^{1/4}(\bar{N}/k)\log d  \right\},\\
        &\cap \left\{ \D(L_i,\rho^{1/4}|L_i|) \le C \left(\rho^{1/4}|L_i|\right)^2\log\frac{|L_i|e}{\rho^{1/4}|L_i|} + C\rho^{1/4}|L_i|\log d \right\} \,.
    \end{split}
\end{equation}
By Lemma~\ref{lem:NoiseMax:Technical}, given below, $C$ may indeed be chosen so that $\Pr(\m{E}_{i,3}^c)=o(1)$. 

\paragraph{Bounding $I_4$.}
Using (\ref{eq:StepII:I1-I2}) and (\ref{eq:StepII:Event-i-3}), 
\begin{align*}
    I_4^2 \lesssim  \frac{\sigma^2}{(\bar{N}/k)^2} \cdot d^{-1} \cdot \left\{ o\left((\bar{N}/k)^2\log k\right) + o\left((\bar{N}/k)\log d\right)  \right\} \,.
\end{align*}
The first term is $o(1)$ because $\sigma^2d^{-1}\log k = O(1)$. The second term is $o(1)$ because $\frac{\sigma^2 }{(\bar{N}/k)}d^{-1}\log d=O(d^{-1}\log d)=o(1)$, since $\bar{N}\gtrsim k\sigma^2$. Thus $I_4=o(1)$.

\paragraph{Bounding $I_5$.}
Using (\ref{eq:StepII:I1-I2}) and (\ref{eq:StepII:Event-i-3}), 
\begin{align*}
    I_5^2 
    &\lesssim
    \sigma^2 \ExptAvg\left[ o(1) + o\left( \frac{1}{|L_i|}\log d \right) \right] \\
    &\lesssim
    o\left(\sigma^2 d^{-1}\right) + o\left(\sigma^2 d^{-1}\frac{1}{(\bar{N}/k)}\log d \right) \,.
\end{align*}
The first term is $o(1)$ since, by assumption, $\sigma^2=o(d)$. Since $\bar{N}\gtrsim k\sigma^2$, the second terms is $o\left( \frac{\log d}{d} \right)=o(1)$. Thus, $I_5=o(1)$.

This conclude the proof of Lemma~\ref{lem:StepII:NoiseTermsNegligible}.

\qed

\paragraph{A Technical Lemma.}

\begin{lemma}\label{lem:NoiseMax:Technical}
    Let $\Z_1,\ldots,\Z_n\sim \m{N}(\0,\Id)$ be independent. For a set $S\subseteq [n]$ let $\W_S = \sum_{i\in S}\Z_i$. There is a univeral $C>0$ such that for $t\ge 1$,
    \[
    \Pr\left( \max_{S\subseteq [n],|S|\le t} \|\W_S\| \ge  Ct\sqrt{\log \frac{ne}{t}} + C\sqrt{t \log d}  \right) \le d^{-5} \,.    
    \]
\end{lemma}
\begin{proof}
    This is a straightforward application of the well-known Lemmas~\ref{lem:MaxGauss-Expt} and Lemma~\ref{lem:Borell-TIS}, along with a standard ``trick''.

    Let $\Net$ be a $1/2$-net of $\SphereD$, of size $\le 5^d$ (e.g. \cite[Example 5.8]{wainwright2019high}). It may be readily verified that for any vector $\x\in\RR^d$, $\|\x\|=\max_{\e\in\SphereD}\langle \x,\e\rangle \le 2\max_{\e\in\Net}\langle \x,\e\rangle$.  Set $N_t=|\Net|\sum_{s=1}^t\binom{n}{s} \lesssim 5^d(ne/t)^t$. 
    Note also that for any $\e\in\SphereD$, $\Expt\left[ \langle \W_S,\e\rangle^2\right] \le |S|\le t$.
    
    By the expectation bound Lemma~\ref{lem:MaxGauss-Expt}, and the Borell-TIS inequality, Lemma~\ref{lem:Borell-TIS}, for $x\ge 0$,
    \begin{align*}
        \Pr\left(\max_{S\subseteq [n],|S|\le t}\|\W_S\|\ge 2\sqrt{t}(\sqrt{2\log N_t} + x) \right)  
        &\le \Pr\left(\max_{S\subseteq [n],|S|\le t,\e\in\Net}\langle \W_S,\e\rangle \ge \sqrt{t}(\sqrt{2\log N_t} + x) \right)  \\
        &\le e^{-x^2/2} \,.
    \end{align*}
    Set $x=\sqrt{10\log(d)}$ to get the claimed bound.

\end{proof}

\newpage

\section{Auxiliary Technical Results}

\subsection{Concentration Inequalities}

The Gaussian Lipschitz concentration inequality \cite[Theorem 2.25]{wainwright2019high}:
\begin{lemma}[Gaussian Lipschitz concentration inequality]
    \label{lem:Gaussian-Lip}
    Let $f:\RR^d\to\RR$ be $L$-Lipschitz, and $\Z\sim\m{N}(\0,\Id)$. For all $t\ge 0$,
    \begin{align*}
        \Pr\left(f(\Z)\ge \Expt f(\Z)+t \right) \le e^{-\frac{t^2}{2L^2} } \,.
    \end{align*}
\end{lemma}

Standard bound on the measure of a spherical cap \cite[Eq. (3.33)]{wainwright2019high}:
\begin{lemma}
    \label{lem:Unif-Sphere}
    Let $\Z\sim\Unif(\SphereD)$. For all $\bm{u}\in \SphereD$ and $t\in(0,1)$,
    \begin{align*}
        \Pr\left(\langle \bm{u},\Z\rangle \ge t \right) \le e^{-dt^2/2} \,.
    \end{align*}
\end{lemma}

We state Bernstein's inequality for independent bounded random variables \cite[Theorem 2.8.4]{vershynin2018high}
\begin{lemma}
    \label{lem:BernsteinBounded}
    Let $X_1,\ldots,X_n$ be independent, centered, with $|X_i|\le K$. Set ${S_n=\sum_{i=1}^n X_i}$. Then for all $t\ge 0$,
	\[
	\Pr\left( \left| S_n \right| \ge t \right) \le 2\exp\left(-\frac{t^2/2}{\sum_{i=1}^n \Var(X_i) + Kt/3}\right) \,.
	\]
\end{lemma}

The following is Bernstein's inequality for sums of independent sub-Exponential random variables \cite[Theorem 2.8.1]{vershynin2018high}:
\begin{lemma}[Bernstein's inequality, sub-Exponential RVs]\label{lem:Bernstein}
	Let $X_1,\ldots,X_n$ be independent and sub-exponential. Set ${S_n=\sum_{i=1}^n X_i}$. Then for all $t\ge 0$,
	\[
	\Pr\left( \left| S_n-\Expt[S_n] \right| \ge t \right) \le 2\exp\left[ -c\min\left( \frac{t^2}{\sum_{i=1}^n \|X_i\|_{\psi_1}^2}, \frac{t}{\max_{1\le i\le n}\|X_i\|_{\psi_1}} \right) \right]\,,
	\]
	where $c>0$ is an absolute constant.
\end{lemma}


We state the following version of the Chernoff bound for Bernoulli random variables. For a citable reference, see for example \cite[Section 2.2]{boucheron2013concentration}:

\begin{lemma}[Chernoff's inequality]
    \label{lem:Chernoff}
    Let $X_1,\ldots,X_n$ be i.i.d. Bernoulli random variables, with $\Expt[X_i]=q$. Let $t\le q \le p$. Then 
    \[
    \Pr\left(\sum_{i=1}^n X_i \ge p\right) \le e^{n\KLb(p;q)},\quad 
    \Pr\left(\sum_{i=1}^n X_i \le t \right) \le e^{n\KLb(t;q)} \,,
    \]
    where
    \[
    \KLb(p;q)=p\log\frac{p}{q}+(1-p)\log\frac{1-p}{1-q}    
    \]
    is the Kullback-Leibler divergence divergence between $\mathrm{Ber}(p)$ and $\mathrm{Ber}(q)$.
\end{lemma}

We shall use Lemma~\ref{lem:Chernoff} with the following easy estimate:
\begin{lemma}\label{lem:KL-binary}
    Let $0<q<p<1/2$. Then  
    \[
        \KLb(p;q)\ge p\log\frac{p}{q}-2p \,.    
    \]
\end{lemma}
\begin{proof}
    One may readily verify that $\log(1-p)\ge -2p$ holds for $0\le p\le 1/2$. Thus, 
    \[
        \KLb(p;q)=p\log\frac{p}{q}+(1-p)\log\frac{1-p}{1-q}\ge p\log\frac{p}{q}+(1-p)\log{(1-p)}\overset{(\star)}{\ge} p\log\frac{p}{q}+\log(1-p)\ge p\log\frac{p}{q}-2p \,,    
    \]
    where $(\star)$ holds since $\log(1-p)$ is negative. 
\end{proof}

The following is a version of Chernoff's inequality, specialized for small deviations, and taken from \cite[Exercise 2.3.5]{vershynin2018high}
\begin{lemma}[Chernoff's inequality; small deviations]
    \label{lem:Chernoff-SmallDeviations}
    In the setting of Lemma~\ref{lem:Chernoff}, for $\delta\in(0,1)$, 
    \begin{equation*}
        \Pr\left(\sum_{i=1}^n X_i \ge (1+\delta)qn\right) \le e^{-c\delta^2 qn},\quad \Pr\left(\sum_{i=1}^n X_i \le (1-\delta)qn\right) \le e^{-c\delta^2 qn}\,,
    \end{equation*}
    where $c>0$ is universal.
\end{lemma}

\subsection{Maxima of Gaussian Random Variables}

We state two elementary results about the maximum of $n$ Gaussian random variables.

\begin{lemma}\label{lem:MaxGauss-Expt}
    Let $\Z=(Z_1,\ldots,Z_n)$ be a Gaussian random vector, such that $\Expt[\Z]=\0$ and ${\Expt[Z_i^2]\le \sigma^2}$ for all $i$. Then
    \[
    \Expt[ \max_{1\le i \le n} Z_i ] \le \sqrt{2\sigma^2\log n} \,.    
    \]
    (When $Z_1,\ldots,Z_n$ are uncorrelated, this is in fact tight to leading order. But we shall not use this stronger fact.)
\end{lemma}
\begin{proof}
    This is classical. For completeness, we give a one-line proof. For all $\beta>0$,
    \[
        \Expt [\max_{1\le i \le n} Z_i ] \le \frac{1}{\beta}  \Expt \log \sum_{i=1}^n e^{\beta Z_i } \overset{(\star)}{\le} \frac{1}{\beta} \log \Expt  \sum_{i=1}^n e^{\beta Z_i } \le \frac1\beta(\log n + \frac12\sigma^2\beta^2) \,,
    \]
    where $(\star)$ follows from Jensen's inequality.
    Now set $\beta=\sqrt{2\sigma^2\log n}$.
\end{proof}

The following is a special (easy) case of the Borell-TIS inequality, see e.g. \cite[Theorem 2.1.1]{adler2009random}. Alternatively, this follows immediately from the Gaussian Lipschitz concentration inequality, Lemma~\ref{lem:Gaussian-Lip}:
\begin{lemma}[Borell-TIS]\label{lem:Borell-TIS}
    Let $\Z=(Z_1,\ldots,Z_n)$ be a Gaussian random vector with $\Expt[\Z]=\0$. Set $\sigma^2=\max_{1\le i \le n} \Expt Z_i^2$. Then for $t\ge 0$,
    \begin{align*}
        \Pr(\max_{1\le i \le n}Z_i \ge \Expt[\max_{1\le i \le n}Z_i] + t )\le e^{-\frac{t^2}{2\sigma^2}},\quad        
        \Pr(\max_{1\le i \le n}Z_i \le \Expt[\max_{1\le i \le n}Z_i] - t )\le e^{-\frac{t^2}{2\sigma^2}}\,.
    \end{align*}
\end{lemma}


\subsection{Results From Information Theory}



The following 
is Fano's inequality, see e.g. \cite[Theorem 5.2]{polyanskiy2014lecture}.

\begin{lemma}[Fano's inequality]
    \label{lem:Fano}
    Let $(\ell,Z,\hat{\ell})$ be random variables such that $\ell,\hat{\ell}\in [k]$, and the Markov chain
    \[
    \ell \longrightarrow Z \longrightarrow \hat{\ell}     
    \]
    holds. Denote $p_e=(\Pr(\ell\ne \hat{\ell}))$. 
    Then
    \begin{align*}
        \Ent(\ell|Z) \le h_b(p_e) + p_e\log(k-1)\,,
    \end{align*}
    where $h_b(p)=p\log\frac1p + (1-p)\log\frac{1}{1-p}$ is the binary entropy function.
\end{lemma}

Lastly is the celebrated I-MMSE relation of \cite[Theorem 2]{guo2005mutual}:

\begin{lemma}[I-MMSE]\label{lem:I-MMSE}
    Let $\X\in\RR^d$ be any random vector with finite second moments, $\Expt\|\X\|^2<\infty$. Let $\Z\sim \m{N}(\0,\Id)$ be indepdent of $\X$, and denote $\Y(s)=\sqrt{s}\X+\Z$. Then
    \begin{align*}
        \frac{d}{ds} \MI(\X;\Y(s)) = \frac{1}{2} \Expt\left[ \|\X-\Expt(\X|\Y(s))\|^2 \right] =: \mathrm{mmse}(s)\,.
    \end{align*}
\end{lemma}

\end{document}